\documentclass[11pt]{article}
\usepackage{amsmath,amsthm,amssymb,amscd,fancybox,ifthen,float,epsfig,subcaption}
\usepackage[all]{xy}
\usepackage{times}
\usepackage{color}
\usepackage{hyperref}
\floatstyle{plain}
\restylefloat{figure}

\usepackage{color}

\definecolor{gr}{rgb}   {0.,   0.8,   0. }
\definecolor{bl}{rgb}   {0.,   0.5,   1. }
\definecolor{mg}{rgb}   {0.7,  0.,    0.7}

\newcommand{\Bk}{\color{black}}
\newcommand{\Rd}{\color{black}}

\newcommand\rad{\operatorname{rad}}

\newcommand\loc{\operatorname{loc}}
\newcommand\comp{\operatorname{comp}}
\newcommand\Dom{\operatorname{Dom}}

\newcommand\dist{\operatorname{dist}}

\newcommand\tot{\operatorname{tot}}\newcommand\scat{\operatorname{sc}}
\newcommand\In{\operatorname{in}}
\newcommand\Int{\operatorname{int}}

\newcommand\out{\operatorname{out}}

\newcommand\cU{\mathcal U}

\newcommand\fw{\mathfrak w}
\newcommand\fW{\mathfrak W}
\newcommand\fA{\mathfrak A}
\newcommand\fB{\mathfrak B}
\newcommand\fC{\mathfrak C}
\newcommand\fE{\mathfrak E}

\newcommand\cB{\mathcal B}
\newcommand\cR{\mathcal R}

\newcommand\bkappa{\boldsymbol \kappa}

\newcommand\tbkappa{\tilde{\bkappa}}

\newcommand\tsigma{\widetilde \sigma}
\newcommand\ttau{\widetilde \tau}

\newcommand\tk{\tilde k}

\newcommand\tA{\widetilde A}
\newcommand\tB{\widetilde B}
\newcommand\tu{\widetilde u}

\newcommand\bx{\boldsymbol x}

\newcommand\cF{\mathcal F}
\newcommand\cE{\mathcal E}
\newcommand\cW{\mathcal W}

\newcommand\cC{\mathcal{C}}

\newcommand\cS{\mathcal{S}}

\newcommand\cD{\mathcal{D}}

\newcommand\tf{\tilde{f}}
\newcommand\tw{\tilde w_{0+}}

\renewcommand\Re{\operatorname{Re}}
\renewcommand\Im{\operatorname{Im}}

\newcommand\fD{\mathfrak D}

\newcommand\bbC{\mathbb C}

\newcommand\bbN{\mathbb N}

\newcommand\bbR{\mathbb R}
\newcommand\bbZ{\mathbb Z}

\newcommand\pa{\partial}

\newcommand\restrictedto{\upharpoonright}

\newcommand\supp{\operatorname{supp}}

\newcommand\Id{\operatorname{Id}}
\newcommand\Sgn{\operatorname{sgn}}

\newtheorem{theorem}{Theorem}
\newtheorem*{theorem*}{Theorem}
\newtheorem{proposition}{Proposition}
\newtheorem{corollary}{Corollary}
\newtheorem{lemma}{Lemma}

\theoremstyle{definition}
\newtheorem{definition}{Definition}

\theoremstyle{remark}
\newtheorem{remark}{Remark}

\begin{document}

\title{Solving the Scattering Problem for Open Wave-Guide Networks,
  I\\Fundamental Solutions and Integral Equations}

\author{Charles L. Epstein\footnote{ Center for Computational Mathematics, 
  Flatiron Institute, 162 Fifth Avenue, New York, NY 10010. E-mail:
  {cepstein@flatironinstitute.org}. }}
\date{November 6, 2025}
\maketitle

\begin{abstract} We introduce a  layer potential representation for the solution
  of the scattering problem defined by two dielectric channels, or open
  wave-guides, meeting along a straight-line, orthogonal to both channels, which
  is well adapted to numerical implementation. This is a simple example of a
  wave-guide network. The main observation is that the outgoing fundamental
  solution for the operator $\pa_{x_1}^2+\pa_{x_2}^2 +k_1^2+q(x_2),$ acting on
  functions defined in $\bbR^2,$ is easily constructed using the Fourier
  transform in the $x_1$-variable and the classical theory of ordinary
  differential equations. These fundamental solutions can then be used to
  represent the solution to the open wave-guide network scattering problem in half
  planes. The $H^2_{\loc}$-regularity of the solution to the scattering problem
  imposes transmission boundary conditions along the common boundary, which then
  leads to integral equations along the intersection of the half planes.  We
  show that, in appropriate Banach spaces, these integral equations are Fredholm
  equations of index zero, which are therefore generically solvable.  We also
  analyze the representation of the guided modes in our formulation.
\end{abstract}

\tableofcontents

\section{Introduction}
This paper is the first part of a three part series on scattering problems for
open, dielectric wave-guide \Rd networks.\footnote{\Rd In the Applied Math,
Engineering and Physics literature an ``open wave-guide'' usually refers to a
translationally invariant device.  We call these {\em bi-infinite wave-guides.}
The main point of our work is that we consider an assemblage of devices that are
asymptotically modeled by bi-infinite wave-guides, which we call a {\em
  wave-guide network.} \Bk} \Bk We refer to these papers as Parts I, II, and
III. Parts II, and III are the papers~\cite{EpWG2023_2}, and~\cite{EpMaSR2023}.

 In these papers we focus on the scalar case and work in the time harmonic
 setting, i.e. we consider solutions of the form $U(x,t)=e^{-i\omega t}u(x),$
 where $\omega>0.$ \Rd In the setting of electromagnetism,  a scalar
 model for an open wave-guide network is specified by perturbations of the
 background permittivity,  \Bk $\epsilon(x),$ and $u$ solves a
 Schr\"odinger equation
\begin{equation}\label{eqn1.700}
  (\Delta+\omega^2\epsilon(x))u(x)=0.
\end{equation}
The function $\epsilon(x)$ is strictly positive and is assumed to be a constant
outside of a compact set, union with tubular neighborhoods of a finite
collection of semi-infinite rays, $\{\ell_1,\dots,\ell_N\},$ see
Figure~\ref{fig3}. \Rd We call the directions of these rays the {\em wave-guide
  directions}. \Bk In each tubular neighborhood the permittivity is independent
of the distance along the ray, depending only on variables in the orthogonal
hyperplane. \Rd This is also a standard model for acoustic scattering with
$\epsilon(x)$ equal to the reciprocal of the squared sound speed.\Bk

In $d$-dimensions the solution to~\eqref{eqn1.700} is assumed to belong to
$H^2_{\loc}(\bbR^d),$ which means that it is continuous, and its gradient is
continuous (in a trace sense) across jumps in the permittivity. What makes the
scattering problem non-standard is that the set where $\epsilon(x)$ differs from
the constant background value is non-compact. The fact that this set is
localized near rays extending to infinity, makes this problem quite similar,
conceptually, to the classical quantum-mechanical $N$-body problem,
see~\cite{PPS1981,Isozaki94,Melrose94,Vasy2000}.

\Rd
\subsection{Scattering problems for wave-guide networks}
In a scattering problem, one typically specifies an ``incoming field,''
$u^{\In},$ which is a solution of~\eqref{eqn1.700} in the exterior of the
scatterer, i.e. where $\epsilon(x)$ assumes its constant, background value. The
incoming data is cut-off near the scatterer, with a smooth cut-off function, $\psi,$ 
introducing an error term
\begin{equation}\label{eqn2.702}
  w(x)=(\Delta+\omega^2\epsilon(x))[\psi(x)u^{\In}(x)].
\end{equation}
The scattered field is then defined as the ``outgoing'' solution, $u^{\scat}$ to
\begin{equation}
  (\Delta+\omega^2\epsilon(x))u^{\scat}(x)=w(x),
\end{equation}
so that $u^{\tot}=u^{\In}-u^{\scat}$ satisfies
\begin{equation}\label{eqn4.702}
   (\Delta+\omega^2\epsilon(x))u^{\tot}(x)=0\text{ for all }x\in\bbR^d.
\end{equation}

If the scatterer lies in a bounded region, and $\omega^2\epsilon(x)=k_1^2,$ for
large $|x|,$ then the correct notion of outgoing solution is that provided by the
Sommerfeld radiation condition
\begin{equation}
  |(\pa_r-ik_1)u^{\scat}(r\eta)|\leq Cr^{-(\frac{d-1}{2}+\delta)},
\end{equation}
for some fixed $\delta>0$, and all $\eta\in S^{d-1}.$ This solution can be found
using the limiting absorption principle: The operator
$\Delta+\omega^2\epsilon(x)$ is an unbounded self adjoint operator acting on
$H^2(\bbR^d).$ Hence  the resolvent operator
$$R(\delta)=(\Delta+\omega^2\epsilon(x)+i\delta)^{-1}:L^2(\bbR^d)\to H^2(\bbR^d)$$
is a bounded operator for $\delta\neq 0.$ For $\omega^2\epsilon(x)-k_1^2$
compactly supported, it is a classical result,
see~\cite{Agmon}, that the limits
\begin{equation}
  \lim_{\delta\to 0^{\pm}}R(\delta)=R(\pm 0)
\end{equation}
are well defined as bounded maps
$$R(\pm 0):r^{-(\frac 12+\nu)}L^2(\bbR^d)\to r^{(\frac  12+\nu)}L^2(\bbR^d),$$
for any $\nu>0.$ Provided that $w\in r^{-(\frac 12+\nu)}L^2(\bbR^d),$ for some
$\nu>0,$ we can define
\begin{equation}
  u^{\scat}=R(+0)w,
\end{equation}
and this gives the desired outgoing solution, which represents the scattered field.

If the scatterer is not bounded, as in the case of wave-guide networks, then several
issues arise. The most obvious one is that the analogue of the Sommerfeld
radiation condition, which specifies the outgoing solution, is not to be found
in the wave-guide literature. While the limiting absorption principle has been
established in some cases, see~\cite{DeBevierePravica1,DeBevierePravica2}, the
error terms, $w,$ that arise in this context often do not belong to $r^{-(\frac
  12+\nu)}L^2(\bbR^d)$ for any $\nu>0.$ This, in turn, is connected to the difficulty in
specifying the incoming field. We return to the problem of admissible data in
Section~\ref{sec_adm_data} of Part I, and Section 5.3 of Part III.

The limiting absorption principle and  physically motivated radiation
conditions, which imply uniqueness, for wave-guide networks, were obtained in a
paper of Andras Vasy on the 3-body problem in quantum mechanics,
see~\cite{Vasy2000}. Similar radiation conditions were earlier obtained by
H. Isozaki, though his hypotheses exclude the class of potentials that arise in
our setting, see~\cite{Isozaki94}. The radiation condition away from the
wave-guide directions is simply the classical Sommerfeld condition, but the
radiation condition in the wave-guide directions is a good deal more complicated
to state. Radiation conditions for wave-guide networks are the main focus of Part III,
see~\cite{EpMaSR2023}.

Ultimately we will solve a scattering problem like that described in
equations~\eqref{eqn2.702}--\eqref{eqn4.702}. We do not use the limiting
absorption principle to solve this problem, and considerable care is required to
define the incoming field. In this paper we reformulate the scattering problem
for the simplest type of {\em non-trivial} wave-guide network in $\bbR^2,$
illustrated in Fig.~\ref{fig0}, as a transmission problem across an artificial
boundary (the $x_2$-axis) separating the two ends of the wave-guide network, see
equations~\eqref{eqn10.702}--\eqref{eqn14.702}. The solution of the transmission
problem is then reduced to solving a system of integral equations on this
artificial boundary.

In Sections~\ref{sec4} and~\ref{sec3.55} we construct
fundamental solutions for the translation invariant operators that define the
two ends of this network. Using the kernel functions for these operators we
obtain integral equations along the $x_2$-axis,~\eqref{eqn143.35}, that enforce
the transmission boundary conditions. In Section~\ref{sec8} we introduce these
integral equations and function spaces in which they are shown to be Fredholm of
index 0. Hence the integral equations are solvable for all data in these spaces
if it can be shown that the null-space is trivial. To prove such a result
requires a uniqueness theorem for the underlying PDE: i.e.~that an outgoing
solution to~\eqref{eqn1.700} is zero.  As noted above this in turn requires an
outgoing radiation condition, which is only provided in Part III, along with the
desired uniqueness theorem for the integral equations.

In order to show that solutions to our integral equations are outgoing, in the
sense defined in Part III, we need to obtain very precise asymptotics for the
fields defined by these solutions. The nature of the asymptotics of
$u^{\scat}(r\eta)$ depends on whether or not $\eta$ is parallel to a wave-guide
direction, which in our case are the directions $\{(\pm 1,0)\}.$ To further
complicate matters, we need to prove asymptotics away from these directions,
that are uniformly correct as $\eta$ approaches $(\pm 1, 0).$ Obtaining these
asymptotic expansions is the focus of Part II. Using these results, in Part III
we show that the solutions to the scattering problem obtained using our methods
agree with those obtained via the limiting absorption principle.

In the rest of this introduction we give a detailed description of the contents
of Part I and a brief description of the contents of Parts II and III.

\Bk
\subsection{The simplest non-trivial case}
The simplest type of a 2-dimensional wave-guide network arises in the so-called
layered medium problem where $\epsilon(x)$ is a function of a single variable,
say $x_2.$ As noted earlier, we call these bi-infinite wave-guides. \Rd For a
wave-guide it is usually assumed that $\epsilon(x_2)=k_1,$ a constant, outside a bounded
interval. \Bk Under these
conditions this problem is very effectively analyzed by applying the Fourier
transform in the $x_1$-variable. A summary of this analysis is presented in
Appendix~\ref{sec1}. An excellent reference for this problem and its
perturbations is~\cite{Christiansen}.  In this and the following part we analyze
the simplest problem for which such an approach is no longer possible, which is
obtained by having two such wave-guides meeting along a common perpendicular
line.

\begin{figure}
  \centering \includegraphics[width= 10cm]{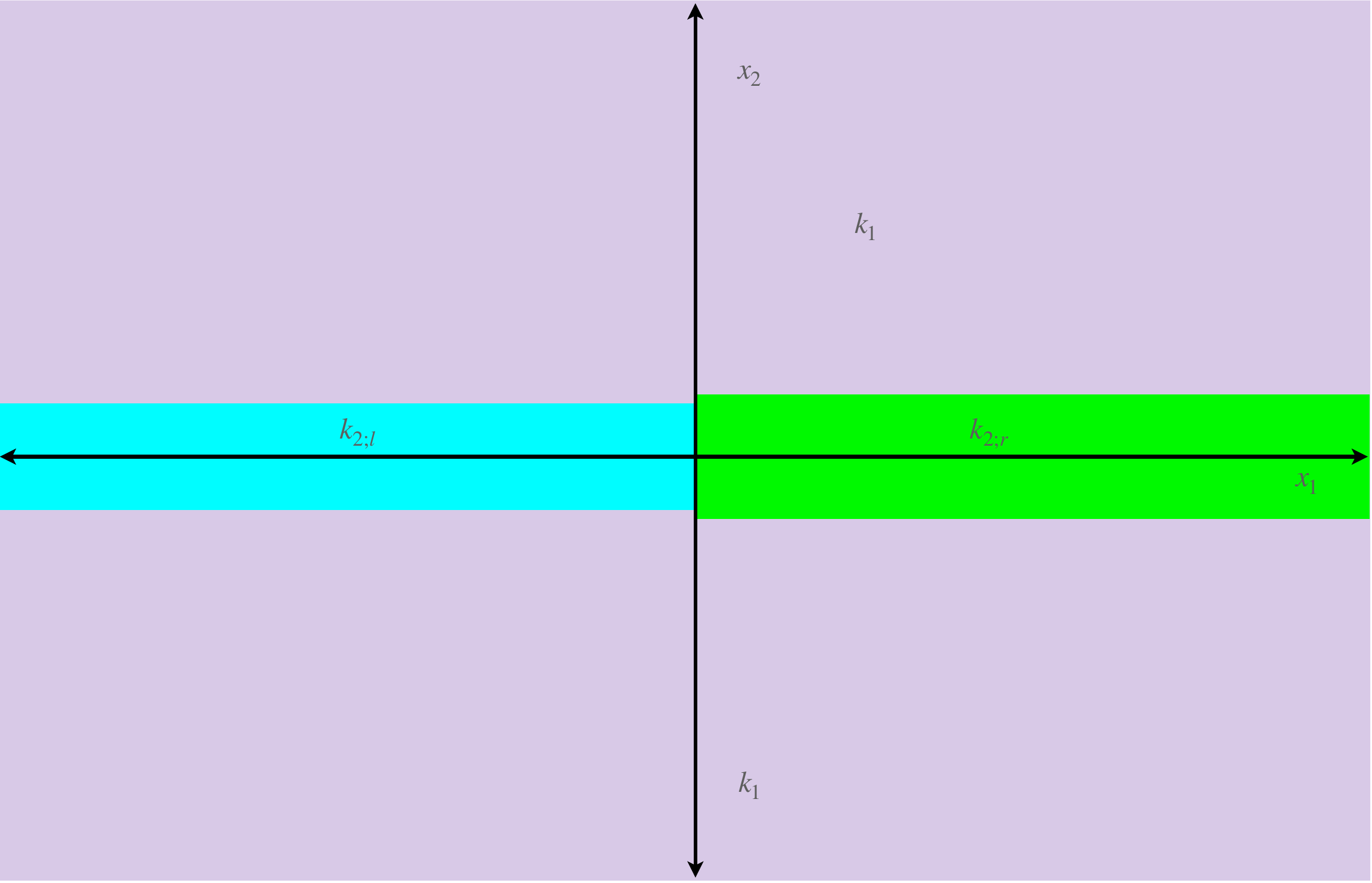}
    \caption{Two dielectric channels meeting along a straight interface. The
      $x_3$-axis is orthogonal to the plane of the image.}  
   \label{fig0}
\end{figure}

These problems have already received a lot of attention in the work of many
authors, for example, see~\cite{DeSantoMartin1997, MagnaniniSantosa, Marcuse,
  Nosich1994,JolyPoirier,BonnetBendhia_etal,
  BonnetTillequin2001,KNH_SIAM_2005,ChandlerMonkThomas2007,ChandlerZhang}.
For the case of two semi-infinite, dielectric channels meeting along a common
perpendicular line, the precise formulation is in terms of a
pair operators:
\begin{equation}
  \Delta +k_1^2+q_{l,r}(x_2),
\end{equation}
see Figure~\ref{fig0}.
Here, and throughout the paper, $l$ refers to $\{x_1<0\},$ and  $r$ refers to
$\{x_1>0\}.$ Though our method of solution applies  to
piecewise continuous potentials with bounded support, for simplicity, we usually assume that
\begin{equation}\label{eqn2.32}
 q_{l,r}(x_2)=(k^2_{2;l,r}-k^2_1)\chi_{[d^-_{l,r},d^+_{l,r}]}(x_2),
\end{equation}
with $k_1=\omega\sqrt{\epsilon_1},$ $k_{2;l,r}=\omega\sqrt{\epsilon_{2;l,r}}.$
We assume the permittivities are positive, real numbers.

\Rd A large part of the interest of wave-guides stems from the fact that they
support {\em wave-guide} modes. In the bi-infinite case these are traveling wave
solutions to the equation $(\Delta+k_1^2+q(x_2))u(x_2)=0$ of the form
\begin{equation}
  u(x_1,x_2)=e^{i\xi x_1}v(x_2),\text{ with }\xi\in\bbR \text{ where }(\pa_{x_2}^2+k_1^2-\xi^2+q(x_2))v(x_2)=0,
\end{equation}
and $v\in L^2(\bbR).$ If $\xi>0,$ then the wave guide mode is right-ward moving;
if $\xi<0,$ it is left-ward moving. For such solutions to exist the potential
has to be positive on part of its support.  For the potentials
in~\eqref{eqn2.32} to support wave-guide modes the wave numbers , $k_{2;l,r},$
within the channels must be larger than $k_1,$ and we usually assume that this
is the case. Our method applies if a wave number in a channel is smaller than $k_1,$
the analysis is simpler and we have not included it.  \Bk

The problem we consider is similar to that considered
in~\cite{BonnetBendhia_etal}, but rather different from that considered
in~\cite{ChandlerZhang}. Chandler-Wilde and Zhang allow more general
potentials, but make assumptions that preclude the existence of wave-guide
modes. They also focus on incoming data similar to plane waves, whereas we
consider incoming fields that decay as $|x_2|\to\infty,$ which includes
wave-guide modes as well as point sources and smooth wave packets formed from
plane waves, see Section~\ref{sec_adm_data}.

We use a symmetric formulation of the transmission problem, which
is suggested by the symmetry of the operator itself.
Data for the scattering problem is defined by incoming solutions
$u^{\In}_{l,r},$ which satisfy the equations
\begin{equation}\label{eqn10.702}
  (\Delta+k_1^2+q_{l,r}(x_2))u^{\In}_{l,r}(x)=0.
\end{equation}
We then look for `outgoing' solutions $u^l,u^r,$ to the equations
\begin{equation}\label{eqn5}
  (\Delta+k^2_1+q_{l,r})u^{l,r}=0\text{ where }x_1< 0, \text{
    resp. } x_1>0,
\end{equation}
with $u^l\in H^2_{\loc}( (-\infty,0]\times \bbR)$ and $u^r\in
  H^2_{\loc}([0,\infty) \times\bbR),$ which satisfy the transmission
    boundary conditions
\begin{equation}\label{eqn5.30}
  \begin{split}
    &g(x_2)\overset{d}{=}
    u^r(0^+,x_2)-u^l(0^-,x_2)=u^{\In}_l(0,x_2)-u^{\In}_r(0,x_2),\\ &h(x_2)
    \overset{d}{=}
    \pa_{x_1}u^r(0^+,x_2)-\pa_{x_1}u^l(0^-,x_2)=\pa_{x_1}u^{\In}(0,x_2)-\pa_{x_1}u^{\In}_r(0,x_2).
  \end{split}
\end{equation}\Rd
For the applications of most immediate interest,  the incoming field is
a sum of \emph{wave-guide modes}. A single incoming wave-guide mode for
$\Delta+k_1^2+q_l(x_2),$ is given by $u^{\In}_l=e^{i\xi_0 x_1}v_0(x_2),$
where $\xi_0>0$ and $v_0(x_2)$ is an $H^2(\bbR)$--solution to
\begin{equation}\label{eqn10.210}
  [\pa_{x_2}^2-\xi_0^2+k_1^2+q_{l}(x_2)]v_0(x_2)=0.
\end{equation}
Such solutions are exponentially decaying outside the channel, which is the
$\supp q_l.$ For such a single mode, with $u^{\In}_r=0,$ the jump data is given
by
\begin{equation}\label{eqn15.570}
  g(x_2)=v_0(x_2),\quad h(x_2)=i\xi_0 v_0(x_2).
\end{equation}

The solution to the transmission problem is constructed using the outgoing
fundamental solutions defined by the operators $(\Delta+k_1^2+q_{l,r}(x_2)).$ As
we prove in Section~\ref{sec8}, the transmission problem has a solution for data $(g,h)$ much
more general than that in~\eqref{eqn15.570}.
Because the domain of integration used to represent the fields, $u^{l,r},$
extends to infinity, it is by no means obvious that these fields will be
outgoing in the sense given in Part III. In fact, this turns out to depend on
the data $(g(x_2),h(x_2)),$ which must itself be outgoing in a certain
sense. This is considered in detail in Section 2 of Part II.\Bk

The regularity and boundary conditions then  imply that  the field
\begin{equation}
  u^{\tot}(x_1,x_2)=
  \begin{cases} &u^l(x_1,x_2)+u^{\In}_l(x_1,x_2)\text{ for }x_1\leq 0,\\
  &u^r(x_1,x_2)+u^{\In}_r(x_1,x_2)\text{ for }x_1>0,
  \end{cases}
\end{equation}
is a weak solution to
\begin{multline}\label{eqn14.702}
  (\Delta+q(x_1,x_2)+k_1^2)u^{\tot}=0,\\\text{ where }q(x_1,x_2)=
  q_l(x_2)\chi_{(-\infty,0]}(x_1)+q_r(x_2)\chi_{(0,\infty)}(x_1),
\end{multline}
which belongs to $H^2_{\loc}(\bbR^2).$ In Part III we show that, for admissible
incoming data, this solution agrees with the limiting absorption solution  and
therefore ours is an explicit method to solve the scattering problem.

Our solutions are outgoing in a sense rather different from that in earlier
papers on {\em this} problem. The outgoing condition we establish in Part III requires the scattered
fields to satisfy the classical Sommerfeld radiation conditions outside
the channels, with appropriate outgoing conditions within the channels. Among
other things this implies that along the line $\{x_1=0\}$ the solutions must
satisfy an estimate like
\begin{equation}
  (\pa_{x_2}\mp ik_1)u^{l,r}(0,x_2)=O((\pm x_2)^{-\frac 32}).
\end{equation}
The transmission boundary condition then implies that
\begin{equation}\label{eqn14.300}
    (\pa_{x_2}\mp ik_1)g(x_2)=O((\pm x_2)^{-\frac 32}),
\end{equation}
which indicates the sorts of conditions that are required on the data for the
existence of an outgoing solution to the transmission problem. 

As noted above, in Part III we introduce radiation conditions for open
wave-guide networks that imply uniqueness.  Under a different definition of
outgoing, the existence and uniqueness of solutions to this problem is
established in~\cite{BonnetBendhia_etal}. We give a different formulation of the
solution that we believe is better suited to numerical approximation. \Rd As shown
in~\cite{EpGo2024, EGHQR2025} our approach provides the basis for an efficient
numerical method to accurately solve this class of scattering problems. \Bk That said,
we have not considered the relationship between our solution and that found
in~\cite{BonnetBendhia_etal}, \Rd  though they are likely to be the same.\Bk

For data of this type we show, in Part II, that the solution found using our
method satisfies the Sommerfeld radiation conditions outside of the channel.
Within the channels we also have contributions from the guided modes, which do
not decay, but are outgoing in an appropriate sense: on the left a guided mode
is outgoing if it is proportional to $v(x_2)e^{-i\xi_0 x_1},$ with $\xi_0>0,$
and outgoing to the right if is proportional to $v(x_2)e^{i\xi_0 x_1},$ with
$\xi_0>0. $ \Rd A radiation condition similar to this for bi-infinite wave-guides is
given in~\cite{CiraoloMagnanini2008}.\Bk

Section~\ref{sec4} introduces our
approach, which uses the explicit fundamental solutions for the perturbed
operators, $\Delta+k_1^2+q_{l,r}(x_2).$ In Sections~\ref{sec3.55} and~\ref{sec6}
we provide details for the construction of, and estimates for the kernels
introduced in Section~\ref{sec4}.  In Section~\ref{sec8} we examine the integral
equations, \eqref{eqn143.35}, that we need to solve and prove mapping results on
the operators defined by these kernels in Banach spaces of continuous functions
with specified rates of decay: For $0\leq\alpha,$ let $\cC_{\alpha}(\bbR)$
denote the subspace of functions $f\in\cC^0(\bbR)$ with
\begin{equation}\label{eqn11.97}
  |f|_{\alpha}=\sup\{(1+|x|)^{\alpha}|f(x)|:\: x\in\bbR\}<\infty.
\end{equation}
We show that the integral equations obtained in Section~\ref{sec8} make sense in
the spaces $\cC_{\alpha}(\bbR)\oplus\cC_{\alpha+\frac 12}(\bbR)$ for
$0<\alpha<\frac 12.$ We next show that, on these spaces, these equations are
Fredholm of index zero, which are therefore generically solvable. \Rd We have not
shown that they are of the form $\Id$+compact.\Bk

In Section~\ref{sec_adm_data} we consider different kinds of physically
interesting data that belong to $\cC_{\alpha}(\bbR)\oplus\cC_{\alpha+\frac
  12}(\bbR)$ for an $0<\alpha<\frac 12.$ In Section~\ref{sec8.91} we derive
equations that allow for the approximate determination of the scattering
relation from incoming wave-guide modes to transmitted and reflected modes,
without solving the full system of equations. In this section we show that the
projections of the solutions onto the wave-guide modes are determined by the
projections of the source terms onto these modes.  Several appendices give
proofs of theorems and ancillary results used  in the main body of the paper.

\subsection{Material in Parts II and III}
We briefly summarize the contents of Parts II and III.

\begin{enumerate}

\item In Part II of the series we obtain precise asymptotics for our solutions,
  which are used to show that they satisfy outgoing radiation conditions
  introduced in Part III. In the complement of the channels this entails showing
  that the solutions satisfy the classical Sommerfeld radiation conditions,
  uniformly `down to the horizon:' If the channel is $\bbR\times [-d,0],$ then
  we need to show that the solution, $u,$ satisfies
 \begin{equation}
   |(\pa_r-ik_1)u(r(\cos\theta,\sin\theta))|\leq\frac{C}{r^{\frac 32}},
 \end{equation}
 uniformly as $\theta\to 0,\pi,$ with a similar estimate for $x_2<-d.$

 As part of this analysis we characterize sources, $\tau, \sigma,$ so that the
 single and double layer potentials,
 \begin{equation}
   \begin{split}
   \cS_k\tau(x)&=\int_{-\infty}^{\infty}g_k(x;0,y_2)\tau(y_2)dy_2,\\
   \cD_k\sigma(x)&=\int_{-\infty}^{\infty}\pa_{y_1}g_k(x;0,y_2)\sigma(y_2)dy_2,
   \end{split}
 \end{equation}
 integrated over the  line $\{x_1=0\},$
 uniformly satisfy Sommerfeld radiation conditions in the half planes, $\{x_1\geq 0\},
 \{x_1\leq 0\}.$ These estimates hold if $\sigma, \tau$ themselves have
 asymptotic expansions
 \begin{equation}
   \begin{split}
     \sigma(y_2)&\sim
     \frac{e^{i|y_2|}}{\sqrt{|y_2|}}\sum_{j=0}^{N}\frac{a_j^{\pm}}{|y_2|^j},\\
         \tau(y_2)&\sim \frac{e^{i|y_2|}}{|y_2|^{\frac 32}}\sum_{j=0}^{N}\frac{b_j^{\pm}}{|y_2|^j},
   \end{split}
 \end{equation}
 for a sufficiently large $N.$ The proofs of these estimates use stationary
 phase, but require  non-conventional contour deformations to obtain estimates
 uniform down to the horizon. 
  
\item In Part III of this series, which is joint with Rafe Mazzeo, we formulate
  physically motivated radiation conditions for the general $d$-dimensional,
  scalar open wave-guide network scattering problem. These conditions imply uniqueness
  for the solution of the PDE, and are satisfied
  by the limiting absorption solution. Indeed these results are implicit
  in~\cite{Vasy2000}, which treats scattering theory for the quantum mechanical
  3-body problem. We  show that the solutions found using our method
  satisfy these conditions, and therefore agree with the limiting absorption
  solutions.  This in turn implies that the integral equations found in Part I
  have a trivial null-space, and are therefore always solvable for data in
  appropriate Banach spaces. We finally show that the channel-to-channel
  scattering coefficients are well defined for this class of problems.
\end{enumerate}

{\small
  
\centerline{\sc Acknowledgments} I would like to thank Leslie Greengard for
suggesting this problem and for many interesting conversations along the way. I
would also like to thank Manas Rachh, Shidong Jiang, Felipe Vico, and Alex Barnett for many
helpful discussions of this material and pointers to the literature on this
problem. I am very grateful to Manas for carefully reading this manuscript
and providing very useful comments. I would like to thank David Jerison, Rafe
Mazzeo and Andras Vasy for useful discussions about these issues. I am also
very grateful for the support of the Flatiron Institute of the Simons
Foundation, and for the support of Stanford University through the Bergman
Visiting Scholarship.  I would finally like the thank the referees for Parts I,
II and III whose careful reading and thoughtful comments have lead to very substantial
improvements in all 3 papers. }

\section{A Layer Potential Approach}\label{sec4}
In this section we reformulate the transmission problem, introduced above, in
terms of integral equations on $\{x_1=0\}.$ The starting point for our approach
is the formula for the solution of the classical transmission problem: find a
function $u$ that solves
$$(\Delta+k^2)u=0,\text{ in }\bbR^2\setminus \{x_1=0\},$$
and the transmission boundary condition
\begin{equation}
  \begin{split}
    u(0^+,x_2)-u(0^-,x_2)&=g(x_2)\\
     \pa_{x_1}u(0^+,x_2)-\pa_{x_1}u(0^-,x_2)&=h(x_2).
  \end{split}
\end{equation}
The `outgoing' solution to this problem is given by
\begin{equation}\label{eqn28.35}
  u(x)=-\cD_k g+\cS_k h,
\end{equation}
where the single and double layers are given by
\begin{equation}
  \cS_kf(x)=\int\limits_{\{y_1=0\}}g_k(x-y)f(y_2)dy_2,\text{ and }
  \cD_kf(x)=\int\limits_{\{y_1=0\}}\pa_{y_1}g_k(x-y)f(y_2)dy_2,
\end{equation}
with
$$g_k(x-y)=\frac{i}{4}H^{(1)}_0(k|x-y|),$$
the outgoing fundamental solution to $\Delta+k^2.$
That this gives a solution follows from the classical jump formul{\ae}
for layer potentials, see~\cite{ColtonKress}.  Whether or not $u$ satisfies the
necessary radiation condition for the uniqueness conditions in~\cite{Odeh1963} to
apply depends on the data, $(g,h).$ This question is discussed in Part II.

Our method for solving the scattering problem for 2 wave-guides~\eqref{eqn5},
and~\eqref{eqn5.30}, uses this general approach.  Let $\fE^{l,r}(x;y)$ denote
the outgoing fundamental solutions for the operators
$\Delta+k_1^2+q_{l,r}(x_2),$ acting on the \emph{whole plane}. The kernels of
these operators take a rather special form:
\begin{equation}
  \fE^{l,r}(x;y)=\frac{i}{4}H^{(1)}_0(k_1|x-y|)+w^{l,r}(x;y),
\end{equation}
where $w^{l,r}$ satisfies the equation
\begin{equation}\label{eqn32.10}
  (\Delta_x+k_1^2+q_{l,r}(x_2))w^{l,r}(x;y)=-\frac{i}{4}q_{l,r}(x_2)H^{(1)}_0(k_1|x-y|).
\end{equation}
\begin{remark}
  Similar ideas for constructing the outgoing fundamental solution appear
  in~\cite{Nosich1994}. The arguments in that paper are of a more physical
  character.
\end{remark}

 The right hand sides of~\eqref{eqn32.10} are compactly supported in the
 $x_2$-variable. As we shall see, these equations can be solved quite
 explicitly by taking the Fourier transform in the $x_1$-variable, and
 solving frequency-by-frequency. The correction terms, $w^{l,r}(x;y),$
 are smooth away from the diagonal, and 2 orders smoother along the diagonal, as
 distributions, than $H^{(1)}_0(k|x-y|).$ Suppose that $\supp q_{l,r}\subset
 [-d,d],$ then, where $x_1=y_1=0,$ the singularities of $w^{l,r}(0,x_2;0,y_2)$
 are contained within the compact set $B_d\cap\{x_2=y_2\},$ where, for $d>0,$
  \begin{equation}
     B_d\overset{d}{=}[-d,d]\times[-d,d].
  \end{equation}

Once the fundamental solutions are constructed we can express the right and left
portions of the solution as
\begin{equation}\label{eqn36.50}
  u^{l,r}=-\cE^{l,r\,'}\sigma+\cE^{l,r}\tau,
\end{equation}
where
\begin{equation}\label{eqn36.34}
  \begin{split}
  &\cE^{l,r}
  f(x)=\int_{-\infty}^{\infty}\fE^{l,r}(x;0,y_2)f(y_2)dy_2=\cS_{k_1}f+\cW^{l,r}f(x)\\
  &\cE^{l,r\, '}
  f(x)=\int_{-\infty}^{\infty}[\pa_{y_1}\fE^{l,r}](x;0,y_2)f(y_2)dy_2=\cD_{k_1}f+\cW^{l,r\,'}f(x),
  \end{split}
\end{equation}
with
\begin{equation}
  \begin{split}
 \cW^{l,r}f(x)=& \int_{-\infty}^{\infty}w^{l,r}(x;0,y_2)f(y_2)dy_2,  \\
  \cW^{l,r\,'}f(x)=&\int_{-\infty}^{\infty}[\pa_{y_1}w^{l,r}](x;0,y_2)f(y_2)dy_2.
  \end{split}
\end{equation}
Using this representation we apply~\eqref{eqn5.30} to derive a system of equations for $(\sigma,\tau)$
in~\eqref{eqn36.50}. These equations are somewhat better behaved than usual in so far as
the singularities of $w^{l,r}$ behave like $|x-y|^2\log|x-y|.$ The restrictions
to the boundary, $x_1=y_1=0,$ are given by
\begin{equation}\label{eqn37.24}
  \begin{split}
    u^r(0,x_2)-u^l(0,x_2)&=\sigma(x_2)+\\ &\int_{-\infty}^{\infty}[\pa_{y_1}w^l-\pa_{y_1}w^r]{(0,x_2;0;y_2)}\sigma(y_2)dy_2+\\
    &\int_{-\infty}^{\infty}[w^r-w^l]{(0,x_2;0;y_2)}\tau(y_2)dy_2=g(x_2),\\
    \pa_{x_1}u^r(0,x_2)-\pa_{x_1}u^l(0,x_2)&=\tau(x_2)+\\ &\int_{-\infty}^{\infty}[\pa_{x_1}w^r-\pa_{x_1}w^l]{(0,x_2;0;y_2)}\tau(y_2)dy_2+\\
    &\int_{-\infty}^{\infty}[\pa^2_{x_1y_1}w^l-\pa^2_{x_1y_1}w^r]{(0,x_2;0;y_2)}\sigma(y_2)dy_2=h(x_2).
  \end{split}
\end{equation}
Only the $H^{(1)}_0$-terms have jumps across $\{x_1=0\};$ as we show in Section~\ref{sec4.1.205}
\begin{equation}
  \pa_{y_1}w^{l,r}(0,x_2;0,y_2)=\pa_{x_1}w^{l,r}(0,x_2;0,y_2)=0,
\end{equation}
so these equations take the very simple form
\begin{equation}
  \left(
  \begin{matrix}
    \Id & D\\C&\Id
  \end{matrix}\right) \left(
  \begin{matrix}
    \sigma\\\tau
  \end{matrix}\right)=  \left(\begin{matrix}
    g\\ h
  \end{matrix}\right).
\end{equation}
The behavior of these equations hinges on the analytic properties of the functions
$w^{l,r}$ along the plane where $x_1=y_1=0,$ which we analyze in the following 2
sections.  It is an interesting feature of this approach, via the fundamental
solutions of $\Delta+k_1^2+q_{l,r}(x_2),$ that if $q_l=q_r,$ then these
equations reduce to
\begin{equation}
   \begin{split}
    u^r(0,x_2)-u^l(0,x_2)&=\sigma(x_2)=g(x_2)\\
    \pa_{x_1}u^r(0,x_2)-\pa_{x_1}u^l(0,x_2)&=\tau(x_2)=h(x_2),
  \end{split}
\end{equation}
exactly as in~\eqref{eqn28.35}.

\section{The Structure of the Perturbed Green's Function}\label{sec3.55}
In Sections~\ref{sec3.55}--\ref{sec6} we simplify notation by dropping the $l,r$ sub- and
super-scripts.  We use the resolvent kernel for $\Delta+k_1^2+q(x_2)+i\delta$ to find the
kernel functions needed to solve the transmission problem above; these functions
are found by solving the equation
      \begin{equation}
        (\Delta_x+k_1^2+q(x_2)+i\delta)w_{\delta}(x;y)=-\frac{i}{4}q(x_2)
        H^{(1)}_0(\sqrt{k^2_1+i\delta}|x-y|),\text{
          for }\delta>0.
      \end{equation}
If $\delta>0,$ then, for fixed $y,$ the right hand side belongs to
$L^2(\bbR^2).$ We then let $\delta\to 0^+,$ and denote this `limiting absorption
solution' by $w_{0+}(x;y).$ This insures that we get the desired
outgoing fundamental solution. To solve the limiting equation we simply take
the Fourier transform in the $x_1$-variable, and use the fact that, as
$\delta\to 0^+,$ we get
\begin{equation}\label{eqn26.51}
  \cF_{x_1}[(i/4)H^{(1)}_0(k_1|x-y|)](\xi)=\frac{ie^{i|x_2-y_2|\sqrt{k_1^2-\xi^2}}e^{-iy_1\xi}}{2\sqrt{k_1^2-\xi^2}};
\end{equation}
in general
$$\sqrt{k_1^2-\xi^2}=i\sqrt{\xi^2-k_1^2},\text{ if }|\xi|>k_1.$$
Let $\tw(\xi,x_2;y)$ denote the Fourier transform of $w_{0^+}$ in the
$x_1$-variable. For $\xi\in\bbR,$ it is the outgoing solution to the ordinary
differential equation
\begin{equation}
  L_{\xi}\tw=(\pa_{x_2}^2-\xi^2+k_1^2+q(x_2))\tw=-q(x_2)\frac{ie^{i|x_2-y_2|\sqrt{k_1^2-\xi^2}}e^{-iy_1\xi}}{2\sqrt{k_1^2-\xi^2}}.
\end{equation}
The spectral theory of $\Delta+q(x_2)$ is reviewed in Appendix~\ref{sec1}.  Let
$R_{\xi,0+}(x_2,z_2)$ denote the `outgoing' resolvent kernels for the
1-dimensional operators $L_{\xi}.$ They are the limits of the resolvent kernels
for $(L_{\xi}+i\delta)^{-1}$ as $\delta\to 0^+,$ constructed out of the basic
solutions, $\tu_{\pm}(\xi,0+;x_2),$ of $L_{\xi}$ and their Wronskian, $W(\xi),$
see~\eqref{eqn209.81}, \eqref{eqn7}, \eqref{eqn18.02}.  Using this kernel we can
write:
\begin{equation}\label{eqn28.51}
  \tw(\xi,x_2;y)=-\frac{ie^{-iy_1\xi}}{2\sqrt{k_1^2-\xi^2}}\int_{-d}^{d}
  R_{\xi,0+}(x_2,z_2)q(z_2)e^{i|z_2-y_2|\sqrt{k_1^2-\xi^2}}dz_2.
\end{equation}

The integral in~\eqref{eqn28.51} extends over the
\emph{finite} interval $ [-d,d]\supset \supp q,$ which is extremely useful from the
perspective of numerical solutions.  Away from the diagonal in $B_d,$ it also decays exponentially as
$|\xi|\to\infty.$ We reconstruct $w_{0+}$ as a
contribution from the continuous spectrum of $L_{\xi},$ and a
contribution from the wave-guide modes. The continuous spectrum
contributes
\begin{equation}\label{eqn81.21} 
  w^c_{0+}(x;y)=\frac{1}{2\pi}\int_{\Gamma_{\nu}^{+}}\tw(\xi,x_2;y)e^{ix_1\xi}d\xi\text{
      for }x_1>0.
\end{equation}
The integral in~\eqref{eqn81.21} is over the contour
$\Gamma_{\nu}^{+},$  which is defined below, see
Figure~\ref{fig2}. In order to be able to deform the contour of integration and
use~\eqref{eqn81.21} to represent $w_{0+},$ it is necessary to assume that $\pm
k_1$ are not roots of Wronskian, $W(\xi),$ of $L_{\xi},$
 see~\eqref{eqn48.51} and~\eqref{eqn55.52}--\eqref{eqn225.81}. If
\begin{equation}\label{eqn45.32}
  q(x_2)=(k_2^2-k_1^2)\chi_{[-d,d]}(x_2),
\end{equation}
then this amounts to the requirement that
\begin{equation}\label{eqn29.96}
  2d\sqrt{k_2^2-k_1^2}\neq n\pi\text{ for }n\in\bbN.
\end{equation}
The details of this construction are in Proposition~\ref{prop5.230} in Appendix~\ref{sec1}.

In this case, as shown in Theorem~\ref{thm0}, the roots of the Wronskian, $\{\pm \xi_n:\:
n=1,\dots N\},$ lie in $(-k_2,-k_1)\cup (k_1,k_2).$ The contour
    $\Gamma_{\nu}^+$ is defined by replacing the intervals
    $\{[\pm\xi_n-\nu,\pm\xi_n+\nu]:\: n=1,\dots, N\}$ in $\bbR$ with
    the semi-circles in the upper half plane
    \begin{equation}
    \{\pm \xi_n+\nu e^{i\theta}:\:\theta\in
           [0,\pi]\},
    \end{equation}
    with a clockwise orientation.  We assume that $\nu>0$ is small
    enough so these semi-circles are disjoint and intersect $\bbR$ within
    $(-k_2,-k_1)\cup (k_1,k_2).$ The contour $\Gamma_{\nu}^-$ is
    obtained by reflecting $\Gamma_{\nu}^+$ in the real axis.
\begin{figure}
  \centering
    \includegraphics[width= 10cm]{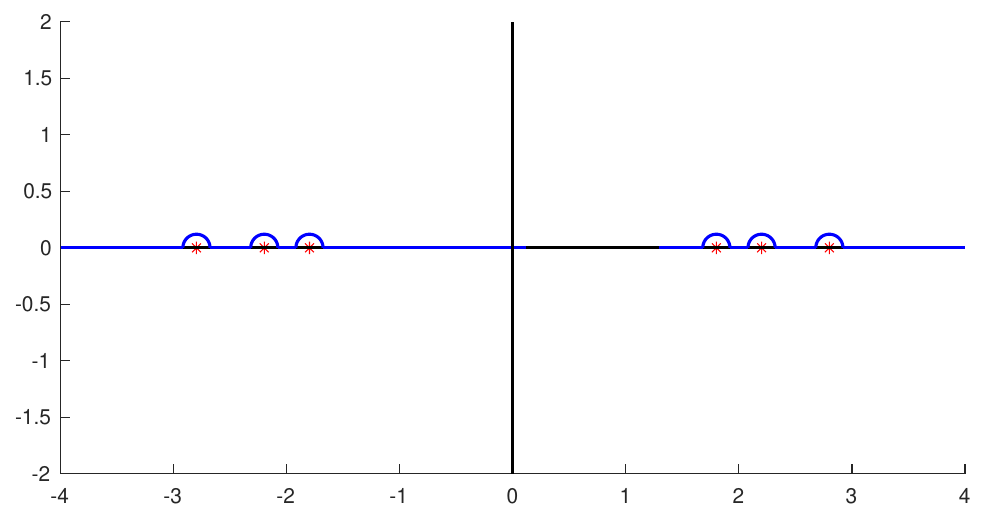}
    \caption{The contour $\Gamma^+_{\nu}$ shown in blue. The roots of
      Wronskian $\{\pm\xi_n\}$ are shown as red asterisks.}
   \label{fig2}
\end{figure}

To $w^c_{0+}(x;y)$ we add a contribution from the wave-guide modes:
\begin{equation}\label{eqn43.34.1}
  w^g_{0+}(x;y)=\sum_{n=1}^Nv_n(x_2)a_n(y_2)e^{i\xi_n(x_1-y_1)}
    \text{
      for }x_1>0,
\end{equation}
which is $i$ times the sum of the residues of $\tw(\xi,x_2;y)e^{ix_1\xi}$ at the
$\{\xi_n\}.$ The $\{v_n(x_2)\}$ are normalized to have $L^2$-norm 1. The
coefficients of the wave-guide mode terms are given by
\begin{equation}
  a_n(y_2)=-\frac{1}{2}\int_{-d}^{d}\frac{e^{-\sqrt{\xi_n^2-k_1^2}|y_2-z_2|}q(z_2)v_n(z_2)dz_2}{\sqrt{\xi_n^2-k_1^2}}.
\end{equation}
These contributions and their $\pa_{x_1}$-derivatives are in
$\cC^1(\bbR).$  Using the fact that
$(\pa_{x_2}^2+k_1^2-\xi_n^2)v_n(x_2)=-q(x_2)v_n(x_2),$ and integration by parts,
we show that
\begin{equation}
  a_n(y_2)=v_n(y_2),
\end{equation}
and therefore
\begin{equation}\label{eqn43.34}
  w^g_{0+}(x;y)=\sum_{n=1}^Nv_n(x_2)v_n(y_2)e^{i\xi_n(x_1-y_1)}
    \text{ for }x_1>0.
\end{equation}
Hence this term and its $x_1$-derivatives decay exponentially as
$|x_2|+|y_2|\to\infty.$ The fundamental solution for $\Delta+k_1^2+q(x_2)$ is
given by
\begin{equation}\label{eqn35.91}
  \fE(x;y)=g_{k_1}(|x-y|)+w_{0+}^{c}(x;y)+w_{0+}^{g}(x;y);
\end{equation}
the guided modes are entirely captured by $w^g_{0+}.$ This and the following
section describes the construction in the right half plane; the left half plane
is obtained by replacing $(x_1,y_1)$ with $(-x_1,-y_1),$ and $\Gamma_{\nu}^+$
with $\Gamma_{\nu}^-,$ its reflection in the $x_1$-axis.

We need to analyze $w_{0+},$ and certain of its derivatives,
along the set $x_1=y_1=0.$ First observe that outside the strip
$|x_2|\leq d,$ this function satisfies the homogeneous elliptic
equation $(\Delta_x+k_1^2)w_{0+}=0,$ and is therefore a
$\cC^{\infty}$-function of $x.$ From~\eqref{eqn26.51} and~\eqref{eqn81.21}  it is
clear that
$\pa_{x_1}w^c_{0+}(x;y)=-\pa_{y_1}w^c_{0+}(x;y),$ hence it suffices to
analyze the smoothness and decay properties of
$\pa_{x_1}^jw^c_{0+}(x;y),$ for $j=0,1,2.$

As noted, away from $x_2=y_2,$ the functions $\pa_{x_1}^jw_{0+}(0,x_2;0,y_2)$
are $\cC^{\infty}$-functions provided $x_2\neq \pm d$ or $y_2\neq \pm d.$ The
function $w_{0+}$  is  $\cC^1$ in a neighborhood of these points with higher
regularity determined by the regularity of $q.$ If $q$ were a smooth
function, then it would follow that
\begin{multline}\label{eqn87.24}
  (\Delta_x+k_1^2+q(x_2))q(x_2)|x-y|^2H^{(1)}_0(k_1|x-y|)=
  \\4 q(x_2)H^{(1)}_0(k_1|x-y|)+O(1+|x-y|\log|x-y|),
\end{multline}
showing that the principal singularity along the diagonal would be given by
\begin{equation}
  w_{0+}(x;y)=-i\frac{q(x_2)}{16}|x-y|^2H^{(1)}_0(k_1|x-y|)+O(|x-y|^3\log|x-y|).
\end{equation}
Even with $q$ given by~\eqref{eqn45.32} this is essentially correct.

\section{Estimates for the Boundary Kernel}\label{sec6}

 The kernels for the integral equations~\eqref{eqn37.24} are constructed from the functions
\begin{equation}\label{eqn87.23}
  \fw^{[j]}(x_2,y_2)=\left[\frac{1}{i}\pa_{x_1}\right]^jw^c_{0+}(x_1,x_2;y_1,y_2)
  \restrictedto_{x_1=y_1=0},\,
  j=0,1,2.
\end{equation}
In this section we state a theorem describing the behavior of these functions.

\begin{theorem}\label{thm00}
  The kernels, $\fw^{[j]}(x_2,y_2),\, j=0,1,2,$ are infinitely differentiable
  outside of $B_d=[-d,d]\times [-d,d].$ Within $B_d$ they are singular along the
  diagonal, where the kernel $\fw^{[j]}(x_2,y_2)$ has an
  $|x_2-y_2|^{2-j}\log|x_2-y_2|$-singularity.

   \begin{enumerate}
   \item If both $|x_2|>d$ and $|y_2|>d,$ then \Rd the kernels are functions of
     $|x_2|+|y_2|,$ and \Bk the following asymptotic expansions hold for
     $j=0,1,2:$
     \begin{equation}\label{eqn171.65}
        \fw^{[j]}(x_2,y_2)\sim
        \frac{e^{ik_1(|x_2|+|y_2|)}}{(|x_2|+|y_2|)^{\frac{j+1}{2}}}\left[M^{\pm,\pm}_{j0}+\sum_{l=1}^{\infty}
          \frac{M^{\pm,\pm}_{jl}}{(|x_2|+|y_2|)^l}\right],\text{ as }|x_2|+|y_2|\to\infty.
     \end{equation}
 \Rd In this set, these kernels are infinitely differentiable and their
 derivatives have asymptotic expansions obtained by differentiating the
 expansions in~\eqref{eqn171.65}. The error terms satisfy uniform estimates
 where $|x_2|+|y_2|>2d+\epsilon,$ for any $\epsilon>0.$\Bk
  \item If one of $|x_2|>d$ or $|y_2|>d,$ then the following asymptotic
     expansions hold for $j=0,1,2:$ 
     \begin{equation}\label{eqn122.56}
       \begin{split}
        \fw^{[j]}(x_2,y_2)&\sim \frac{e^{ik_1|y_2|}}
           {|y_2|^{\frac{j+1}{2}}}\left[\sum_{l=0}^{\infty}\frac{b^{\pm}_{jl}(x_2)}{|y_2|^l}\right],
           \text{
           where } \pm y_2\to \infty,\, |x_2|<d;\\
              \fw^{[j]}(x_2,y_2)&\sim
       \frac{e^{ik_1|x_2|}}{|x_2|^{\frac{j+1}{2}}}\left[\sum_{l=0}^{\infty}
       \frac{c^{\pm}_{jl}(y_2)}{|x_2|^l}\right],       \text{
           where }\pm x_2\to \infty,\, |y_2|<d.
       \end{split}
     \end{equation}
  \Rd In these sets, the kernels are infinitely differentiable in the variable
  whose absolute value is restricted to be greater than $d;$ their derivatives
  in this variable have asymptotic expansions obtained by differentiating the
  expansions in~\eqref{eqn122.56}. The error terms satisfy uniform estimates
  where $|x_2|>d+\epsilon,$ or $|y_2|>d+\epsilon,$ as appropriate, for any
  $\epsilon>0.$\Bk\Bk
     \item
       The kernels
       $$\pa_{x_2}\fw^{[j]}(x_2,y_2) \mp
     ik_1\fw^{[j]}(x_2,y_2),\quad\pa_{y_2}\fw^{[j]}(x_2,y_2) \mp
     ik_1\fw^{[j]}(x_2,y_2)$$
     have asymptotic expansions obtained by applying
     $\pa_{x_2}\mp ik_1,$ or $\pa_{y_2}\mp ik_1$ to the appropriate expansion
     in~\eqref{eqn171.65}, or~\eqref{eqn122.56}. In particular, they satisfy
     \begin{multline}\label{eqn173.65}
               \pa_{x_2}\fw^{[j]}(x_2,y_2) \mp
               ik_1\fw^{[j]}(x_2,y_2)=\\
               \begin{cases} O\left((|x_2|+|y_2|)^{-\frac{j+3}{2}}\right)\text{
                 as } \pm x_2+|y_2|\to\infty,\\
               O\left((|x_2|)^{-\frac{j+3}{2}}\right)\text{
                 as } \pm x_2 \to\infty \text{ with }|y_2|\leq d.
               \end{cases}
     \end{multline}
     \begin{multline}\label{eqn174.65}
               \pa_{y_2}\fw^{[j]}(x_2,y_2) \mp
               ik_1\fw^{[j]}(x_2,y_2)=\\
               \begin{cases} O\left((|x_2|+|y_2|)^{-\frac{j+3}{2}}\right)\text{
                 as }  |x_2|+\pm y_2\to\infty,\\
               O\left((|x_2|)^{-\frac{j+3}{2}}\right)\text{
                 as } \pm y_2 \to\infty \text{ with }|x_2|\leq d.
               \end{cases}
     \end{multline}
   \item The kernel
     \begin{equation}\label{eqn51.300}
          i\fw^{[1]}(x_2,y_2)+\pa_{x_1}w_{0+}^g(0,x_2;0,y_2)=0.
     \end{equation}
          \end{enumerate}
\end{theorem}
\noindent

\begin{remark}
\Rd More complete descriptions of the singularities, within $B_d,$ of the
kernels $\fw^{[j]}(x_2,y_2)$ are given in~\eqref{eqn271.678}. \Bk
\end{remark}
The proof of this theorem requires tedious, but rather standard analysis, using
the properties of solutions to second order ODEs, stationary phase and
integration by parts. \Rd The statements in parts (1) and (2) regarding 
asymptotic expansions for derivatives of the kernel follow, largely from a
classical result found in Coddington and Levinson,~\cite{CoddingtonLevinson}, which
states
     \begin{theorem*}[ Theorem 3.2 (b), Chapter 5 of~\cite{CoddingtonLevinson}]
       If $f(t)\sim \sum_{k=0}^{\infty}p_kt^{-k},$ $f(t)$ is continuously
       differentiable for $t>t_0$ and $f'(t)$ has an asymptotic expansion, then
       \begin{equation}
         f'(t)\sim -\sum_{k=1}^{\infty}kp_kt^{-(k+1)}.
       \end{equation}
     \end{theorem*}

     \noindent
In light of this result all that is needed is a proof that the derivatives of the kernels have
asymptotic expansions, and estimates for the remainder terms. \Bk The details of the
proof of Theorem~\ref{thm00} are given in Appendix~\ref{AppWEsts}.

Similar estimates hold for
$w_{0+}(x_1,x_2;0,y_2),$ and its derivatives
$$\pa_{y_1}w_{0+}(x_1,x_2;0,y_2), \pa_{x_1}w_{0+}(x_1,x_2;0,y_2),\text{ and
}\pa_{x_1}\pa_{y_1}w_{0+}(x_1,x_2;0,y_2),$$ where $x_1>0$ is bounded. Using~\eqref{eqn81.21}
   we express the derivatives of $w_{0+}^c$ as the contour integrals:
   \begin{equation}\label{eqn81.210} 
     \pa_{x_1}^jw^c_{0+}(x;0,y_2)=
     \frac{1}{2\pi}\int_{\Gamma_{\nu}^{+}}(i\xi)^j\tw(\xi,x_2;0,y_2)e^{ix_1\xi}d\xi\text{
    for }x_1>0.
   \end{equation}
  The only difference between these integrals and those estimated in
   Appendix~\ref{AppWEsts} is the factor of $e^{i\xi x_1}$ in the integrand. Recalling
   that
   $$\pa_{x_1}^jw^c_{0+}(x;y)=(-1)^j\pa_{y_1}^jw^c_{0+}(x;y),$$
   it suffices to
  consider these expressions to estimate $\pa_{x_1}w^c_{0+},$
  $\pa_{y_1}w^c_{0+}$ and $\pa_{x_1}\pa_{y_1}w^c_{0+}.$ These estimates are
  essentially the same as those stated in Theorem~\ref{thm00}, as the principal
  term is a stationary phase contribution arising from $\xi=0.$ The contribution
  from the guided modes is clearly infinitely differentiable in the
  $x_1$-variable for all $x_1,$ and satisfies the same estimates as for $x_1=0.$

   It is easy to see that the estimates derived in Section~\ref{sec6.1} for
   $\fw^{[j]},j=0,1,2$ hold equally well for bounded $x_1.$ The addition of the
   factor $e^{i\xi x_1}$ does not change the analysis of the contribution from
   the semi-circular components, $\{C^{\pm}_{j,\nu}\},$ as $\Im\xi\geq 0,$ and
   therefore $|e^{ix_1\xi}|\leq 1$ on this part of the contour. Where
   $|x_2|,|y_2|>d,$ for bounded $x_1,$ the extension to the right half plane has
   the asymptotic expansion
   \begin{equation}
      \pa_{x_1}^jw^c_{0+}(x_1,x_2;0,y_2)=
      C_j\frac{e^{ik_1(|x_2|+|y_2|)}}{(|x_2|+|y_2|)^{\frac{j+1}{2}}}+
      O\left((|x_2|+|y_2|)^{-{\frac{j+3}{2}}}\right)
      \text{ for }j=0,1,2.
   \end{equation}
  
   It is similarly straightforward to handle the estimates where either
   $|x_2|<d,$ or $|y_2|<d.$  Again, since the principal
   contribution comes from $\xi=0,$  for bounded $x_1>0,$ 
   we have
    \begin{equation}
       \begin{split}
         \pa_{x_1}^jw^c_{0+}(x_1,x_2;0,y_2)&=
         \frac{e^{ik_1|y_2|}b^{[j]}_{\pm}(x_2)}{|y_2|^{\frac{j+1}{2}}}+O(|y_2|^{-\frac{j+3}{2}})\text{
             where }\pm y_2>d,\,
         |x_2|<d,\\ \pa_{x_1}^jw^c_{0+}(x_1,x_2;0,y_2)&=\frac{e^{ik_1|x_2|}c^{[j]}_{\pm}(y_2)}
         {|x_2|^{\frac{j+1}{2}}}
           +O(|x_2|^{-\frac{j+3}{2}})\text{ where }\pm x_2>d,\, |y_2|<d.
       \end{split}
    \end{equation}
 In Section~\ref{sec8} we show that the sources $(\sigma,\tau)$ appearing
 in~\eqref{eqn36.50} belong to $\cC_{\alpha}(\bbR)\oplus\cC_{\alpha+\frac
   12}(\bbR),$ for an $0<\alpha<\frac 12,$ see~\eqref{eqn11.97}.  The estimates
 of $w_c(x;0,y_2), \pa_{y_1}w_c(x;0,y_2),$ for bounded $x_1,$ suffice to show
 that the representations of $u^{l,r}$ with such $(\sigma,\tau)$ are given by
 absolutely convergent integrals.

From ellipticity it follows that $w^c_{0+}(x_1,x_2;0,y_2)$ is a 
$\cC^{\infty}$-function if $x_1>0,$ away from $x_2,y_2=\pm d.$  It clear that if $f(y_2)$
is a bounded continuous function, then for any finite $L$ we have
\begin{multline}
  \lim_{x_1\to
    0^+}\pa_{x_1}^j\int_{-L}^{L}w^c_{0+}(x_1,x_2;0,y_2)f(y_2)dy_2=\\
  \int_{-L}^{L}\pa_{x_1}^jw^c_{0+}(0,x_2;0,y_2)f(y_2)dy_2\text{ for }j=0,1,2.
\end{multline}
From the Fourier representations it is clear that the asymptotics, for
$|x_2|+|y_2|$ large, hold uniformly as $x_1\to 0^+.$ Hence if $f$ is
continuous and
\begin{equation}
  |f(y_2)|\leq M\frac{(1+|y_2|)^{\frac j2}}{(1+|y_2|)^{\frac 12+\epsilon}},
\end{equation}
for an $\epsilon>0,$ then,  for $j=0,1,2,$
\begin{equation}
  \begin{split}
  \lim_{x_1\to
    0^+}\pa_{x_1}^j\int_{-\infty}^{\infty}&w^c_{0+}(x_1,x_2;0,y_2)f(y_2)dy_2=\\
&\lim_{x_1\to 0^+}
\int_{-\infty}^{\infty}\pa_{x_1}^jw^c_{0+}(x_1,x_2;0,y_2)f(y_2)dy_2\\
&=\int_{-\infty}^{\infty}\pa_{x_1}^jw^c_{0+}(0,x_2;0,y_2)f(y_2)dy_2\\
&=\left[\pa_{x_1}^j\int_{-\infty}^{\infty}w^c_{0+}(x_1,x_2;0,y_2)f(y_2)dy_2\right]_{x_1=0}.
  \end{split}
\end{equation}
With these observations we can now show that the representation
in~\eqref{eqn36.50} can be used with solutions of the corresponding
boundary integral equations to find and represent the scattered
fields, $u^{l,r}(x).$

   \section{The Integral Equations}\label{sec8}
  Using the computations in Section~\ref{sec6} and Theorem~\ref{thm00} we now
  express the kernels appearing in the integral equations, \eqref{eqn37.24}, in
  terms of $\fw^{[j]}$ and the guided modes for the relevant equations. We
  reintroduce the $l,r$ sub- and superscripts, letting $\fw^{[j]}_{l,r}$ denote
  the kernels defined by $q_{l,r},$ and $w^g_{l,r}$ the contributions of the
  guided modes $\{v_n^{l,r}(x_2)e^{\pm i\xi_n^{l,r}}:\:n=1,\dots, N^{l,r}\},$
  from~\eqref{eqn43.34}. With this notation we have
   \begin{equation}
     \begin{split}
       W^0_{l,r}=w^{l,r}(0,x_2;0,y_2)&=\fw^{[0]}_{l,r}(x_2,y_2)+w^g_{l,r}(0,x_2;0;y_2),\\
       \pa_{x_1}w^{l,r}(0,x_2;0,y_2)&=\pa_{y_1}w^{l,r}(0,x_2;0,y_2)=0,\\
       W^{2}_{l,r}=\pa^2_{x_1y_1}w^{l,r}(0,x_2;0,y_2)&=\fw^{[2]}_{l,r}(x_2,y_2)+
       \pa^2_{x_1y_1}w^g_{l,r}(0,x_2;0;y_2).
       \end{split}
   \end{equation}
   We assume that $\supp q_{l,r}\subset (-d,d).$ 

   The integral equations therefore can be written
   \begin{equation}\label{eqn141.34}
    \left(\begin{matrix}\Id &W^0_r-W^0_l\\W^2_l-W^2_r&\Id\end{matrix}\right)
              \left(\begin{matrix}\sigma\\\tau\end{matrix}\right)=
        \left(\begin{matrix}g\\h\end{matrix}\right),
   \end{equation}
   which, to simplify notation, we rewrite as
     \begin{equation}\label{eqn143.35}
     \left(\begin{matrix}\Id & D\\C&\Id\end{matrix}\right)
              \left(\begin{matrix}\sigma\\\tau\end{matrix}\right)=
        \left(\begin{matrix}g\\h\end{matrix}\right).
     \end{equation}
The main results of this section provide a natural functional analytic setting where
this is  a Fredholm integral equation of index zero, which implies that it is solvable
subject to finitely many linear conditions on $(g,h).$ In Part III we prove that
the null-space is trivial and therefore the equation is always solvable.

   The analysis of $\fw^{[2]}_{l,r}$ shows that the operator $W^2_l-W^2_r$
   is compact on $L^2(\bbR),$ amongst other spaces. While $W^0_r-W^0_l$ is
   smoothing, the  $(|x_2|+|y_2|)^{-\frac 12}$ asymptotic behavior of
   $\fw^{[0]}_{l,r}$ prevents it from being defined on $L^2(\bbR),$ let alone
   compact.   We work instead with  the following subspaces of $\cC^0(\bbR):$
   \begin{definition}
      For $\alpha\in\bbR,$  let $\cC_{\alpha}(\bbR)$ denote
      continuous functions on $\bbR$ with
   \begin{equation}
     |f|_{\alpha}=\sup\{(1+|x|)^{\alpha}|f(x)|:\:x\in\bbR\}<\infty.
   \end{equation}
   \end{definition}
   
   These spaces are somewhat like H\"older spaces, in that, 
   $\cC^{\infty}_{c}(\bbR)$ is not dense in $\cC_{\alpha}(\bbR)$ with respect to
   the $|\cdot|_{\alpha}$-norm. A usable replacement is the fact that
   $\cC^{\infty}_{c}(\bbR)$ is dense in $\cC_{\alpha}(\bbR)$ with respect to the
   $|\cdot|_{\alpha'}$-norm, for any $0<\alpha'<\alpha.$ It is important to have
   a criterion for when a bounded linear operator $A:\cC_{\alpha}(\bbR)\to \cC_{\alpha}(\bbR)$
   is compact. We give a simple sufficient condition.
   \begin{proposition}\label{prop3.97}
     Let $0<\alpha<\beta,$ and let $A:\cC_{\alpha}(\bbR)\to \cC_{\beta}(\bbR)$
     be a bounded linear operator. Let
     $B_r=\{f\in\cC_{\alpha}(\bbR):\:|f|_{\alpha}<r\}.$ If, for any $0<X,$ the
     image $AB_r$ restricted to $[-X,X]$ is a uniformly equicontinuous family of
     functions, then $A:\cC_{\alpha}(\bbR)\to \cC_{\alpha}(\bbR)$ is a compact
     operator.
   \end{proposition}
   \begin{proof}
     To prove the proposition we need to show that if $\{f_n\}\subset
     \cC_{\alpha}(\bbR)$ is a bounded sequence, then $\{Af_n\}$ has a
     $\cC_{\alpha}(\bbR)$-convergent subsequence. The hypotheses of the
     proposition imply that there are positive constants $M,r$ so that
     $\{f_n\}\subset B_r$ and 
     \begin{equation}
       |Af|_{\beta}\leq M|f|_{\alpha}\text{ for all }f\in\cC_{\alpha}(\bbR).
     \end{equation}
         The restriction of $\{Af_n\}$ to any interval $[-X,X]$ is a bounded,
         uniformly equicontinuous family. Hence a simple diagonal argument using
         the Arzela-Ascoli theorem produces a subsequence $\{f_{n_j}\}$ so that
         $\{Af_{n_j}\}$ converges to $g\in\cC^0(\bbR)$ uniformly on any interval
         $[-X,X].$ In fact this sequence also converges in $\cC_{\alpha}(\bbR).$

       For the terms of the sequence we have the estimate
       \begin{equation}
         |Af_{n_j}(x)|(1+|x|)^{\beta}\leq Mr.
       \end{equation}
       Letting $j\to\infty$ shows that
       \begin{equation}
         |g(x)|(1+|x|)^{\beta}\leq Mr
       \end{equation}
       as well. These estimates and the triangle inequality show that 
       \begin{equation}
         |Af_{n_j}(x)-g(x)|(1+|x|)^{\alpha}\leq 2Mr\frac{(1+|x|)^{\alpha}}{(1+|x|)^{\beta}}.
       \end{equation}
        For an $\epsilon>0,$ we can therefore choose $X$ so that
       \begin{equation}
           |Af_{n_j}(x)-g(x)|(1+|x|)^{\alpha}\leq \epsilon\text{ if }|x|>X.
       \end{equation}
       As $\{Af_{n_j}\restrictedto_{[-X,X]}\}$ converges uniformly to
       $g\restrictedto_{[-X,X]},$ there is a $J$ so that if $j>J,$ then
       \begin{equation}
         |Af_{n_j}(x)-g(x)|(1+|x|)^{\alpha}\leq \epsilon\text{ if }|x|\leq X.
       \end{equation}
       Together these estimates show that
       \begin{equation}
         |Af_{n_j}-g|_{\alpha}<\epsilon\text{ if }j>J,
       \end{equation}
       which completes the proof of the proposition.
   \end{proof}
  
   In the estimates below the function $m(q_l,q_r)$ is a continuous function of
   the norms $\|q_l\|_{L^{\infty}},\|q_r\|_{L^{\infty}},$ and
   $\|q_l-q_r\|_{L^{\infty}},$ which satisfies:
   \begin{equation}
     m(q,q)=0.
   \end{equation}
   The estimates on the kernels follow from Theorem~\ref{thm00}.
   The kernel, $k_D(x_2,y_2),$ of $D$ is at least $\cC^1$ and satisfies an
   estimate of the form,
   \begin{equation}
     |k_D(x_2,y_2)|\leq \frac{m(q_l,q_r)}{(1+|x_2|+|y_2|)^{\frac 12}}.
   \end{equation}
   The kernel, $k_C(x_2,y_2),$ of $C$ is singular on the diagonal in $B_d,$ with
   a singularity of the form $\log|x_2-y_2|\chi_{B_d}(x_2,y_2)$ and
   \begin{equation}
     |(1-\varphi(x_2,y_2))k_C(x_2,y_2)|\leq
     \frac{m(q_l,q_r)}{(1+|x_2|+|y_2|)^{\frac 32}}.
   \end{equation}
   Here $\varphi\in\cC^{\infty}_{c}(B_{d+2\epsilon})$ for an $\epsilon>0,$ with
   \begin{equation}
     \varphi(x_2,y_2)=
           1\text{ for }(x_2,y_2)\in B_{d+\epsilon}.
   \end{equation}

  We begin with the following boundedness result for the operator
   appearing in~\eqref{eqn143.35}.
 \begin{proposition}\label{prop1}
   For $0<\alpha<\frac 12,$ there is a constant $M_{\alpha}$ so that
   if $(\sigma,\tau)\in\cC_{\alpha}(\bbR)\oplus\cC_{\alpha+\frac 12}(\bbR),$ then
   \begin{equation}\label{eqn153.06}
     |D\tau|_{\alpha}+|C\sigma|_{\alpha+\frac 12}\leq M_{\alpha}
     m(q_l,q_r)\left[ |\tau|_{\alpha+\frac 12}+|\sigma|_{\alpha}\right].
   \end{equation}
  \Rd  If the supports of $q_l, q_r$ are contained in $(-d,d),$ then the functions
  $C\sigma, D\tau$ belong to $\cC^{\infty}((-\infty,-d]\cup [d,\infty));$
      $C\sigma\in\cC^1([-d,d]),$ and $D\tau$ is H\"older continuous in $[-d,d].$ \Bk
 \end{proposition}

 The proof relies on the following lemma
 \begin{lemma}\label{lem1.0}
   If $0<\alpha<1,$ and $\alpha+\beta>1,$ then, for $x_2>0,$ we have
   the estimate
   \begin{equation}
     \int_{0}^{\infty}\frac{dy_2}{y_2^{\alpha}(x_2+y_2)^{\beta}}\leq \frac{M_{\alpha,\beta}}{x_2^{\alpha+\beta-1}}.
   \end{equation}
   \end{lemma}
   \begin{proof}[Proof of Lemma]
     If we let $y_2=x_2 t,$ then the integral becomes
     \begin{equation}
       \frac{1}{x_2^{\alpha+\beta-1}}\int_{0}^{\infty}\frac{dt}{t^{\alpha}(1+t)^{\beta}}
       \leq \frac{M_{\alpha,\beta}}{x_2^{\alpha+\beta-1}}.
     \end{equation}
   \end{proof}

   \begin{proof}[Proof of Proposition]

     We first consider $C\sigma,$ splitting it into a compactly supported
     part, $C_0,$ whose kernel is given by $\varphi(x_2,y_2) k_C(x_2,y_2),$ and
     an unbounded part, $C_1,$ with kernel $(1-\varphi(x_2,y_2))k_C(x_2,y_2).$ It is clear
     that the compactly supported part satisfies
   \begin{equation}
     |C_0\sigma(x_2)|\leq m(q_l,q_r) |\sigma|_{\alpha}\chi_{[-(d+2\epsilon),d+2\epsilon]}(x_2).
   \end{equation}

   To estimate the other part we observe that, if $0<\alpha<1,$ then
   applying the lemma gives
   \begin{equation}
     \begin{split}
       |C_1\sigma(x_2)|&\leq\int_{\bbR}(1-\varphi(x_2,y_2))|k_C(x_2,y_2)||\sigma(y_2)|dy_2\\
       &\leq \int_{\bbR}\frac{m(q_l,q_r)|\sigma|_{\alpha} 
         dy_2}{|y_2|^{\alpha}(|x_2|+|y_2|)^{\frac 32}}\\
            &\leq K'_{\alpha}\frac{m(q_l,q_r)|\sigma|_{\alpha}}{|x_2|^{\alpha+\frac 12}}.
     \end{split}
   \end{equation}
     Which shows that
    \begin{equation}\label{eqn156.06}
     |C\sigma|_{\alpha+\frac 12}\leq M m(q_l,q_r)|\sigma|_{\alpha}.
   \end{equation}
   We now estimate $|D\tau|_{\alpha},$ assuming that
   $0<\alpha<\frac 12;$ applying the lemma gives
   \begin{equation}
     \begin{split}
       |D\tau(x_2)|&\leq
       m(q_l,q_r)|\tau|_{\alpha+\frac 12}\int_{\bbR}\frac{dy_2}{(1+|y_2|)^{\alpha+\frac
           12}(1+|x_2|+|y_2|)^{\frac 12}}\\
       &\leq  K_{\alpha}\frac{m(q_l,q_r)|\tau|_{\alpha+\frac 12}}{|x_2|^{\alpha}}.
     \end{split}
   \end{equation}
   The estimate in~\eqref{eqn153.06} follows from this and~\eqref{eqn156.06}.

     \Rd The functions $C\sigma(x_2)$ and $D\tau(x_2)$ are infinitely
     differentiable outside the supports of $q_l$ and $q_r.$ This follows from
     the estimates on the kernels of $k_C, k_D$ in Theorem~\ref{thm00} and the
     facts that these kernels are infinitely differentiable in $x_2$ where
     $|x_2|>d,$ and that the derivatives of the kernels have asymptotic
     expansions obtained by differentiating the expansions of $k_C, k_D$ term by
     term.  That $C\sigma\in\cC^1(\bbR)$ follows from the estimates in
     Theorem~\ref{thm00}, which show that $k_C(x_2,y_2)$ is a continuously
     differentiable in $x_2,$ and $\pa_{x_2}k_C(x_2,y_2)$ decays at least as
     rapidly as $k_C(x_2,y_2).$ That $D\tau(x_2)$ is H\"older continuous for
     $|x_2|<d+\epsilon,$ for any $\epsilon>0,$ follows from the description of
     the singularity of $k_D(x_2,y_2)$ in~\eqref{eqn271.678}. \Bk

 \end{proof}

 To prove the solvability of~\eqref{eqn143.35} we observe that
 \begin{equation}\label{eqn80.300}
   \left(\begin{matrix}\Id & -D\\0&\Id\end{matrix}\right) \left(\begin{matrix}\Id & D\\C&\Id\end{matrix}\right)\left(\begin{matrix}\Id & 0\\-C&\Id\end{matrix}\right)=
         \left(\begin{matrix}\Id-DC & 0\\0&\Id\end{matrix}\right).
 \end{equation}
 Hence to show that the operator in~\eqref{eqn143.35} is Fredholm of
 index zero it suffices to prove the compactness of the composition,
 \begin{equation}
   DC:\cC_{\alpha}(\bbR)\longrightarrow \cC_{\alpha}(\bbR)\text{ for }0<\alpha<\frac 12.
 \end{equation}

   In Table~\ref{tab1} we show the leading order
   asymptotics for the kernels of $W_{l,r}^j,$ for $j=0,2,$ assuming the
   channel is centered at $x_2=\gamma.$ The differences between $W_{l,r}^j,$
   and the leading terms, shown in Table~\ref{tab1}, are
   $O((|x_2|+|y_2|)^{-\frac{j+3}{2}});$ these differences are smoothing and
   improve decay and therefore Proposition~\ref{prop3.97} implies that they
   define compact operators from $\cC_{\alpha+\frac
     12}(\bbR)\to\cC_{\alpha}(\bbR),$ ($j=0$)
   $\cC_{\alpha}(\bbR)\to\cC_{\alpha+\frac 12}(\bbR),$  ($j=2$) resp. for any
   $\alpha<\frac 12.$
 
 \begin{table}[h]
   \begin{center}
  \begin{tabular}{c|c|c}
$c^j_{-+}\frac{e^{ik_1(|y_2-\gamma|+|x_2-\gamma|)}}{(|y_2-\gamma|+|x_2-\gamma|)^{\frac{j+1}{2}}}$ &
  $\frac{e^{ik_1|y_2-\gamma|}b^j_+(x_2-\gamma)}{|y_2-\gamma|^{\frac{j+1}{2}}}$ &
  $c^j_{++}\frac{e^{ik_1(|x_2-\gamma|+|y_2-\gamma|)}}{(|y_2-\gamma|+|x_2-\gamma|)^{\frac{j+1}{2}}}$\\
  \hline $\frac{e^{ik_1|y_2-\gamma|}c^j_-(y_2-\gamma)}{|x_2-\gamma|^{\frac{j+1}{2}}}$ & $c_{00}^j(x_2-\gamma,y_2-\gamma)$
  &$\frac{e^{ik_1|x_2-\gamma|}c^j_+(y_2-\gamma)}{|x_2-\gamma|^{\frac{j+1}{2}}}$ \\ \hline
  $c^j_{--}\frac{e^{ik_1(|x_2-\gamma|+|y_2-\gamma|)}}{(|y_2-\gamma|+|x_2-\gamma|)^{\frac{j+1}{2}}}$ &
  $\frac{e^{ik_1|y_2-\gamma|}b^j_-(x_2-\gamma)}{|y_2-\gamma|^{\frac{j+1}{2}}}$ &
  $c^j_{+-}\frac{e^{ik_1(|x_2-\gamma|+|y_2-\gamma|)}}{(|x_2-\gamma|+|y_2-\gamma|)^{\frac{j+1}{2}}}$
  \end{tabular}
  \end{center}
  \caption{Schematic for the structure of the leading terms of
    $W^j_{l,r}$ assuming the channel is centered on $\gamma.$ If the channel
    has width $2\delta,$ then $+$ is the requirement that a variable is greater than $\gamma+\delta,$ and
    $-$ is the requirement that a variable is less than $\gamma-\delta.$ }
  \label{tab1}
 \end{table}

 \begin{remark}\label{rmk10}
 Before proceeding with our analysis, we observe that the kernels for $C$ and
 $D$ have asymptotic expansions exactly of the form given in
 Theorem~\ref{thm00}. In the statement of the theorem we normalize the
 coordinates so that the support of the potential is a symmetric interval
 $[-d,d].$ Clearly if we have both a left and a right potential this may not be
 possible. We can pick coordinates so that $\supp q_l=[-\delta_l,\delta_l],$ but then $\supp
 q_r=[\gamma_r-\delta_r,\gamma_r+\delta_r],$ for some positive $\gamma_r$ and $\delta_r.$
 The terms of the asymptotic expansion coming from $q_r$ are of the form given
 in Table~\ref{tab1}. Since, for example, if $x_2,y_2>\gamma_r,$ $\alpha>0,$ then
 \begin{equation}
   \frac{1}{(x_2+y_2-2\gamma_r)^{\alpha}}=\frac{1}{(x_2+y_2)^{\alpha}}\cdot
   \sum_{j=0}^{\infty} C_{j,\alpha}\left(\frac{2\gamma_r}{x_2+y_2}\right)^j,
 \end{equation}
it is apparent that, for large $|x_2|+|y_2|$ these expansions can be rewritten
to take exactly the form in Theorem~\ref{thm00}.  Hence the kernels for $C$ and
$D$ also have such expansions, outside an interval contains $\supp q_l\cup\supp
q_r,$
 \end{remark}
 
To analyze the kernel of composition $D C,$ we choose $d>0,$ as above, so that the
supports of $q_l$ and $q_r$ are contained in $(-d,d).$ We choose a
function $\varphi\in\cC^{\infty}_c(\Int B_{d+4\epsilon}),$ which equals $1$ in
$B_{d+2\epsilon}.$ We let $C_0$ (resp. $D_0$) have the kernel
$\varphi(x_2,y_2)k_C(x_2,y_2),$ (resp. $\varphi(x_2,y_2)k_D(x_2,y_2)$), and $C_1$ (resp. $D_1$) have kernel
$(1-\varphi(x_2,y_2))k_C(x_2,y_2),$
(resp. $(1-\varphi(x_2,y_2))k_D(x_2,y_2)$). The composition then splits into the
terms
\begin{equation}
  DC=D_0C_0+D_0C_1+D_1C_0+D_1C_1.
\end{equation}
The kernels of $D_0$ and $C_0$ are compactly supported. The kernel of $C_0$ is
differentiable, whereas the kernel of $D_0$ has a $\log|x_2-y_2|$-singularity
for $(x_2,y_2)\in B_d,$ and a $\log(|x_2+d|+|y_2+d|)+\log(|x_2-d|+|y_2-d|)$ in
$B_d^c.$ Such a kernel maps bounded data into H\"older continuous data. The
kernels of $D_1,C_1$ are continuously differentiable and have specified rates of
decay; it is not difficult to show that
\begin{equation}\label{eqn163.57}
  D_0C_0+D_0C_1+D_1C_0:\cC_{\alpha}(\bbR)\to\cC_{\alpha}(\bbR)\text{ is compact
    for any }0<\alpha<\frac 12.
\end{equation}

Using smooth cut-off functions, we now divide $D_1,C_1$ into operators
$D_{11},C_{11},$ with kernels supported in the set
$\{(x_2,y_2):\:|x_2|>d+\epsilon,|y_2|>d+\epsilon\},$ and the remainders
$D_{10}=D_1-D_{11},C_{10}=C_1-C_{11}.$ This additional splitting is useful as
the kernels for $D_{11},C_{11}$ have simpler asymptotics outside of the channels
centered on the $x_2$ and $y_2$ axes. The leading terms in the expansions of the
kernels of $D_{10}, C_{10}$ define finite rank operators, and the remaining terms
are obviously compact. It therefore follows that
\begin{equation}\label{eqn164.57}
  D_{10}C_{10}+D_{10}C_{11}+D_{11}C_{10}:\cC_{\alpha}(\bbR)\to\cC_{\alpha}(\bbR)\text{ is compact
    for any }0<\alpha<\frac 12.
\end{equation}

This leaves just the $D_{11}C_{11}$ term. If we suppose that the left channel is
centered at $0$ and the right at $\gamma,$ then,
from Table~\ref{tab1}, it follows that  the leading terms of $k_{C_{11}},k_{D_{11}}$ are given by
   \begin{multline}\label{eqn164.95}
     \tk^{0}_{j}=\sum_{\chi_0,\chi_1\in\{-,+\}}
     \psi(\chi_0 x_2)
     \psi(\chi_1y_2)\Bigg[c_{\chi_0\chi_1}^{l,j}\frac{e^{ik_1(|x_2|+|y_2|)}}{(|x_2|+|y_2|)^{\frac{j+1}{2}}}-\\
       c_{\chi_0\chi_1}^{r,j}\frac{e^{ik_1(|x_2-\gamma|+|y_2-\gamma|)}}{(|x_2-\gamma|+|y_2-\gamma|)^{\frac{j+1}{2}}}\Bigg],
   \end{multline}
   with $j=0$ for $D$ and $j=2$ for $C,$ and $\psi(z)\in\cC^{\infty}(\bbR)$ is a
   non-negative,  function supported where $z\geq d+\epsilon$ and equal to $1$
   where $z>d+4\epsilon.$ Note that $d$ is selected so that $\supp q_l,\supp
   q_r$ are compact subsets of $(-d,d).$ The very simple form of these kernels
   allows us to estimate the leading order part of the kernel of the
   composition, $D_{11}C_{11}.$
   \begin{proposition}\label{prop2.96}
     There are positive constants $m,M$ so that
     \begin{equation}\label{eqn165.95}
       \left|\int_{-\infty}^{\infty}\tk^0_0(x_2,z)\tk^0_2(z,y_2)dz\right|\leq
       \frac{M}{(|x_2|+m)^{\frac 12}(|y_2|+m)^{\frac 32}}.
     \end{equation}
   \end{proposition}
   \begin{proof}
     From~\eqref{eqn164.95} it follows that the integral
     in~\eqref{eqn165.95} is a sum of integrals over either
     $(-\infty,-d]$ or $[d,\infty)$ consisting of terms of the form
     \begin{equation}\label{eqn176.97}
               e^{ik_1(|x_2-a|+|y_2-b|)}\frac{\psi(|z|)e^{ik_1
                   (|z-a|+|z-b|)}}{(|x_2-a|+|z-a|)^{\frac
                   12}(|y_2-b|+|z-b|)^{\frac 32}},
     \end{equation}
     where $a$ and $b$ equal either $0$ or $\gamma.$ Within the domain of
     integration neither $z-a,$ nor $z-b$ changes sign. All of the various terms
     are estimated by integrating by parts. For example, in case $a=b=0,$ an
     integration by parts shows that the integral over $[d,\infty)$  equals
     \begin{multline}
      \int_{d}^{\infty}\frac{\psi(z)e^{2ik_1z}dz}{(|x_2|+z)^{\frac
                   12}(|y_2|+z)^{\frac 32}}=
     \frac{1}{2ik_1}\Bigg[ \int_{d}^{\infty}
      \frac{-e^{2ik_1z}\psi'(z)dz}{(|x_2|+z)^{\frac
                12}(|y_2|+z)^{\frac 32}}+\\
            \frac{1}{2}\int_{d}^{\infty}\frac{e^{2ik_1  z}\psi(z)}{(|x_2|+z)^{\frac
                12}(|y_2|+z)^{\frac 32}}\left(\frac{1}{|x_2|+z}+
            \frac{3}{|y_2|+z}\right)dz\Bigg].
     \end{multline}
     The first integral on the right is easily seen to satisfy an estimate like
     that in~\eqref{eqn165.95}.  Integrating by parts one more time in the other
     terms gives a similar formula from which the estimates follows easily. All
     other types of terms appearing in~\eqref{eqn176.97} are estimated using the same integrations by parts.
   \end{proof}

   We have the following corollary.
   \begin{corollary}
  For any $0<\alpha<\frac 12,$ the operator
  $D_{11}C_{11}:\cC_{\alpha}(\bbR)\to\cC_{\frac 12}(\bbR)$ is bounded,
  and therefore $DC:\cC_{\alpha}(\bbR)\to\cC_{\alpha}(\bbR)$ is a compact
  operator.
   \end{corollary}
   \begin{proof}
     The first statement is an immediate consequence of the estimate
     in~\eqref{eqn165.95}, and the fact that the differences between
     $k_{D_{11}}$ and $k_{C_{11}}$ and their leading parts are bounded
     by $\frac{M}{(1+|x_2|+|y_2|)^{\frac 32}}$ and
     $\frac{M}{(1+|x_2|+|y_2|)^{\frac 52}},$ respectively. The
     $x_2$-derivative of the kernel of $D_{11}C_{11}$ is easily seen
     to satisfy the same type of estimates, hence
     Proposition~\ref{prop3.97} applies to show that
     $D_{11}C_{11}:\cC_{\alpha}(\bbR)\to\cC_{\alpha}(\bbR)$ is a
     compact operator.  The second statement follows from this
     observation along with~\eqref{eqn163.57} and~\eqref{eqn164.57}.
   \end{proof}

   This corollary  immediately implies:
   \begin{corollary}For any $0<\alpha<\frac 12$ the operator
     $(\Id-DC):\cC_{\alpha}(\bbR)\to \cC_{\alpha}(\bbR)$ is a Fredholm operator of
   index 0.
   \end{corollary}

   \noindent
   From this corollary and~\eqref{eqn80.300} we deduce the following:
   \begin{corollary} For any $0<\alpha<\frac 12$ the operator
     \begin{equation}\label{eqn97.668}
       \left(\begin{matrix}\Id& D\\C&\Id\end{matrix}\right)
     \end{equation}
     acting from $ \cC_{\alpha}(\bbR)\oplus \cC_{\alpha+\frac
       12}(\bbR)$ to itself is a Fredholm operator of index 0.
   \end{corollary}
   \begin{remark}
     I want to thank Tristan Goodwill and Manas Rachh for pointing out
     a small gap in an earlier version of the proof of this corollary, and suggesting a
     correction.
   \end{remark}

   \begin{remark}
  To prove the solvability for arbitrary data in $\cC_{\alpha}(\bbR)$
  we still need to show that the operator in~\eqref{eqn97.668} has a
  trivial null-space. This is proved in Part III, where it is seen to
  follow from the uniqueness of the outgoing solution to the original
  scattering problem.
  \end{remark}
  
Suppose that $(g,h)\in \cC_{\alpha}(\bbR)\oplus \cC_{\alpha+\frac 12}(\bbR),$ and
$  (\Id-DC)\sigma=g-Dh$
is solvable for $\sigma\in\cC_{\alpha}(\bbR),$  if we set $\tau=h-C\sigma\in
\cC_{\alpha+\frac 12}(\bbR),$ then the pair
$(\sigma,\tau)$  solves~\eqref{eqn143.35}.  In this case the solution to the
transmission problem is given by~\eqref{eqn36.34},
 \begin{equation}
  u^{l,r}=-\cE^{l,r\,'}\sigma+\cE^{l,r}\tau,
\end{equation}
which implies that
\begin{equation}\label{eqn160.08}
  \begin{split}
  u^{l,r}(x)
 &=\cS_{k_1}\tau(x)+\int_{-\infty}^{\infty}w^{l,r}(x;0,y_2)\tau(y_2)dy_2-\\
 &\cD_{k_1}\sigma(x)-\int_{-\infty}^{\infty}\pa_{y_1}w^{l,r}(x;0,y_2)\sigma(y_2)dy_2,
  \end{split}
\end{equation}
where $\mp x_2>0.$  The kernel for
single layer satisfies the estimate
\begin{equation}
 \left| \frac{H^{(1)}_0(k_1|x-(0,y_2)|)}{4}\right|\leq
 \frac{M}{[x_1^2+(x_2-y_2)^2]^{\frac 14}}. 
\end{equation}
The kernel of the double layer is given by
\begin{equation}
  \pa_{y_1}\frac{iH^{(1)}_0(k_1|x-y|)}{4}\restrictedto_{y_1=0}=-i\frac{k_1x_1}{4|x-(0,y_2)|}
  \pa_zH^{(1)}_0(k_1|x-(0,y_2)|).
\end{equation}
For $x_1\neq 0,$ as $y_2\to\infty,$ it satisfies the estimate
\begin{equation}
  \left|\pa_{y_1}\frac{H^{(1)}_0(k_1|x-y|)}{4}\restrictedto_{y_1=0}\right|\leq
  \frac{Mx_1}{[x_1^2+(x_2-y_2)^2]^{\frac 34}}.
\end{equation}
The estimates proved for
$w^{l,r}(0,x_2;0,y_2),\pa_{y_1}w^{l,r}(0,x_2;0,y_2)$ also hold where
$x_1\neq 0,$ which, along with Proposition~\ref{prop1}, shows that
representations for $u^{l,r}(x_1,x_2)$ in~\eqref{eqn160.08} are given by
absolutely convergent integrals.

\begin{remark}
  One can imagine other uses for the fundamental solution, $\fE,$
  of operators like $(\Delta+q(x_2)+k_1^2)$ constructed above.  A
  simple example would be to change the electrical properties of a bi-infinite
  wave-guide in a compact set replacing $q(x_2)$ by $q(x_2)+Q(x_1,x_2)$,
  with $Q$ a compactly supported function.  Suppose that $u^{\In}$ is
  a solution to $(\Delta+q+k_1^2)u^{\In}=0,$ and we seek an outgoing
  solution, $u^{\out},$ to
\begin{equation}
  (\Delta+q+k_1^2+Q)[u^{\In}+u^{\out}]=0.
\end{equation}
Using the fundamental solution this can be rewritten as a Lipmann-Schwinger
type equation:
\begin{equation}
  (\Id+\fE Q)u^{\out}=-\fE Q u^{\In}.
\end{equation}
At least for small $Q,$ this equation can be solved using a Neumann series
\begin{equation}
  u^{\out}=-\fE Q\sum_{j=0}^{\infty}(-1)^j(\chi_Q\fE Q)^ju^{\In},
\end{equation}
where $\chi_Q$ is the characteristic function of $\supp Q.$ To compute the terms
of the sum only requires a knowledge of the kernel of $\fE$ on $\supp
Q\times\supp Q.$

A similar approach can be used to study the effect of placing an
non-transparent obstacle in the channel. For these cases the scattered
field can be represented in terms of the sum of a single and double
layer with respect to the kernel of $\fE$ over the boundary of the
obstacle. This will lead to a second kind Fredholm equation on the
boundary of the obstacle.

A more ambitious application might be to study a network of channels
meeting in a compact set.  Using an idea similar to that employed
in~\cite{BonChaFli_2022} one can decompose $\bbR^2$ into a collection
of truncated sectors, $\{S_1,\dots, S_N\}$ each containing a single
semi-infinite channel
$$\{x:\:|\langle x,v_j^{\bot}\rangle-c_j| \leq d_j,\, \langle
x,v_j\rangle>e_j\}, \text{ for }v_j\in\bbR^2\text{  unit vectors,}$$
with electrical properties modeled by an operator of the form
$$L_j=(\Delta+q_j(\langle x,v_j^{\bot}\rangle-c_j)+k_1^2).$$
Here $\langle\cdot,\cdot\rangle$ is the Euclidean inner product in $\bbR^2.$
The channels meet in a compact interaction zone, $D.$  See
Figure~\ref{fig3}. Using our construction we can build a fundamental
solution, $\fE_j=\cS_{k_1}+W_j,$ for each operator $L_j.$

\begin{figure}
  \centering
  \includegraphics[width= 8cm]{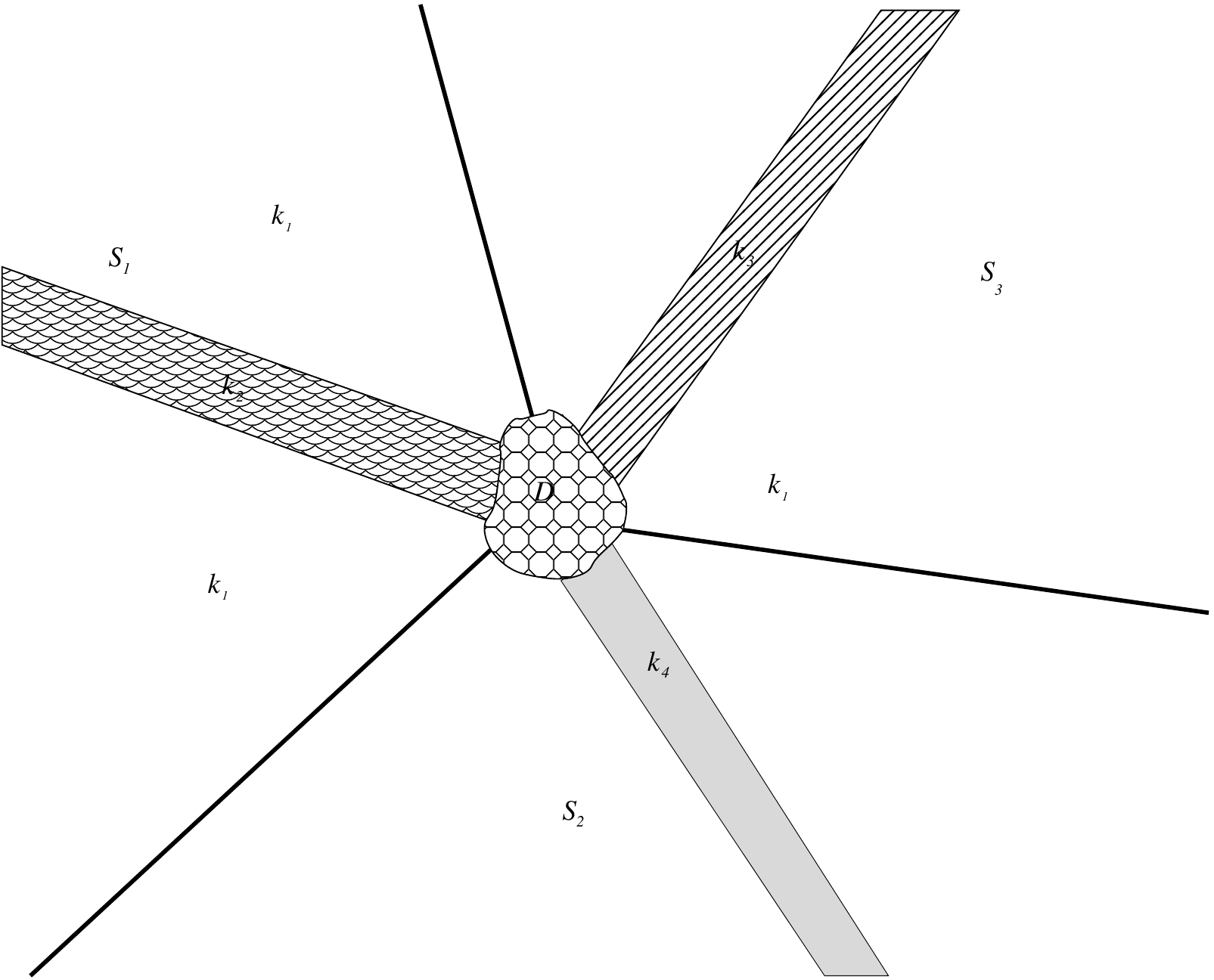}
    \caption{Three dielectric channels meeting in a compact interaction zone,
      $D,$ showing sectors $S_1, S_2,S_3.$}
    \label{fig3}
\end{figure}

Solutions to $L_ju_j=0$ in $S_j$ can then be written as sums of single and
double layers w.r.t. $\fE_j$ integrated over $\pa S_j.$ Imposing jump conditions
across the common boundaries of the sectors would then lead to systems of
integral equations over $\cup \pa S_j,$ analogous to~\eqref{eqn37.24}. These
would be supplemented with boundary conditions on $\pa D.$ As before, the
kernels of the various $W_j$ decay like $(|s|+|t|)^{-\frac 12}$ as one goes out
to infinity along components of the $\pa S_j.$ Unfortunately, unless the $\pa
S_j$ is orthogonal to the channel lying in $S_j,$ the normal derivatives of
these kernels will also decay at this rate. Thus it seems unlikely that these
integral equations will be well posed in any useful function space. Adding a
little dissipation does lead to tractable integral equations, which, in light of
the limiting absorption principle, provides a viable method for the approximate
solution of such problems.

\end{remark}

\section{Admissible Data}\label{sec_adm_data}

In general, our method for solving the transmission problem specified
in~\eqref{eqn5}--\eqref{eqn5.30} is applicable to data
$(g,h)\in\cC_{\alpha}(\bbR) \oplus\cC_{\alpha+\frac 12}(\bbR),$ for some
$0<\alpha<\frac 12.$ If the incoming fields $u^{\In}_{l,r}$ are sums of wave-guide
modes, then they decay exponentially as $|x_2|\to\infty,$ and are therefore
admissible as data for our method. In general, there are two other types of
incoming data that naturally arise in this context: plane waves, and point
sources.

In this setting, point sources will arise from taking a fundamental solution for a
bi-infinite wave-guide, $\fE^{l,r}(x;y),$ which is precisely what we have constructed
above. Using our representation, we have
\begin{equation}
  \fE^{l,r}(x;y)=g_{k_1}(x-y)+w^{l,r}(x;y).
\end{equation}
If we fix a point $y_0=(y_{01},y_{02})$ with $y_{01}<0,$ then we can use
$u^{\In}_l(x)=\fE^l(x;y_0)$ as point source in the left half plane at $y_0;$ we
can let $u^{\In}_r=0.$ The data for the transmission problem
is then
  \begin{equation}
    g(x_2)=\fE^{l}(0,x_2;y_0),\, h(x_2)=\pa_{x_1}\fE^{l}(0,x_2;y_0).
  \end{equation}
  The analysis in Appendix~\ref{AppWEsts}  is easily adapted to show that
  $g(x_2)=O(|x_2|^{-\frac 12}), $ and $h(x_2)=O(|x_2|^{-\frac 32}).$ In fact, these
  functions are `outgoing,' and have complete asymptotic expansions, as $\pm
  x_2\to\infty,$ of the form
  \begin{equation}
    \begin{split}
      &g(x_2)\sim
      \frac{e^{ik_1|x_2|}}{|x_2|^{\frac 12}}\sum_{j=0}^{\infty}\frac{a_j^{\pm}}{|x_2|^{j}},\\
        &h(x_2)\sim \frac{e^{ik_1|x_2|}}{|x_2|^{\frac 32}}\sum_{j=0}^{\infty}\frac{b_j^{\pm}}{|x_2|^{j}}.
    \end{split}
  \end{equation}
  This sort of asymptotic behavior is needed in order for the solutions we
  obtain to be outgoing.

 The case of incoming plane waves is similar. 
In the case that the wave-guide is a bi-infinite channel, with $k_1<k_2,$ as described
by~\eqref{eqn2.32}, the scattering problem for an incoming plane wave
`from above' has an elementary solution: Let $\bkappa=(\kappa_1,\kappa_2)$ satisfy
\begin{equation}
  \kappa_1^2+\kappa_2^2=k_1^2, \quad \kappa_2<0,
\end{equation}
then the function $v^{\In}=e^{i\bkappa\cdot\bx}$ is a reasonable incoming field
`from above' for the single channel modeled by $(\Delta+q(x_2)+k_1^2),$ with
$q(x_2)=\chi_{[-d,d]}(x_2)(k_2^2-k_1^2).$ We would like to find the outgoing
scattered wave, $v^{\scat},$ produced by this incoming field. We let
$\bkappa'=(\kappa_1,-\kappa_2),$ and $\tbkappa=(\kappa_1,\hat{\kappa}_2),
\tbkappa'=(\kappa_1,-\hat{\kappa}_2),$ where
$\hat{\kappa}_2=\sqrt{k_2^2-\kappa_1^2}>0.$ The scattered field can be found
using the jump conditions directly and takes the form predicted by the Fresnel
relations
\begin{equation}
  v^{\scat}(\bkappa;\bx)=\begin{cases}&\alpha^+(\bkappa)e^{i\bkappa'\cdot\bx}\text{ where
  }x_2>d,\\
  &\alpha^0(\bkappa)e^{i\tbkappa\cdot\bx}+\beta^0(\bkappa)e^{i\tbkappa'\cdot\bx}\text{ where }|x_2|<d,\\
  &\alpha^-(\bkappa)e^{i\bkappa\cdot\bx}\text{ where }x_2<-d.
  \end{cases}
\end{equation}
The determinant of the linear system that defines the coefficients,
$(\alpha^+(\bkappa),\alpha^0(\bkappa),$ $\beta^0(\bkappa),\alpha^-(\bkappa)),$
is a non-zero multiple of $2\kappa_2\hat{\kappa}_2\cos
2\hat{\kappa}_2d+i(k_2^2-k_1^2)\sin2\hat{\kappa}_2d,$ which does not vanish
provided that $\kappa_2\neq 0.$ Hence if $\bkappa=k_1(\cos\theta,-\sin\theta),$
then these coefficients depend smoothly on $\theta\in (0,\pi).$ A general field
incoming `from above' takes the form
\begin{equation}\label{eqn6.80}
  v^{\In}_{\mu}(x)=\int_{0}^{\pi} e^{ik_1(\cos\theta,-\sin\theta)\cdot\bx}d\mu(\theta),
\end{equation}
with $d\mu$ a finite measure on $(0,\pi).$
By linearity, it produces an `outgoing' scattered field of the form
\begin{equation}
  v^{\out}_{\mu}=\int_{0}^{\pi} v^{\scat}(k_1(\cos\theta,-\sin\theta);\bx) d\mu(\theta).
\end{equation}

While our method for analyzing a pair of intersecting semi-infinite wave-guides
does not apply directly to incoming fields that do not decay as $|x_2|\to\infty,$ if
$d\mu(\theta)=m(\theta)d\theta,$ with $m\in\cC^{\infty}_c((0,\pi)),$ then
$$v^{\tot}_{\mu}(x_1,x_2)=v^{\In}_{\mu}(x_1,x_2)\chi_{[d,\infty)}(x_2)+v^{\out}_{\mu}(x_1,x_2)$$
  satisfies the transmission boundary conditions and is a weak solution of the PDE
  $(\Delta+q(x_2)+k_1^2)v^{\tot}_{\mu}=0.$ A stationary phase computation shows
  that $v^{\out}_{\mu}(\bkappa;r\eta)$ satisfies the Sommerfeld radiation
  conditions if either $\eta,$ or $-\eta$ belongs to the $\supp m,$ and is
  rapidly  vanishing at infinity otherwise. Moreover we have asymptotic
  expansions
\begin{equation}\label{eqn174.210} 
  \begin{split}
 & v^{\tot}_{\mu}(0,x_2)\sim
  \begin{cases}
    & \frac{e^{-ik_1x_2}}{\sqrt{x_2}}\sum\limits_{j=0}^{\infty}\frac{a_j^-}{x_2^j}+
    \frac{e^{ik_1x_2}}{\sqrt{x_2}}\sum\limits_{j=0}^{\infty}\frac{a_j^+}{x_2^j}\text{ for }x_2>0,\\
    &  \frac{e^{-ik_1x_2}}{\sqrt{|x_2|}}\sum\limits_{j=0}^{\infty}\frac{b_j^+}{x_2^j}\quad\text{ for }x_2<0,
  \end{cases}\\
 &  \pa_{x_1}v^{\tot}_{\mu}(0,x_2)\sim
  \begin{cases}
    & \frac{e^{-ik_1x_2}}{x^{3/2}_2}\sum\limits_{j=0}^{\infty}\frac{a_j^{'-}}{x_2^j}+
    \frac{e^{ik_1x_2}}{x_2^{3/2}}\sum\limits_{j=0}^{\infty}\frac{a_j^{'+}}{x_2^j}\text{ for }x_2>0,\\
    &  \frac{e^{-ik_1x_2}}{|x_2|^{3/2}}\sum\limits_{j=0}^{\infty}\frac{b_j^{'+}}{x_2^j}\quad\text{ for }x_2<0.
  \end{cases}
  \end{split}
\end{equation}
and therefore our method of solution, with $q_l=q,$ and data determined by
$u^{\In}_l=v^{\tot}_{\mu}, u^{\In}_r=0$ does apply to this case.  To obtain the
needed estimates for our approach to apply it suffices for
$m\in\cC^{2}_c((0,\pi)).$

Unfortunately, as is clear from~\eqref{eqn174.210}, data of this type may not be
outgoing along the ray $\{x_1=0, x_2>0\},$ hence, from~\eqref{eqn14.300}, it is
clear that the solution found by our method also will not be outgoing. This is where
the symmetric formulation of the transmission problem proves its worth. If we
let $v^{\out}_{\mu;l},$ $v^{\out}_{\mu;r},$ denote the scattered fields obtained
using the foregoing method with $q=q_l,$ $q=q_r,$ respectively, then we can let
\begin{equation}
  u^{\In}_{l,r}(x_1,x_2)=
  \begin{cases}
    &v^{\In}_{\mu}(x_1,x_2)\chi_{[d_l^+,\infty)}(x_2)+v^{\out}_{\mu;l}(x_1,x_2),\text{
        for }l, x_1<0,\\
       &v^{\In}_{\mu}(x_1,x_2)\chi_{[d_r^+,\infty)}(x_2)+v^{\out}_{\mu;r}(x_1,x_2),\text{
        for }r, x_1>0,\\
  \end{cases}
\end{equation}
where $\supp q_{l,r}=[d_{l,r}^-,d_{l,r}^+].$ These two fields have the same
incoming component where $x_2\gg 0,$ given by $v^{\In}_{\mu}.$ Hence the data
for the transmission
problem $(u^{\In}_l(0,x_2)-u^{\In}_r(0,x_2),\pa_{x_1}[u^{\In}_l(0,x_2)-u^{\In}_r(0,x_2)]),$
has no incoming part and, for $x_2\gg 0,$ is given by
\begin{equation}
  \begin{split}
     &g(x_2)=\int_{0}^{\pi}
     (\alpha^+_l(\theta)-\alpha^+_r(\theta))e^{ik_1x_2\sin\theta} d\mu(\theta),\\
      &h(x_2)=ik_1\int_{0}^{\pi}
    (\alpha^+_l(\theta)-\alpha^+_r(\theta))\cos\theta e^{ik_1x_2\sin\theta} d\mu(\theta).
  \end{split}
\end{equation}
If $d\mu=m(\theta)d\theta,$ with $m\in\cC^{\infty}_c((0,\pi)),$ then $g$ and $h$
have asymptotic expansions like those in~\eqref{eqn174.210}, but with
$a_j^-=a_j^{'-}=0,$ for all $j.$ This data is therefore outgoing, and the
solution to the scattering problem produced by our method can also expected to
be.  It is worth mentioning that if $m\in\cC^{\infty}_c((0,\pi)),$ with $\supp
m\subset [\theta_0,\pi-\theta_0],$ for a $\theta_0>0,$ then the incoming wave
packet $v_{\mu}^{\In}(r\eta)=O(r^{-N})$ for any $N>0,$ if
$0<\eta_2<\sin\theta_0.$
\begin{remark}
  I want to thank Manas Rachh for explaining the trick used here for removing the incoming
  part of the data in a transmission problem.
\end{remark}

\section{The Projections onto Wave-Guide Modes}\label{sec8.91}
In the foregoing pages we have explained a method to find and
represent solutions to the scattering problem that results from two
semi-infinite wave-guides meeting along a common perpendicular
line. The solution is represented in each half plane by layer
potentials along this line, with sources $(\sigma,\tau),$
see~\eqref{eqn160.08}. The solutions in each half plane can be split
into a contribution from the wave-guide modes and `radiation,'
\begin{equation}
  u^{l,r}(x)=u^{l,r}_g(x)+u^{l,r}_{\rad}(x).
\end{equation}
If $\{(v_n^{l,r},\xi_n^{l,r}):\:n=1,\dots,N_{l,r}\}$ are the guided modes, which
are real valued and normalized to have $L^2$-norms 1, then the projection into
the guided modes is given by
\begin{equation}\label{eqn210.94}
     u^{l,r}_g(x_1,x_2)=\sum_{n=1}^{N_{l,r}}v_{n}^{l,r}(x_2)\int_{-\infty}^{\infty}u^{l,r}(x_1,y_2)v_n^{l,r}(y_2)dy_2
    \end{equation}
This has a very simple expression in terms of our representation, coming
entirely from the $w^g_{0+}$-term in the expression for $\fE^{l,r},$ see~\eqref{eqn35.91}.

We assume that $u^{l,r}$ is given by~\eqref{eqn160.08}, with sources
$(\sigma,\tau)\in\cC_{\alpha}(\bbR)\oplus\cC_{\alpha+\frac 12}(\bbR).$ As noted
in Section~\ref{sec8} the representations for $u^{l,r}$ and $\pa_{x_1}u^{l,r}$ are
in terms of absolutely convergent integrals.  The key observation is the
fact that the wave-guide modes are orthogonal to the continuous spectrum.
\begin{proposition}\label{prop2.94}
  Let $(\sigma,\tau)\in \cC_{\alpha}(\bbR)\oplus\cC_{\alpha+\frac 12}(\bbR)$ and
  let $v_{n}^{l,r}(x_2)e^{i\xi_n^{l,r}x_1}$ be a wave-guide mode for
  $\Delta+k_1^2+q_{l,r}(x_2),$ then, for all  $\pm x_1>0,$ 
  \begin{equation}
    \begin{split}
     &\int\limits_{-\infty}^{\infty}\int\limits_{-\infty}^{\infty}[g_{k_1}(|(x_1,x_2-y_2)|)+
        w^{c;l,r}_{0+}(x_1,x_2;0,y_2)]\tau(y_2)v_n^{l,r}(x_2)dy_2dx_2=0,\\
     &\int\limits_{-\infty}^{\infty}\int\limits_{-\infty}^{\infty}\pa_{y_1}[g_{k_1}(|(x_1-y_1,x_2-y_2)|)+\\
        &\phantom{mmmmmmmmmmmm}w^{c;l,r}_{0+}(x_1,x_2;y_1,y_2)]_{y_1=0}\sigma(y_2)v_n^{l,r}(x_2)dy_2dx_2=0.
    \end{split}
  \end{equation}
\end{proposition}
\begin{proof}
  We give the details for the right half plane. The estimates proved in the
  previous sections show that these integrals are absolutely convergent and
  therefore we can change the order of the integrations. To prove the
  proposition we show that, for each $y_2$ and $x_1>0,$ we have
   \begin{equation}\label{eqn212.93} 
    \begin{split}
     &\int\limits_{-\infty}^{\infty}[g_{k_1}(|(x_1,x_2-y_2)|)+
        w^{c;r}_{0+}(x_1,x_2;0,y_2)]v_n^{r}(x_2)dx_2=0,\\
     &\int\limits_{-\infty}^{\infty}\pa_{y_1}[g_{k_1}(|(x_1-y_1,x_2-y_2)|)+
        w^{c;r}_{0+}(x_1,x_2;y_1,y_2)]_{y_1=0}v_n^{r}(x_2)dx_2=0.
    \end{split}
  \end{equation}
  To prove these statements we use the Sommerfeld integral representation for
  the free space fundamental solution, see~\eqref{eqn26.51}.

  We begin with the single layer term, which can be written as the iterated
  integral:
  \begin{equation}\label{eqn213.91}
    \frac{1}{2\pi}
    \int\limits_{-\infty}^{\infty}\left[\int\limits_{-\infty}^{\infty}\frac{ie^{i|x_2-y_2|\sqrt{k_1^2-\xi^2}}}{2\sqrt{k_1^2-\xi^2}}+
       \int\limits_{\Gamma_{\nu}^+} \tw^{c;r}(\xi,x_2;y_2)\right] e^{ix_1\xi }d\xi\,v_n^{r}(x_2)dx_2.
  \end{equation}
  We would like to change the order of integrations in this integral, which
  would be  easily justified if the integral in $\xi$ were over any finite
  interval. Note that
  \begin{equation}
    \int\limits_{k_1+1}^{\infty}\left|\frac{ie^{-|x_2-y_2|\sqrt{\xi^2-k_1^2}}e^{-iy_1\xi}}{2\sqrt{\xi^2-k_1^2}}\right|d\xi\leq
    M[1+|\log|x_2-y_2|].
  \end{equation}
  Combining this with the estimates~\eqref{eqn73.9} and~\eqref{eqn92.9}, and the fact that
  $|v_n^r(x_2)|\leq Me^{-|x_2|\sqrt{\xi_n^{r\,2}-k_1^2}}$ shows that these integrals are
  absolutely convergent and therefore we can interchange the order of the
  integrations. By analyticity we can also replace the integral in the first term with an
  integral over $\Gamma_{\nu}^+,$ to obtain
   \begin{equation}
        \frac{1}{2\pi}
        \int\limits_{\Gamma_{\nu}^+}\left[\int\limits_{-\infty}^{\infty}
          \frac{ie^{i|x_2-y_2|\sqrt{k_1^2-\xi^2}}}{2\sqrt{k_1^2-\xi^2}}+
       \tw^{c;r}(\xi,x_2;y_2)\right] e^{ix_1\xi}v_n^{r}(x_2)dx_2d\xi.
   \end{equation}
   Using the fact that $\xi^{r\, 2}_nv_n^r=(\pa_{x_2}^2+k_1^2+q_r(x_2))v_n^r$ and
   integrating by parts  it follows that
   \begin{equation}\label{eqn216.93}
     \int\limits_{-\infty}^{\infty}\left[\frac{ie^{i|x_2-y_2|\sqrt{k_1^2-\xi^2}}}{2\sqrt{k_1^2-\xi^2}}+
       \tw^{c;r}(\xi,x_2;y_2)\right] e^{ix_1\xi}v_n^{r}(x_2)dx_2=\frac{ie^{ix_1\xi}v_n^r(y_2)}{\xi_n^{r\,2}-\xi^2}.
   \end{equation}
   Hence the double integral is
   \begin{equation}
     \frac{1}{2\pi i}\int\limits_{\Gamma_{\nu}^+}\frac{ie^{ix_1\xi}v_n^r(y_2)d\xi}{\xi^2-\xi_n^{r\,2}},
   \end{equation}
   which, using Cauchy's theorem, is easily seen to vanish for $x_1>0.$

   The double layer is almost the same; the double integral in~\eqref{eqn213.91} is
   replaced by
     \begin{equation}\label{eqn218.93}
    \frac{i}{2\pi}
    \int\limits_{-\infty}^{\infty}\left[\,\int\limits_{\Gamma_{\nu}^+}
    \left[\frac{ie^{i|x_2-y_2|\sqrt{k_1^2-\xi^2}}}{2\sqrt{k_1^2-\xi^2}}+
       \tw^{c;r}(\xi,x_2;y_2)\right]\xi e^{ix_1\xi}d\xi\right] v_n^{r}(x_2)dx_2.
     \end{equation}
     The integral involving $\tw^{c;r}(\xi,x_2;y_2)$ is again easily seen to be
     absolutely convergent, but the additional factor of $\xi$ makes the other
     term more subtle. We need an estimate for this term that takes account of
     the fact that $x_1>0.$
     \begin{lemma}
       For $x_1>0, \lambda>0$ and $R>k_1$ let
       \begin{equation}
       f^{\pm}_R(x_1,\lambda)=\pm\int\limits_{\pm R}^{\pm \infty}
       \frac{e^{ix_1\xi-\lambda\sqrt{\xi^2-k_1^2}}\xi d\xi}{\sqrt{\xi^2-k_1^2}}.
       \end{equation}
       These functions satisfy the estimates
       \begin{equation}
         |f^{\pm}_R(x_1,\lambda)|\leq \frac{Me^{-\lambda\sqrt{R^2-k_1^2}}}{x_1}
       \end{equation}
     \end{lemma}
     \begin{proof}
       As $f^-_R(x_1,\lambda)=\overline{f^+_R(x_1,\lambda)}$ it suffices to do the $+$-case.
       If we let $s=\sqrt{\xi^2-k_1},$ then
       \begin{equation}
         f^+_R(x_1,\lambda)=\int\limits_{\sqrt{R^2-k_1^2}}^{\infty}e^{ix_1\sqrt{s^2+k_1^2}-\lambda s}ds.
       \end{equation}
       Noting that
       \begin{multline}\label{eqn222.93}
         \pa_s\left[\frac{\sqrt{s^2+k_1^2}}{s}e^{ix_1\sqrt{s^2+k_1^2}}\right]=\\
         ix_1e^{ix_1\sqrt{s^2+k_1^2}}-\frac{k_1^2}{s^2\sqrt{s^2+k_1^2}}e^{ix_1\sqrt{s^2+k_1^2}},
       \end{multline}
       integration by parts shows that
       \begin{multline}
         f^+_R(x_1,\lambda)=
         \frac{1}{ix_1}\Bigg[
           \frac{k_1e^{-\lambda\sqrt{R^2-k_1^2}+ix_1R}}{\sqrt{R^2-k_1^2}}
           +\\
           \int\limits_{\sqrt{R^2-k_1^2}}^{\infty}\left[\frac{\lambda\sqrt{s^2+k_1^2}}{s}+
             \frac{k_1^2}{s^2\sqrt{s^2+k_1^2}}\right]e^{ix_1\sqrt{s^2+k_1^2}-\lambda
             s}ds\Bigg].
       \end{multline}
      The integral is easily seen to be $O(e^{-\lambda\sqrt{R^2-k_1^2}}),$ which
      completes the proof of the lemma.
     \end{proof}

     We rewrite the integral in~\eqref{eqn218.93} as
     \begin{multline}
       \frac{i}{2\pi}
    \int\limits_{-\infty}^{\infty}\int\limits_{\Gamma_{\nu}^+}
    \left[\frac{ie^{i|x_2-y_2|\sqrt{k_1^2-\xi^2}}}{2\sqrt{k_1^2-\xi^2}}\right]\xi
    e^{ix_1\xi}d\xi\, v_n^{r}(x_2)dx_2=\\
    \frac{i}{2\pi}
    \int\limits_{-\infty}^{\infty}\left[\,\int\limits_{\Gamma_{\nu}^+\cap D_R}
      \left[\frac{ie^{i|x_2-y_2|\sqrt{k_1^2-\xi^2}}}{2\sqrt{k_1^2-\xi^2}}\right]\xi
      e^{ix_1\xi}d\xi+f_R^+(x_1,|x_2-y_2|)+
    f_R^-(x_1,|x_2-y_2|)\right] v_n^{r}(x_2)dx_2.
     \end{multline}
     In the part of the integral over $\Gamma_{\nu}^+\cap D_R$ we can
     interchange the order of the integrations. We use the lemma to estimate
     the other terms
     \begin{equation}
       \begin{split}
       \Bigg| \int\limits_{-\infty}^{\infty}[f_R^+(x_1,|x_2-y_2|)&+
         f_R^-(x_1,|x_2-y_2|)] v_n^{r}(x_2)dx_2\Bigg| \\
       &\leq
       \frac{M}{x_1}\int_{-\infty}^{\infty}e^{-\sqrt{R^2-k_1^2}|x_2-y_2|}e^{-\sqrt{\xi_n^{r\,2}-k_1^2}|x_2|}dx_2\\
       &\leq\frac{M}{x_1(\sqrt{R^2-k_1^2}-\sqrt{\xi_n^{r\,2}-k_1^2})}.
       \end{split}
     \end{equation}
     Thus we see that
       \begin{multline}\label{eqn227.93}
       \int\limits_{-\infty}^{\infty}\left[\,\int\limits_{\Gamma_{\nu}^+}
    \left[\frac{ie^{i|x_2-y_2|\sqrt{k_1^2-\xi^2}}}{2\sqrt{k_1^2-\xi^2}}+
       \tw^{c;r}(\xi,x_2;y_2)\right]\xi e^{ix_1\xi}d\xi\right]
    v_n^{r}(x_2)dx_2=\\
    \lim_{R\to\infty}
     \int\limits_{\Gamma_{\nu}^+\cap  D_R}\int\limits_{-\infty}^{\infty}\left[
    \left[\frac{ie^{i|x_2-y_2|\sqrt{k_1^2-\xi^2}}}{2\sqrt{k_1^2-\xi^2}}+
       \tw^{c;r}(\xi,x_2;y_2)\right]v_n^{r}(x_2)dx_2\right]\,\xi e^{ix_1\xi}
    d\xi.
       \end{multline}
       The $x_2$-integral is computed in~\eqref{eqn216.93}, giving
     \begin{multline}\label{eqn228.93}
      \int\limits_{-\infty}^{\infty}\left[\,\int\limits_{\Gamma_{\nu}^+}
    \left[\frac{ie^{i|x_2-y_2|\sqrt{k_1^2-\xi^2}}}{2\sqrt{k_1^2-\xi^2}}+
       \tw^{c;r}(\xi,x_2;y_2)\right]\xi e^{ix_1\xi}d\xi\right]
    v_n^{r}(x_2)dx_2=\\
    \lim_{R\to\infty}
     v_n^{r}(y_2)\int\limits_{\Gamma_{\nu}^+\cap  D_R}\frac{i\xi e^{ix_1\xi}}{\xi_n^{r\,2}-\xi^2} d\xi=0.   
     \end{multline}
    The last equality follows from Cauchy's theorem. It completes the proof
    of~\eqref{eqn212.93} and also of the proposition.
\end{proof}

Using this proposition we can compute the projections of $u^{l,r}$ to the
respective right and left wave-guide modes given in~\eqref{eqn210.94}. In our
representation
\begin{equation}
  u^{l,r}(x)=\int_{-\infty}^{\infty}\fE^{l,r}(x;0,y_2)\tau(y_2)dy_2-\int_{-\infty}^{\infty}\pa_{y_1}\fE^{l,r}(x;0,y_2)\sigma(y_2)dy_2.
\end{equation}
We give the details for $\{x_1>0\};$ 
Proposition~\ref{prop2.94} and~\eqref{eqn43.34}  show that
\begin{equation}\label{eqn229.94}
  u_g^r(x_1,x_2)=\sum_{n=1}^{N_r}[\langle
    \tau,v_n^r\rangle+   i\xi_n^r\langle \sigma,v_n^r\rangle]v_n^r(x_2)e^{i\xi_n^r x_1},
\end{equation}
where
\begin{equation}
  \langle f,g\rangle=\int\limits_{-\infty}^{\infty}f(x)\overline{g(x)}dx.
\end{equation}
Hence the projections of $u^{l,r}$ onto the wave-guide modes are completely
determined by the projections, $\{\langle\tau,v_n^{l,r}\rangle,\langle
\sigma,v_n^{l,r}\rangle:\: n=1,\dots,N_{l,r}\},$ of the source terms onto these modes.

If these projections could be determined directly from the data, then we could
determine the scattering relation, from incoming wave-guide modes to outgoing
modes, without having to solve the complete problem. Starting with the
equation~\eqref{eqn143.35} we can almost find equations for the coefficients
in~\eqref{eqn229.94}.  Projecting these equations into span of the wave-guide
modes we obtain
\begin{equation}
  \begin{split}
    P^{l,r}_g\sigma+ P^{l,r}_g D\tau&=P^{l,r}_gg\\
     P^{l,r}_gC\sigma+ P^{l,r}_g \tau&=P^{l,r}_gh,
  \end{split}
\end{equation}
where we let
\begin{equation}
  P^{l,r}_gf=\sum_{n=1}^{N_{l,r}}\langle f,v_n^{l,r}\rangle v_n^{l,r}(x_2).
\end{equation}

These equations can be rewritten as
\begin{equation}
  \begin{split}
    P^{l,r}_g\sigma+ P^{l,r}_g DP_g^{l,r}\tau&=P^{l,r}_gg-P^{l,r}_g D(\Id-P_g^{l,r})\tau\\
     P^{l,r}_gCP^{l,r}_g\sigma+ P^{l,r}_g \tau&=P^{l,r}_gh-P^{l,r}_gC(\Id-P^{l,r}_g)\sigma.
  \end{split}
\end{equation}
While this is not quite a system of equations for the projections
$(P^{l,r}_g\sigma,P^{l,r}_g\tau),$ the facts that $P^{l,r}_g(\Id-P^{l,r}_g)=0,$
and the norms of $\|D\|$ and $\|C\|$ are proportional to $m(q_l,q_r)$ suggests
that, at least for two channels with small contrast,  dropping these terms 
leads to equations
\begin{equation}
  \begin{split}
    P^{l,r}_g\tsigma+ P^{l,r}_g DP_g^{l,r}\ttau&=P^{l,r}_gg\\
     P^{l,r}_gCP^{l,r}_g\tsigma+ P^{l,r}_g \ttau&=P^{l,r}_gh,
  \end{split}
\end{equation}
whose solutions, $(P^{l,r}_g\tsigma,P_g^{l,r}\ttau),$  should be very close to
$(P^{l,r}_g\sigma,P_g^{l,r}\tau).$

\section{Some Concluding Remarks}
In the foregoing pages we have constructed outgoing fundamental
solutions for operators of the form $\Delta+k_1^2+q(x_2),$ and shown
how to use them to represent the solution to the scattering problem
defined by two semi-infinite wave-guides meeting along a common
perpendicular line. The construction of the `outgoing' fundamental
solution is in a form that lends itself to numerical
implementation, see~\cite{EpGo2024,EGHQR2025}. We have shown that the resultant system of integral
equations is Fredholm of index zero on the spaces
$\cC_{\alpha}(\bbR)\oplus\cC_{\alpha+\frac 12}(\bbR),$ with
$0<\alpha<\frac 12,$ and are therefore generically (w.r.t. $k_1$) solvable.

We have only presented a detailed analysis of this problem for the case of
potentials given by~\eqref{eqn2.32}, though it is clear that our approach will
apply, {\em mutatis mutandis}, if $q(x_2)$ is a bounded, measurable function
with bounded support. The principal difference will be that the basic solutions,
$\tu_{\pm}(\xi,0^+;x_2),$ of the ODE, and their Wronskian no longer have
explicit formul{\ae} in terms of elementary functions within the support of $q.$
These formul{\ae} need to be replaced by (standard) estimates. To implement the
method numerically, the functions $\tu_{\pm}$ have to be computed numerically
within the support of $q.$ This is done in~\cite{EGHQR2025}. The problem of
having 2 open wave-guides that are of the form considered here outside a compact
set, is a relatively compact perturbation which, while requiring further
analysis, should not pose serious additional difficulties.

In Part II we show that under reasonable hypotheses on the data, which are satisfied by
wave-guide modes, points sources and wave-packets, the sources found by solving
the integral equations along $\{x_1=0\}$ satisfy many additional estimates and
even admit asymptotic expansions, that is
\begin{equation}
  \begin{split}
   &\sigma(x_2)\sim\frac{e^{ik_1|x_2|}}{|x_2|^{\frac
       12}}\sum_{l=0}^N\frac{a^{\pm}_l}{|x_2|^l}+O\left(|x_2|^{-N-\frac
      32}\right),\\
    &\tau(x_2)\sim\frac{e^{ik_1|x_2|}}{|x_2|^{\frac 32}}\sum_{l=0}^N\left[\frac{b^{\pm}_l}{|x_2|^{l}}\right]+O\left(|x_2|^{-(N+5/2)}\right), \text{ as }|x_2|\to\infty.
    \end{split}
\end{equation}

Using these asymptotic expansions we show that the solutions given by $u^{l,r}$
also have complete expansions that are uniformly correct as $\eta_1,\, \eta_2\to
0^{\pm}.$ The existence of these expansions implies that the solutions satisfy
precisely the sort of outgoing radiation condition that one expects from the
work of Isozaki, Melrose, Vasy et al.  The proofs of the asymptotic expansions
use fairly classical techniques, combined with a novel contour deformation
argument. To complete this analysis, and prove the uniqueness of the solutions
found using our method, requires a much more sophisticated, microlocal analysis
of this class of problems, which is given in Part III.
\newpage
\appendix

\centerline{\LARGE\bf Appendix}
\bigskip

\noindent
In these appendices we collect a variety of background results, and study the
construction of the limiting absorption solution in the case of a bi-infinite
channel in $\bbR^2,$ and prove Theorem~\ref{thm00}.

 \section{The Bi-infinite Case}\label{sec1}
  In order to estimate the correction terms $w^{l,r}(x;y)$ we need to
  have a good description of the limit of the kernels for
  $(D_q+i\delta)^{-1}$ as $\delta\to 0^+,$ where
  \begin{equation}\label{eqn0}
    D_q=\Delta+k_1^2+q(x_2),
  \end{equation}
  acts on $H^2(\bbR^2).$ We are employing the limiting absorption principle
  limit, which, gives the outgoing solution to
  \begin{equation}\label{eqn0.0}
    (\Delta+k_1^2+q)u=f,
  \end{equation}
  for certain functions $f,$ which includes, but is not limited to compactly
  supported functions.
  
  In this section we use the Fourier transform in the $x_1$-variable and basic ODE theory to construct
  the kernels for these operators   where we usually take
  \begin{equation}\label{eqn172.40}
  q(x_2)=(k_2^2-k_1^2)\chi_{[-d,d]}(x_2).
  \end{equation}  
 It would be more standard to consider the spectral theory and
 resolvent of the operator $\Delta+q(x_2).$ However our analysis
 relies on detailed analyticity properties of the kernel of
 $(\Delta+q+k_1^2+i\delta)$ for $\delta>0,$ which is why we consider
 the shifted operator. The substance of these results generalizes
 easily to piecewise continuous functions, $q(x_2),$ with support in
 $[-d,d].$

With $H^2(\bbR^2)$ as the domain, $D_q$ defines an unbounded self
adjoint operator on $L^2(\bbR^2).$ The spectrum of this operator,
$\sigma(D_{q}),$ is well known to lie in the interval
$(-\infty,k_2^2].$ In this section we present a construction for the resolvent
    kernel of this operator, which allows for the construction of the
    perturbation terms $w^{l,r},$ by taking for $q$ either $q_l$ or
    $q_r.$ To that end we need compute the kernel of the limit
    $$\lim_{\delta\to 0^+}(D_{q}+i\delta)^{-1},$$
which we denote by $\cR_{0+}(x;y).$ By a small abuse of terminology, in the
    sequel we refer to $\cR_{0+}$ as the \emph{resolvent kernel}, or outgoing
    resolvent kernel.

Our construction of the resolvent kernel uses the partial Fourier transform in
the $x_1$-variable, which we denote by
\begin{equation}
  \tu(\xi,x_2)=\int_{-\infty}^{\infty}u(x_1,x_2)e^{-i\xi x_1}dx_1.
\end{equation}
 To construct the resolvent kernel we  the need
 kernels for inverses of the operators
 $$L_{\xi}+i\delta=
 \pa^2_{x_2}+k_1^2+q(x_2)-\xi^2+i\delta,\,\xi\in\bbR,$$
 with domain $H^2(\bbR).$ 
 These kernels are constructed from the \emph{basic solutions} to
\begin{equation}\label{eqn209.81}
  \pa^2_{x_2}\tu_{\pm}(\xi,\delta;x_2)+(k_1^2+q(x_2)-\xi^2+i\delta)
  \tu_{\pm}(\xi,\delta;x_2)=0,
\end{equation}
which satisfy
\begin{equation}\label{eqn7}
  \tu_{\pm}(\xi,\delta;x_2)=e^{\pm ix_2\sqrt{k^2_1-\xi^2+i\delta}}\text{
    for }\pm x_2>d.
\end{equation}
The  $\sqrt{z}$ is defined on $\bbC\setminus (-\infty,0],$ to be positive
on $(0,\infty);$  for $\delta,\xi\in\bbR,$ 
$$\Sgn\Im\sqrt{k_1^2-\xi^2+i\delta}=\Sgn\delta,$$
and therefore, for $\delta>0,$ and $\xi^2>k_1^2,$
\begin{equation}\label{eqn176.55}
  \sqrt{k_1^2-\xi^2+i\delta}=i\sqrt{\xi^2-k_1^2-i\delta}.
\end{equation}
Hence, for $\delta>0,$   $\tu_{\pm}(\xi,\delta;x_2)$ decays
exponentially as $\pm x_2\to\infty.$ 
Taking $\delta\to 0^+$ we get the basic solutions that
satisfy, as $\pm x_2\to\infty,$
\begin{equation}
  \tu_{\pm}(\xi,0^+;x_2)=\begin{cases} &e^{\pm i
    x_2\sqrt{k_1^2-\xi^2}}\text{ for }|\xi|\leq k_1,\\
  &e^{\mp
    x_2\sqrt{\xi^2-k_1^2}}\text{ for }|\xi|> k_1.
  \end{cases}
\end{equation}
The solutions $\tu_+(\xi,0^+;x_2)$ are outgoing as $x_2\to\infty,$ and
$\tu_-(\xi,0^+;x_2)$ are outgoing as $x_2\to-\infty.$  If $q(-x_2)=q(x_2),$ then it
is easy to show that 
\begin{equation}\label{eqn8}
  \tu_{-}(\xi,\delta;x_2)=\tu_{+}(\xi,\delta;-x_2).
\end{equation}

\begin{remark}\label{rmk1}
  It should be noted that while the families
  $\tu_{\pm}(\xi,\delta;x_2)$ are analytic as functions of $\xi^2,$ except at $\xi^2=k_1^2,$
  their analyticity properties as functions of $\xi$ are more
  complicated. This is because, with our choice of square-root, the
  composition $\xi\mapsto\sqrt{k_1^2-\xi^2}$ is analytic and
  single-valued in $\bbC\setminus (-\infty,-k_1]\cup [k_1,\infty).$
      We see that letting $\delta\to 0^+$ implies that, for $\xi$ real
   with for $|\xi|>k_1,$
   $$\sqrt{k_1^2-\xi^2}=i\sqrt{\xi^2-k_1^2},$$ where
   $\sqrt{\xi^2-k_1^2}\in(0,\infty).$ Fortunately our applications
   only require analyticity for $\xi$ in small neighborhoods of the
   intervals $(k_1,k_2),(-k_2,-k_1)\subset\bbR.$ Since this avoids the
   branch points at $\pm k_1,$ the $\sqrt{\xi^2-k_1^2}$ has a single
   valued, analytic determination in a neighborhood of these open
   intervals, which is positive for $k_1<|\xi|<k_2.$

   As solutions to an ODE,
   the functions $\tu_{\pm}(\xi,0^+;x_2)$ are specified by their
   behavior for $\pm x_2>d,$ and therefore have analytic extensions to
   $\tu_{\pm}(\zeta,0^+;x_2),$ for $\zeta$ in a neighborhood,
   $\cU\subset\bbC,$ of $(k_1,k_2)\cup (-k_2,-k_1).$ As these
   solutions are also determined by their asymptotics, for small enough
   $\delta,$ these extensions satisfy
   $$\tu_{\pm}(\xi,\delta;x_2)=\tu_{\pm}(\sqrt{\xi^2-i\delta},0^+;x_2).$$
   Note also that
   $\tu_{\pm}(\zeta,0^+;x_2)=\tu_{\pm}(-\zeta,0^+;x_2),$ for $\zeta\in
   \cU.$ This does not require $q(x_2)=q(-x_2).$

   The simple case of the $\Delta+k^2$ is instructive. Constructing
   the outgoing solution to $(\Delta+k^2)u=f,$ via a 1-dimensional
   Fourier transform gives the formula
   \begin{equation}
     u(x_1,x_2)=-\frac{i}{4\pi}\int_{-\infty}^{\infty}\int_{-\infty}^{\infty}
     \frac{e^{i\sqrt{k^2-\xi^2}|x_2-y_2|}e^{i\xi x_1}\tf(\xi,y_2)dy_2
       d\xi}{\sqrt{k^2-\xi^2-k^2}},
   \end{equation}
   where $\sqrt{k^2-\xi^2}=i\sqrt{\xi^2-k^2}$ if $|\xi|\geq k.$ The
   resolvent kernel of $\pa_{x_2}^2-\xi^2+k^2$ is an even function of $\xi,$
   and has an analytic extension to a neighborhood of
   $(-\infty,-k)\cup (k,\infty).$ This is {\em not} the restriction of its
   analytic extension to the upper, or lower half plane.
\end{remark}

In order to satisfy the equation and boundary conditions at $x_2=\pm
d,$ implied by $\Dom(L_{\xi})=H^2(\bbR),$ it is necessary for $\tu_+$
to have the form
\begin{equation}\label{eqn9}
  \tu_+(\xi,\delta;x_2)=
  \begin{cases}
    &a_0e^{ ix_2\sqrt{k_2^2-\xi^2+i\delta}}+b_0e^{
      -ix_2\sqrt{k^2_2-\xi^2+i\delta}}\text{ for }|x_2|<d,\\
    &a_-e^{ ix_2\sqrt{k_1^2-\xi^2+i\delta}}+b_-e^{
      -ix_2\sqrt{k^2_1-\xi^2+i\delta}}\text{ for }x_2<-d.
  \end{cases}
\end{equation}
The simple form of $\tu_+(\xi,\delta;x_2)$ for $|x_2|<d$ assumes that
$q$ is given by~\eqref{eqn172.40}, which we assume for the remainder
of this section.  If $\delta=0^+,$ and $|\xi|<k_1,$ then $\tu_+$
oscillates as $|x_2|\to \infty,$ whereas if $|\xi|>k_1,$ then this
solution decays exponentially as $x_2\to\infty,$ but typically grows
exponentially as $x_2\to-\infty.$

Using the solutions described in~\eqref{eqn7}--\eqref{eqn9}, we now
construct the inverse for $L_{\xi}+i\delta.$  The  inverse is given by
\begin{multline}\label{eqn18.02}
  R_{\xi,\delta}(\tf)(x_2)=\frac{1}{W(\xi,\delta)}
  \Bigg[\int_{-\infty}^{x_2}\tu_+(\xi,\delta;x_2)\tu_-(\xi,\delta;y_2)\tf(y_2)dy_2+\\
    \int_{x_2}^{\infty}\tu_-(\xi,\delta;x_2)\tu_+(\xi,\delta;y_2)\tf(y_2)dy_2\Bigg],
\end{multline}
where
\begin{equation}\label{eqn48.51}
  W(\xi,\delta)=u_-(\xi,\delta;x_2)\pa_{x_2}u_+(\xi,\delta;x_2)-u_+(\xi,\delta;x_2)\pa_{x_2}u_-(\xi,\delta;x_2),
\end{equation}
is the Wronskian, which is independent of $x_2.$ This operator agrees
with the bounded inverse of $L_{\xi}+i\delta,$ where it is defined,
but is also defined as an operator from $L^2_{\comp}(\bbR)\to
H^2_{\loc}(\bbR),$ even for $\delta=0^+,|\xi|< k_1.$ We denote the
operator by $R_{\xi,\delta},$  with
\begin{equation}
  R_{\xi,0^+}=\lim_{\delta\to 0^+}R_{\xi,\delta}.
\end{equation}

The operator $\pa_{x_2}^2+k_1^2+q(x_2),$ acting on
$H^2(\bbR)$ is self adjoint and its spectrum is
easily shown lie in the interval $(-\infty,k_2^2].$ Thus
  $L_{\xi}+i\delta$ is invertible on
$L^2(\bbR)$ provided $\delta\neq 0,$ and also if $\delta=0,$ but
$|\xi|>k_2.$ Indeed it is also invertible with $\delta =0,$ for all
but finitely many $\xi$ with $k_1<|\xi|<k_2.$ Because it is self adjoint we
have the norm estimate for
$(L_{\xi}+i\delta)^{-1}=R_{\xi,\delta}$
\begin{equation}\label{eqn28.20}
  \|R_{\xi,\delta}\|\leq
  \frac{1}{\dist(\xi^2-i\delta,(-\infty,k_2^2])}\leq
    \begin{cases}&\frac{1}{|\delta|}\text{ for }|\xi|<k_2,\\
      &\frac{1}{\sqrt{(\xi^2-k_2^2)^2+\delta^2}}\text{ for }|\xi|\geq
      k_2.
      \end{cases}
\end{equation}

If $f\in L^2(\bbR^2),$ then
the $L^2$-solution to $(D_{q}+i\delta)u=f$ is given by
\begin{equation}\label{eqn12}
  u(x_1,x_2)=\frac{1}{2\pi}\int_{-\infty}^{\infty}e^{i\xi x_1}R_{\xi,\delta}(\tf)(x_2)d\xi,
\end{equation}
where
\begin{equation}
  \tf(\xi,x_2)=\int_{-\infty}^{\infty}f(x_1,x_2)e^{-i\xi x_1}dx_1.
\end{equation}
Using Plancherel's formula and~\eqref{eqn28.20} we see that, for $\delta\neq 0,$
  \begin{equation}
    \int_{\bbR^2}|u(x_1,x_2)|^2dx\leq
    \frac{1}{2\pi}\int_{\bbR^2}\frac{|\tf(\xi,x_2)|^2}{\delta^2+(\xi^2-k_2^2)_+^2}dx_2d\xi
    \leq \frac{1}{\delta^2}\int_{\bbR^2}|f(x_1,x_2)|^2dx.
  \end{equation}
 If $q=0,$ then the limit of~\eqref{eqn12} as $\delta\to 0^+,$ is just
 Sommerfeld's integral expressing the outgoing fundamental solution to
 $\Delta+k_1^2$ as a Fourier transform in the $x_1$-variable.


Relation~\eqref{eqn8} shows that
\begin{equation}\label{eqn55.52}
  W(\xi,\delta)=2u_+(\xi,\delta;0)\pa_{x_2}u_+(\xi,\delta;0);
\end{equation}
from~\eqref{eqn9} we conclude that
\begin{equation}
   W(\xi,\delta)=2i (a_0^2-b_0^2)\sqrt{k^2_2-\xi^2+i\delta}.
\end{equation}
It follows from Remark~\ref{rmk1} and~\eqref{eqn55.52} that $W(\xi,0^+)=\lim_{\delta\to
  0^+}W(\xi,\delta)$ has an analytic extension to the open set,
$\cU;$  for $\pm\sqrt{\xi^2-i\delta}\in\cU,$
\begin{equation}\label{eqn225.81}
  W(\xi,\delta)= W(\sqrt{\xi^2-i\delta},0^+).
  \end{equation}

Suppose that $W(\xi_0,\delta)=0;$ this happens if and only if
\begin{equation}\label{eqn16}
  \tu_+(\xi_0,\delta;x_2)=c\tu_-(\xi_0,\delta;x_2)=c\tu_+(\xi_0,\delta;-x_2),
\end{equation}
for some non-zero constant $c.$
For $x_2<0,$ this shows
\begin{equation}
  \tu_-(\xi_0,\delta;x_2)=e^{-ix_2\sqrt{k_1^2-\xi_0^2+i\delta}}.
\end{equation}
If $\delta> 0,$ then $\Im \sqrt{k_1^2-\xi_0^2+i\delta}>0,$ and
therefore $\tu_+(\xi_0,\delta;x_2)\in L^2(\bbR),$ would be an
$L^2$-eigenvector with eigenvalue $k_1^2-\xi_0^2+i\delta.$ As
$\pa_{x_2}^2+k_1^2+q(x_2)$ is self-adjoint, its spectrum is real,
and therefore such roots cannot exist. If a relation like~\eqref{eqn16} holds, then $\delta=0.$

If $k_1^2-\xi_0^2<0,$ then
\begin{equation}
  \Im \sqrt{k_1^2-\xi_0^2}=\lim_{\delta\to 0^+}\Im \sqrt{k_1^2-\xi_0^2+i\delta}>0.
\end{equation}
hence
\begin{equation}
  \tu_-(\xi_0,0^+;x_2)=e^{x_2\sqrt{\xi_0^2-k_1^2}},
\end{equation}
which decays exponentially as $x_2\to-\infty.$  The
function, $\tu_+(\xi_0,0^+;x_2)$ is an $L^2$-eigenfunction of
$\pa_{x_2}^2+k_1^2+q(x_2)$ with eigenvalue $\xi_0^2.$

\begin{definition}
 A \emph{wave-guide} solution is a function $v\in H^2(\bbR)$ that is a solution to
\begin{equation}
  \pa_{x_2}^2v-\xi^2v+(k_1^2+q(x_2))v=0,
\end{equation}
which satisfies the estimate
\begin{equation}
  |v(x_2)|\leq Ce^{-\sqrt{\xi^2-k_1^2}|x_2|}\text{ for }|x_2|>d.
\end{equation}
\end{definition}
The functions $e^{\pm i\xi x_1}v(x_2)$ are solutions to the homogeneous equation
$D_{q}u=0.$ These solutions are strongly localized within the channel $|x_2|<d.$
The solution with $0<\xi$ is a right-ward moving wave, and that with $\xi<0$ is
left-ward moving.

In
Appendix~\ref{app0} we prove the following theorem.
\begin{theorem}\label{thm0}
  If $0<k_1<k_2,$ for $\xi\in\bbR,$ and $\delta\ge 0,$
  there are finitely many simple solutions, $\{\pm \xi_n:\:n=1,\dots, N\},$ to the equation
  \begin{equation}\label{eqn37.20}
    W(\xi,\delta)=0,
  \end{equation}
  all of which satisfy $\delta=0^+,$ and
  \begin{equation}
   k_1<|\xi_n|<k_2.
  \end{equation}
  These are the only solutions to equation~\eqref{eqn37.20}. For 
  $0<k_1<k_2,$ and any $d>0,$ there is at least one non-trivial solution to $W(\xi,0^+)=0.$
\end{theorem}
\noindent
The number of solutions, $N,$ is a non-decreasing function of $d,$ the
width of the channel.
\begin{remark}
  As it reduces to the study of 2nd order ODEs with compactly
  supported potentials, the analysis of guided modes in $2d$ is quite
  classical, and this result is well known. The analysis in the higher
  dimensional case is more recent. See~\cite{JolyPoirier} and the
  references therein.
\end{remark}

\begin{remark}\label{rmk5.9}
  In the sequel we assume that $0<k_1<\xi_n<k_2,$ for $n=1,\dots,N.$  We can show that
  \begin{equation}
    \lim_{\delta\to 0^+}(\xi-\xi_n)R_{\xi_n,\delta}\tf(x_2)=
    \frac{\tu_{+}(\xi_n,0^+;x_2)\langle \tu_{+}(\xi_n,0^+;\cdot),\tf\rangle}{c W_{\xi}(\xi_n)},
  \end{equation}
  where $c$ is defined in~\eqref{eqn16}. Following \S 2.6
  of~\cite{Titchmarsh1946_ODEs}, we conclude that
  \begin{equation}
    v_n(x_2)=\frac{ \tu_{+}(\xi_n,0^+;\cdot)}{\sqrt{|c W_{\xi}(\xi_n)|}}
  \end{equation}
  has $L^2$-norm 1.
\end{remark}

Our main interest is in analyzing the behavior of solutions to
$(D_{q}+i\delta) u_{\delta}=f$ as $\delta\to 0^+.$ In general, the
limiting solution does not belong to $L^2(\bbR^2)$ and will not exist
unless $f$ satisfies certain conditions.  For $\delta\neq 0,$
$u_{\delta}\in H^2(\bbR^2),$ and the PDE, $(D_q+i\delta)u_{\delta}=f,$
is equivalent to
 \begin{equation}
  (L_{\xi}+i\delta)\tu_{\delta}=\tf,\text{ for }\xi\in\bbR,
 \end{equation}
 which implies that
 \begin{equation}
   \tu_{\delta}(\xi,\cdot)=R_{\xi,\delta}\tf(\xi,\cdot).
 \end{equation}
 We  consider what happens where the Wronskian
vanishes.  Using~\eqref{eqn16}, we see that at such a root the
limit of $W(\xi,\delta)R_{\xi,\delta}\tf$ satisfies
\begin{equation}
  \lim_{\delta\to 0^+}W(\xi_0,\delta)R_{\xi_0,\delta}(\tf)=\frac{1}{c}
  \int_{-\infty}^{\infty}\tu_+(\xi_0,0^+;x_2)\tu_+(\xi_0,0^+;y_2)\tf(\xi_0,y_2)dy_2.
  \end{equation}
From this relation it is clear that in order for $\lim_{\delta\to
  0^+}R_{\xi_0,\delta}\tf(\xi_0,\cdot)$ to exist, where $W(\xi_0,0^+)=0,$ it would be
necessary for
\begin{equation}\label{eqn28}
  \int_{-\infty}^{\infty}\tu_+(\xi_0,0^+;y_2)\tf(\xi_0,y_2)dy_2=0.
\end{equation}
In fact, we are only interested in the inverse Fourier transform,
\begin{equation}\label{eqn173.31}
\lim_{\delta\to 0^+} \frac{1}{2\pi} \int_{-\infty}^{\infty}R_{\xi,\delta}(\tf)(\xi,x_2)e^{i\xi x_1}d\xi,
\end{equation}
which may well have a limit even if
$\lim_{\delta\to
  0^+}R_{\xi,\delta}(\tf)(\xi,x_2)$ does not exist for all real $\xi.$
This is because for data analytic in $\xi$ we can deform the contour
and avoid the singularities at $\{\pm\xi_n\}.$
\begin{proposition}\label{prop5.230}
 Suppose that $\pm k_1$ are not  roots of $W(\xi,0^+).$  If $f$ satisfies
 the following properties:
 \begin{enumerate}
   \item The $\supp f\subset \bbR\times [-L,L]$ for some finite $L.$
 \item For each $x_2\in[-L,L],$ the distributional partial Fourier
   transform of $f$ in the $x_1$-variable, $\tf(\cdot,x_2),$ is in
   $L^1_{\loc}(\bbR),$ and in $L^2(\bbR\setminus [-k_2,k_2]).$
   \item The function $\xi\to\tf(\xi,x_2)$ has an analytic extension
     to $\cU,$ a complex neighborhood of $(-k_2,-k_1)\cup (k_1,k_2).$
     \item $\tf(\xi,x_2)\in L^1_{\loc}(\bbR^2).$
 \end{enumerate}
In this case the limit
\begin{equation}\label{eqn43.23}
v_{0+}(x_1,x_2)=  \lim_{\delta\to 0^+}
\frac{1}{2\pi}   \int_{-\infty}^{\infty}R_{\xi,\delta}(\tf)(\xi,x_2)e^{i\xi x_1}d\xi
\end{equation}
exists in $H^2_{\loc}(\bbR^2)$ and defines a solution to $D_{q}v_{0+}=f.$
\end{proposition}

\begin{proof}
  For $\delta>0,$ the integral on the right hand side
  of~\eqref{eqn43.23}, which we denote $v_{\delta},$ satisfies
  \begin{equation}\label{eqn75.51}
    (D_q+i\delta)v_{\delta}=f.
  \end{equation}
Under the hypotheses on $f$  we can express the Fourier transform in
$x_1$ as
\begin{equation}
  R_{\xi,\delta}[\tf(\xi,\cdot)](x_2)=
  \frac{U_{\xi,\delta}[\tf(\xi,\cdot)](x_2)}
{W(\xi,\delta)}=
\frac{U_{\sqrt{\xi^2-i\delta},0^+}[\tf(\xi,\cdot)](x_2)}
{W(\sqrt{\xi^2-i\delta},0^+)},
\end{equation}
for $\xi\in\cU.$ In light of Remark~\ref{rmk1}, and the fact that the
$y_2$-integrals defining the numerator, $U_{\sqrt{\xi^2-i\delta},0^+}[\tf(\xi,\cdot)](x_2),$
extend over a compact interval, this function has an analytic extension as a
function of $\xi$ to the neighborhood, $\cU,$ of $(-k_2,-k_1)\cup (k_1,k_2).$
The denominator also has an analytic extension, with simple zeroes at
$\{\pm\sqrt{\xi^2_n+i\delta}:\:n=1,\dots,N\}.$ Those with $+$-signs lie in the
upper half plane near to $\{|\xi_n|\};$ those with $-$-signs lie in the lower
half plane near to $\{-|\xi_n|\}.$

Let $\nu>0$ be a small number so that, for all small enough $\delta,$
tending to $0^+,$ the numerator,
$U_{\xi,\delta}[\tf(\xi,\cdot)](x_2),$ is analytic in $ \cB_{2\nu},$ where
\begin{equation}
 \cB_{\epsilon}= \bigcup_{n=1}^ND_{\epsilon}(\xi_n)\cup
 D_{\epsilon}(-\xi_n)\text{ for }0<\epsilon.
\end{equation}
Furthermore, assume that $\nu>0$ is less than $1/4$ the minimum
distance between successive values of $\{\pm\xi_n\}\cup\{\pm k_1,\pm
k_2\}.$ Let $\Gamma^+_{\nu}$ be the contour, which lies along the real
axis, except for semi-circles, in the upper half plane, of radius
$\nu$ centered on the zeros (both positive and negative) of
$W(\xi,0^+).$ See Figure~\ref{fig2}.  For $\delta>0,$ the numerator of
the integrand is analytic in the region between the real axis and
$\Gamma^+_{\nu},$ and the denominator has simple zeroes at
$\{\pm\sqrt{\xi^2_n+i\delta}:\:n=1,\dots,N\},$ which, for small enough
$\delta,$ lie in $\cB_{\nu}.$

For small enough $\delta>0,$ we can therefore replace the integration along $\bbR$
in~\eqref{eqn43.23}  with an integral over the contour
$\Gamma_{\nu}^+.$ The residue theorem implies that 
\begin{equation}\label{eqn47.24}
  \begin{split}
v_{\delta}(x_1,x_2)&= \frac{1}{2\pi}
\int_{-\infty}^{\infty}R_{\xi,\delta}(\tf)(\xi,x_2)e^{i\xi x_1}d\xi
=\frac{1}{2\pi}
\int_{\Gamma^+_{\nu}}R_{\xi,\delta}(\tf)(\xi,x_2)e^{i\xi x_1}d\xi  +\\
&i\sum_{n=1}^{N}\frac{U_{\xi_n,0^+}(\tf)
  (\sqrt{\xi_n^2+i\delta},x_2)e^{i\sqrt{\xi_n^2+i\delta}\, x_1}}{
  W_{\xi}(\xi_n,0^+)}.
\end{split}
\end{equation}
From this expression it is quite clear that we can let $\delta\to
0^+,$ and $v_{\delta}$ and its derivatives converge locally uniformly
to $v_{0+}=\cR_{0+}f$ and its derivatives, with
\begin{equation}\label{eqn48.24}
 v_{0+}(x_1,x_2)=\frac{1}{2\pi}
\int_{\Gamma^+_\nu}R_{\xi,0^+}(\tf)(\xi,x_2)e^{i\xi x_1}d\xi  +
i\sum_{n=1}^{N}  v_n(x_2)\langle \tf(\xi_n,\cdot), v_n\rangle e^{i\xi_n  x_1},
\end{equation}
where the $\{v_n\}$ are defined in Remark~\ref{rmk5.9}. 
Taking the limit in~\eqref{eqn75.51} in the weak sense we deduce that
$D_{q}v_{0+}=f,$ weakly.  The fact that $v_{0+}\in H^2_{\loc}$ follows
easily from~\eqref{eqn48.24}, and standard estimates. By elliptic regularity it follows that
$D_{q}v_{0+}=f$ holds in the $H^2_{\loc}$--sense.
\end{proof}

As $\Gamma^+_{\nu}$ lies in the closed upper half plane, $e^{i\xi x_1}$ is
exponentially decaying along the semi-circles as $x_1\to\infty.$ Hence this
gives the correct asymptotics as $x_1\to\infty,$ showing that the guided modes
are outgoing to the right. To get the asymptotics as $x_1\to -\infty$ we need to
replace the contour with $\Gamma_{\nu}^-,$ its reflection across the real axis
into the lower half plane. Using the analyticity properties of the integrand in
the set $\cB_{\nu},$ we see that this replaces the sum in~\eqref{eqn48.24} with
\begin{equation}
  -i\sum_{n=1}^{N}  v_n(x_2)\langle \tf(-\xi_n,\cdot), v_n\rangle e^{-i\xi_n
      x_1},
\end{equation}
which shows that the guided wave contributions are also outgoing to the left.

  \subsection{Proof of Theorem~\ref{thm0}}\label{app0}
  In this appendix we prove:
\begin{theorem*}
  If $0<k_1<k_2,$ for $\xi\in\bbR,$ and $\delta\ge 0,$
  there are finitely many simple solutions, $\{\pm \xi_n:\:n=1,\dots, N\},$ to the equation
  \begin{equation}\label{eqn37.201}
    W(\xi,\delta)=0,
  \end{equation}
  all of which satisfy $\delta=0^+,$ and
  \begin{equation}
   k_1<|\xi_n|<k_2.
  \end{equation}
  These are the only solutions to equation~\eqref{eqn37.201}. For 
  $0<k_1<k_2,$ and any $d>0,$ there is at least one non-trivial solution to $W(\xi,0^+)=0.$
\end{theorem*}

\begin{proof}
From the discussion preceding the statement of the theorem in the
previous section, it is clear that we only need to consider the case
that $k_1^2<\xi^2<k_2^2;$ we let
\begin{equation}
     A=\sqrt{k_2^2-\xi^2},\tB=\sqrt{\xi^2-k^2_1}.
\end{equation}
Both $A$ and $\tB$ are positive real numbers.  To find $a_0,
b_0$ we need to solve the linear system:
\begin{equation}
  \left(\begin{matrix}
    e^{iAd}&e^{-iAd}\\iAe^{iAd}&-iAe^{-iAd}\end{matrix}\right)\left(\begin{matrix}
    a_0\\ b_0\end{matrix}\right)=
  \left(\begin{matrix} e^{-\tB d}\\-\tB e^{-\tB d}\end{matrix}\right).
\end{equation}
Solving we see that
\begin{equation}\label{eqn182.32}
  \left(\begin{matrix}
        a_0\\ b_0\end{matrix}\right)=\frac{ie^{-\tB d}}{2A}\left(\begin{matrix}
    (A+i\tB)e^{-iAd}\\ (A-i\tB)e^{iAd}\end{matrix}\right).
\end{equation}
Since $W(\xi,0^+)=2iA(a_0^2-b_0^2);$ we have
\begin{equation}\label{eqn184.31}
  2iA(a_0^2-b_0^2)=2ie^{-2\tB d}[(A^2-\tB^2)\frac{\sin 2dA}{2A}-\tB\cos 2dA].
\end{equation}
But for the points $\xi=\pm k_2,\pm\sqrt{(k_1^2+k_2^2)/2}\,$ the equation
$a_0^2-b_0^2=0$ is equivalent to
\begin{equation}\label{eqn22}
  \tan 2
  d\sqrt{k^2_2-\xi^2}=\frac{2\sqrt{(k^2_2-\xi^2)(\xi^2-k^2_1)}}
    {k^2_1+k^2_2-2\xi^2}.
\end{equation}
Since this is an analytic equation (at least on a Riemann surface covering
$\bbC$) the set of solutions is discrete. Indeed, for a given $d,$ there is a
finite set of points
$$\{k_1\leq \xi_j\leq k_2:\: j=1,\dots,N\},$$
which solve this equation. Clearly $\{-\xi_j\}$ are also solutions.

The right hand side of~\eqref{eqn22} goes from $0$ to $+\infty$ for
$\xi\in[k_1,\sqrt{(k_1^2+k_2^2)/2}),$ and from $-\infty$ to $0$ 
  $\xi \in(\sqrt{(k_1^2+k_2^2)/2},k_2].$ The left hand side of~\eqref{eqn22} vanishes
at $\xi=k_2,$ and therefore, no matter how small $d\sqrt{k_2^2-k_1^2}$ is, the
graph of the left hand side crosses that of the right hand side for some value
of $\xi\in (k_1,k_2).$ See Figure~\ref{fig1}(a).

If
\begin{equation}
  d\sqrt{2(k_2^2-k_1^2)}=\frac{\pi}{2}+(n-1)\pi,\text{ for some }n\in\bbN,
\end{equation}
then $\xi=\pm\sqrt{(k_1^2+k_2^2)/2}$ are roots of $A(a_0^2-b_0^2)=0.$ It is easy to
see that $\xi=\pm k_2$ are never solutions to $A(a_0^2-b_0^2)=0.$ It is possible
that $\xi=\pm k_1$ are roots, but only if
\begin{equation}
  2d\sqrt{k_2^2-k_1^2}=n\pi\text{ for some }n\in\bbN.
\end{equation}

That~\eqref{eqn22} has no positive solutions less than $k_1$ follows from the
observation that, for $|\xi|<k_1,$ the right hand side is a non-zero imaginary
number, and the left hand side is real. For $\xi>k_2,$ the left and right hand
sides of~\eqref{eqn22} are purely imaginary numbers with opposite signs. Hence,
all roots of $W(\xi)$ lie in the intervals $(-k_2,-k_1]\cup [k_1,k_2).$

\end{proof}

Figure~\ref{fig1} shows plots of the two sides of~\eqref{eqn22} (left in blue,
right in red) on a single
graph, with $k_1=16, k_2=18$ and $d=0.05$ and $0.8.$ 

\begin{figure}
  \centering
  \begin{subfigure}[t]{.45\textwidth}
    \includegraphics[width= 6cm]{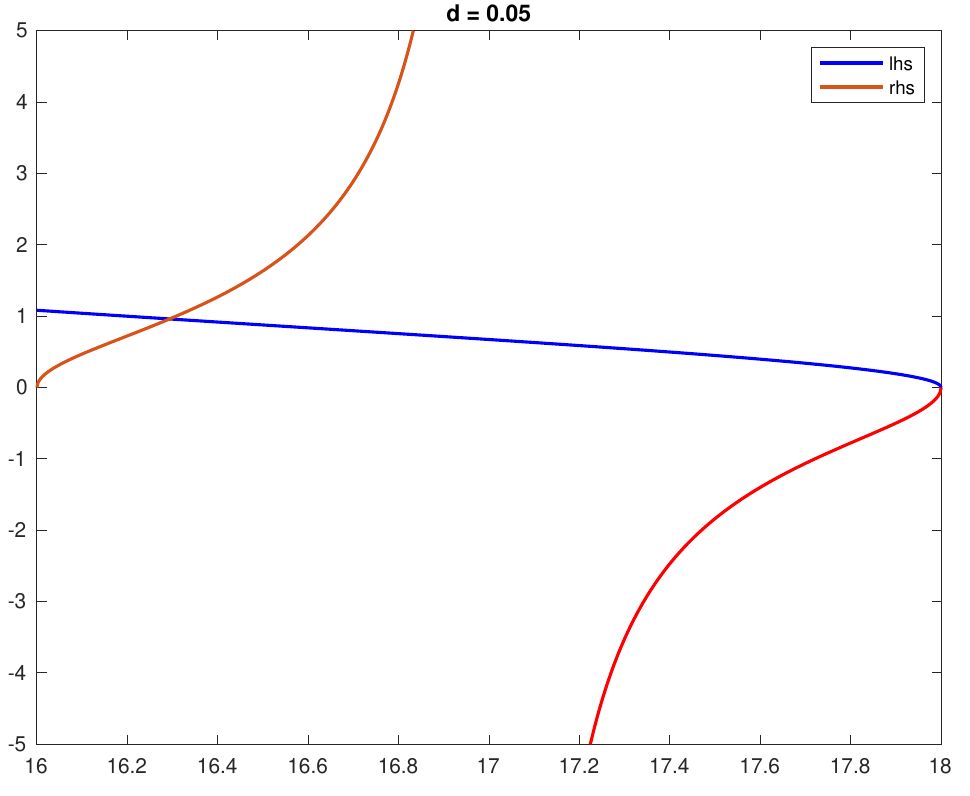}
    \caption{Plot showing left and right sides of~\eqref{eqn22} with $d=0.05.$}  
  \end{subfigure}\quad
  \begin{subfigure}[t]{.45\textwidth}
     \includegraphics[width=6cm]{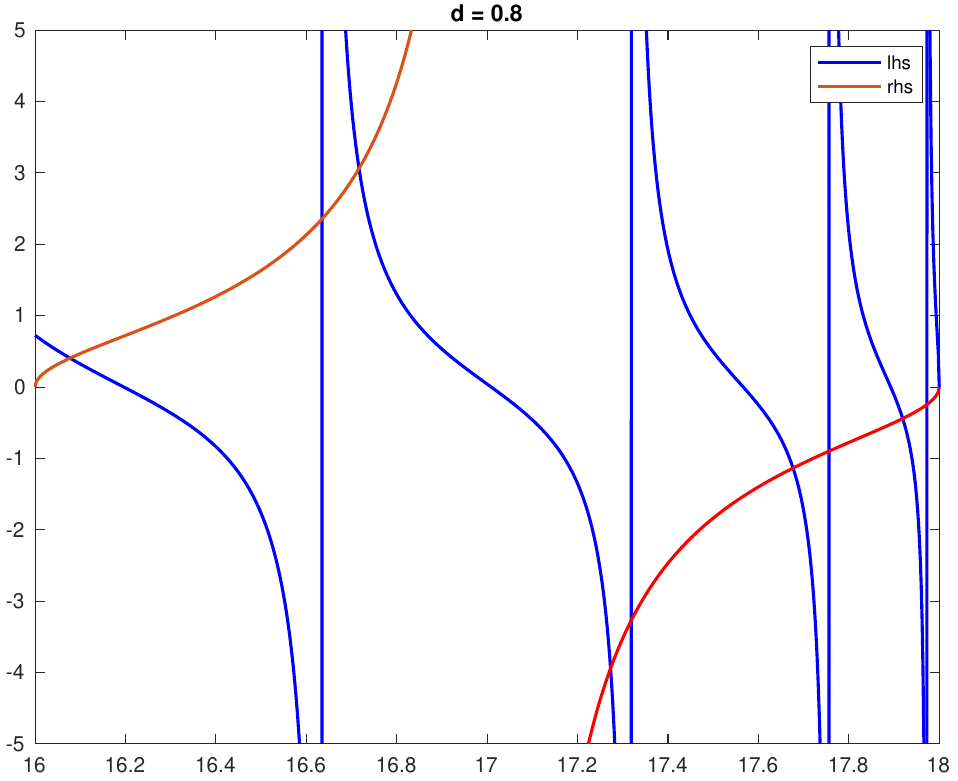}
\caption{Plot showing left and right sides of~\eqref{eqn22} with $d=0.8.$}
\end{subfigure}
    \caption{Zeros of the Wronskian with $k_1=16, k_2=18.$}  \label{fig1}
\end{figure}

\section{Proof of Theorem~\ref{thm00}}\label{AppWEsts}
In this section we give the details of the proof of Theorem~\ref{thm00}.  We
first need to give  a precise description of the limiting absorption resolvents
$\lim_{\delta\to 0^+}(\pa_{x_2}^2+k_1^2+q(x_2)-\xi^2+i\delta)^{-1},$ which are
constructed in Appendix~\ref{sec1}, see~\eqref{eqn18.02}.

\subsection{Estimates on the Wronskian and Basic Solutions of $L_{\xi}$}\label{sec4.1.205}
We  spell out the properties of the basic solutions,
$\tu_{\pm}(\xi,0+;x_2),$ and their Wronskian, $W(\xi),$ in greater
detail.  For these calculations we assume that $q$ is given
in~\eqref{eqn45.32}. Similar estimates hold for any piecewise
continuous, bounded potential, which is supported in an interval,
though we can no longer rely on explicit formul{\ae} for $\tu_{\pm}$
within the support of $q,$ or for the Wronskian.

The Wronskian is given by
\begin{equation}\label{eqn87.55}
  W(\xi)=
  \begin{cases}
    &-e^{2i Bd}\left[(2\xi^2-k_1^2-k_2^2)\frac{\sin 2d A}{A}-2iB\cos
      2dA\right]\text{ for }|\xi|<k_1,\\
    &-e^{-2\tB d}\left[(2\xi^2-k_1^2-k_2^2)\frac{\sin 2d A}{A}+2\tB\cos
      2dA\right]\text{ for }k_1<|\xi|<k_2,\\
    &-e^{-2 \tB d}\left[(2\xi^2-k_1^2-k_2^2)\frac{\sinh 2d \tA}{\tA}+2\tB\cosh
      2d\tA\right]\text{ for }k_2<|\xi|.
  \end{cases}
\end{equation}
Throughout this appendix
\begin{equation}
  \begin{split}
    A=\sqrt{k_2^2-\xi^2},\quad& B=\sqrt{k_1^2-\xi^2},\\
        \tA=\sqrt{\xi^2-k_2^2},\quad& \tB=\sqrt{\xi^2-k_1^2}.
  \end{split}
\end{equation}
These formul{\ae} and those for $\tu_+$ below are provided for the convenience of the reader, as all can be
derived from the formula for $|\xi|<k_1$ by choosing the correct branch of the
square root:  $\sqrt{z}$ is defined in $\bbC\setminus (-\infty,0],$
  with $\sqrt{x}\in (0,\infty)$ for $x\in (0,\infty).$

The Wronskian is analytic in a neighborhood of $\pm k_2,$ (where $A,\tA$
vanish), but has square root singularities at $\pm k_1.$ It never vanishes at
$\pm k_2,$ and does not vanish at $\pm k_1,$ provided that
\begin{equation}
  2d\sqrt{k_2^2-k_1^2}\neq n\pi\text{ for any }n\in\bbZ.
\end{equation}
As $|\xi|\to\infty,$ we have the lower bound
\begin{equation}
  |W(\xi)|>M|\xi|.
\end{equation}

Since $q$ is supported in $[-d,d]$ formul{\ae} for $\tw(\xi,x_2;y)$ as $|x_2|\to\infty$
are given by
\begin{equation}\label{eqn92.55}
  \begin{split}
    \tw(\xi,x_2;y)&=-\frac{\tu_+(\xi,0+;x_2) e^{-iy_1\xi}}{2W(\xi)\sqrt{\xi^2-k_1^2}}\times
    \\&\int_{-d}^{d}\tu_+(\xi,0+;-z_2)q(z_2)
    e^{-|z_2-y_2|\sqrt{\xi^2-k_1^2}}dz_2,\text{ for }x_2>d,\\
     \tw(\xi,x_2;y)&=-\frac{\tu_+(\xi,0+;-x_2)e^{-iy_1\xi}}{2W(\xi)\sqrt{\xi^2-k_1^2}}\times
     \\&\int_{-d}^{d}\tu_+(\xi,0+;z_2)q(z_2)
    e^{-|z_2-y_2|\sqrt{\xi^2-k_1^2}}dz_2,\text{ for }x_2<-d,
  \end{split}
\end{equation}
which shows that we need formul{\ae} for $\tu_+(\xi,0+;z_2)$ for all $z_2\geq
-d.$ We also include the formula for $z_2<-d,$ with $|\xi|<k_1:$ 
\begin{equation}
   \tu_+(\xi,0+;z_2)=
    \begin{cases}
      &e^{iBz_2},\text{ for }z_2>d,\\
      &e^{iBd}\left[\cos A(d-z_2)-iB\frac{\sin
          A(d-z_2)}{A}\right]\text{ for }|z_2|<d,\\
      &e^{iBd}\Big[\left(\cos 2dA \cos B(d+z_2)+A\sin 2dA\frac{\sin
          B(d+z_2)}{B}\right)
        +\\ &iB\left(\cos 2dA\frac{\sin B(d+z_2)}{B}-\frac{\sin 2d
          A}{A}\cos B(d+z_2)\right)\Big]\text{ for }z_2<-d.
    \end{cases}
\end{equation}
Note that for all $z_2<d,$
$$\tu_+(\xi,0+;z_2)=e^{i\sqrt{k_1^2-\xi^2}d}\left[\theta(\xi,z_2)+
  \sqrt{k_1^2-\xi^2}\varphi(\xi,z_2)\right],$$
where $\theta$ and $\varphi$ are entire functions $\xi.$ For the
convenience of the reader we include certain formul{\ae} for $k_1<|\xi|<k_2,$
\begin{equation}
   \tu_+(\xi,0+;z_2)=
    \begin{cases}
      &e^{-\tB z_2},\text{ for }z_2>d\\
      &e^{-\tB d}\left[\cos A(d-z_2)+\tB\frac{\sin A(d-z_2)}{A}\right]\text{ for }|z_2|<d;
    \end{cases}
\end{equation}
and for $k_2<|\xi|,$
\begin{equation}
   \tu_+(\xi,0+;z_2)=
    \begin{cases}
      &e^{-\tB z_2},\text{ for }z_2>d\\
      &e^{-\tB d}\left[\cosh \tA(d-z_2)+\tB\frac{\sinh \tA(d-z_2)}{\tA}\right]\text{ for }|z_2|<d.
    \end{cases}
\end{equation}

If $\{\xi_n\}$ are the positive roots of the Wronskian, then
\begin{equation}
  \gamma^{-}_n=\{-\xi_n+\nu e^{i\theta}:\:\theta\in [0,2\pi)\}\text{ and }
     \gamma^{+}_n=\{\xi_n+\nu e^{i\theta}:\:\theta\in [0,2\pi)\}
\end{equation}
are circles of radius $\nu$ centered in these roots.  We let $-\gamma^{\pm}_{n+}$
denote the intersection of these circles with the upper half plane, oriented from left to
right. Before proceeding with this analysis, we observe that
$\tw(\xi,x_2;0,y_2)=\tw(-\xi,x_2;0,y_2),$ which implies that, in the integral defining
$\fw^{[1]},$ the contributions from $\Gamma_{\nu}^+\cap\bbR$ cancel exactly, and
therefore
\begin{equation}
  \fw^{[1]}=\frac{1}{2\pi}\sum_{n=1}^N\left[\int_{-\gamma^{-}_{n+}}\xi\tw(\xi,x_2;0,y_2)d\xi+
    \int_{-\gamma^{+}_{n+}}\xi\tw(\xi,x_2;0,y_2)d\xi\right].
\end{equation}
Using the symmetries of the integrand it is not difficult to show that
\begin{equation}
  \fw^{[1]}=\frac{1}{2\pi}\sum_{n=1}^N\int_{-\gamma^{+}_n}\xi\tw(\xi,x_2;0,y_2)d\xi;  
\end{equation}
where $-\gamma^+_n$ indicates that the circles about the $\{\xi_n\}$ are oriented in
the clockwise direction. This gives $-i$ times the sum of the residues of
$\xi\tw(\xi,x_2;0,y_2)$ at  the $\{\xi_n\}.$ It is not difficult to see that
\begin{equation}
  i\fw^{[1]}+\pa_{x_1}w_{0+}^g(0,x_2;0,y_2)=0,
\end{equation}
as the function $\pa_{x_1}w_{0+}^g$ is $i$ times the sum of the residues of
$i\xi\tw(\xi,x_2;0,y_2)$ at  the $\{\xi_n\}.$  This proves~\eqref{eqn51.300}.

While the kernel $\fw^{[1]}$ does not play any role in the integral equations
derived in Section~\ref{sec8}, we nonetheless include estimates for its behavior
as we need these results for the representation formula where $x_1\neq 0,$ which
involves $\pa_{y_1}w_{0+}(x;0,y_2).$

\subsection{Asymptotics for $|x_2|,|y_2|>d$}\label{sec6.1}
For $x_2>d,$ we write these functions in the form
\begin{multline}\label{eqn207.667}
    \fw^{[j]}(x_2,y_2)=\frac{1}{2\pi}\int_{\Gamma_{\nu}^+}\xi^j\tw(\xi,x_2;0,y_2)d\xi\\ =\frac{1}{2\pi}\int_{\Gamma_{\nu}^+}\xi^j\frac{\tu_+(\xi,0+;x_2)\fA(\xi,\sqrt{\xi^2-k_1^2};y_2)}{\fW(\xi,\sqrt{k_1^2-\xi^2})\sqrt{k_1^2-\xi^2}}d\xi,
\end{multline}
where, for $|\xi|<k_1,$ $\sqrt{\xi^2-k_1^2}=-i\sqrt{k_1^2-\xi^2}$ and we write
\begin{equation}
  \fA(\xi,\upsilon;y_2)=\int_{-d}^{d}q(z_2)e^{-\upsilon|y_2-z_2|}
  e^{-\upsilon d}\left[\cos A(d+z_2)+\upsilon\frac{\sin A(d+z_2)}{A}\right]dz_2.
\end{equation}
We rewrite the Wronskian as $W(\xi)=\fW(\xi,\sqrt{k_1^2-\xi^2}),$ with
\begin{equation}
  \fW(\xi,\upsilon)=-e^{2i \upsilon d}\left[(2\xi^2-k_1^2-k_2^2)\frac{\sin 2d A}{A}-2i\upsilon\cos
      2dA\right]\text{ for }|\xi|<k_1.
\end{equation}
These formul{\ae} hold for other ranges of $\xi$ by using the correct
branch of $\sqrt{k_1^2-\xi^2};$ as noted these formul{\ae} do not
depend on the choice of branch for $A=\sqrt{k_2^2-\xi^2}.$ We use this
formulation, as $\fA(\xi,\upsilon;y_2), $ and $\fW(\xi,\upsilon)$ are entire
functions of $(\xi,\upsilon),$ which allows us to better keep track of the
square root singularities at $\pm k_1$ in our subsequent computations.

If $d<x_2$ and $d<y_2,$ then the formula in~\eqref{eqn207.667} simplifies to
\begin{multline}\label{eqn98.201}
  \fw^{[j]}(x_2,y_2)=\frac{1}{2\pi}\int_{\Gamma_{\nu}^+}\xi^j
  \frac{e^{i\sqrt{k_1^2-\xi^2}(x_2+y_2)}\fA_0(\xi,-i\sqrt{k_1^2-\xi^2})}
  {\fW(\xi,\sqrt{k_1^2-\xi^2})\sqrt{k_1^2-\xi^2}}d\xi,
\end{multline}
where
\begin{equation}
  \fA_0(\xi,\upsilon)=\int_{-d}^{d}q(z_2)e^{\upsilon z_2}
  e^{-\upsilon d}\left[\cos A(d+z_2)+\upsilon\frac{\sin A(d+z_2)}{A}\right]dz_2.
\end{equation}
From this formula it is evident that the $ \{\fw^{[j]}(x_2,y_2)\}$ are functions only of
$x_2+y_2.$ With other choices of signs, $x_2+y_2$ in this formula is replaced
with $|x_2|+|y_2|,$ verifying the claim in the theorem.

To analyze $\fw^{[j]}$ we split the integral over $\xi$ into the segment where
$|\xi|<k_1+2\epsilon,$ for an $\epsilon<\min\{\nu/2,k_1/2\}$, a segment with
$k_1+\epsilon<|\xi|<k_2,$ and finally the segment where $k_2<|\xi|.$ The leading
order behavior is determined, via stationary phase, by the integral over
a small neighborhood of $\xi=0;$ whereas the diagonal singularities come from
large $|\xi|.$

We begin with $\xi\in [-(k_1+2\epsilon),k_1+2\epsilon],$ and
denote these contributions by $\fw^{[j]}_0(x_2,y_2).$ With $x_2,y_2>d$ we have
\begin{multline}\label{eqn98.20}
  \fw^{[j]}_{0}(x_2,y_2)
  =\frac{i}{2\pi}\int_{-(k_1+2\epsilon)}^{k_1+2\epsilon}\xi^j\frac{e^{i\sqrt{k_1^2-\xi^2}(x_2+y_2)}\fA_0(\xi,-i\sqrt{k_1^2-\xi^2})
  \varphi(\xi)}
  {\fW(\xi,\sqrt{k_1^2-\xi^2})\sqrt{k_1^2-\xi^2}}d\xi.
\end{multline}
Here $\varphi\in\cC^{\infty}_{c}((-(k_1+2\epsilon),k_1+2\epsilon)),$ with
$\varphi(t)=1,$ for $|t|<k_1+\epsilon.$

To treat this integral we use analyticity to deform the portions of the contour
near to $\pm k_1:$ we replace an interval $[-\epsilon-k_1,\epsilon-k_1]$ with a
smooth contour meeting the real axis smoothly, lying in the upper half plane and
an interval $[-\epsilon+k_1,\epsilon+k_1]$ with a smooth contour meeting the
real axis smoothly, lying in the lower half plane. We call this contour
$\Gamma^+_{\nu,\epsilon}.$ An example is shown in Figure~\ref{fig22}. On this
contour the analytic continuation of $\sqrt{k_1^2-\xi^2}$ is smooth, has a
positive imaginary part. Hence the integrand is smooth, compactly supported and
has a single non-degenerate critical point at $\xi=0.$ Using Cauchy's theorem to
deform the contour in the integral, and a standard stationary phase argument we
see that there are complete asymptotic expansions
\begin{multline}\label{eqn58.55}
 \fw^{[j]}_{0}(x_2,y_2)= \frac{i}{2\pi}\int_{\Gamma_{\nu,\epsilon}^+}\xi^j\frac{e^{i\sqrt{k_1^2-\xi^2}(x_2+y_2)}\fA_0(\xi,-i\sqrt{k_1^2-\xi^2})\varphi(\xi)}
       {\fW(\xi,\sqrt{k_1^2-\xi^2})\sqrt{k_1^2-\xi^2}}d\xi\\ \sim
      \frac{e^{ik_1(x_2+y_2)}}{(x_2+y_2)^{\frac{j+1}{2}}}\left[M_{j0}+\sum_{l=1}^{\infty}\frac{M_{jl}}{(x_2+y_2)^l}\right]
      \text{  for }j=0,1,2.
\end{multline}

A similar argument applies in all cases considered in the subsequent
sections. To treat the right half plane we use the contour
$\Gamma^-_{\nu,\epsilon},$ which is obtained by modifying $\Gamma^-_{\nu}$ is
the same way, i.e.~the change to the contour in the left half plane lies in
$\Im\xi>0,$ and that in the right half plane lies in $\Im\xi <0.$

\begin{figure}
  \centering
    \includegraphics[width= 10cm]{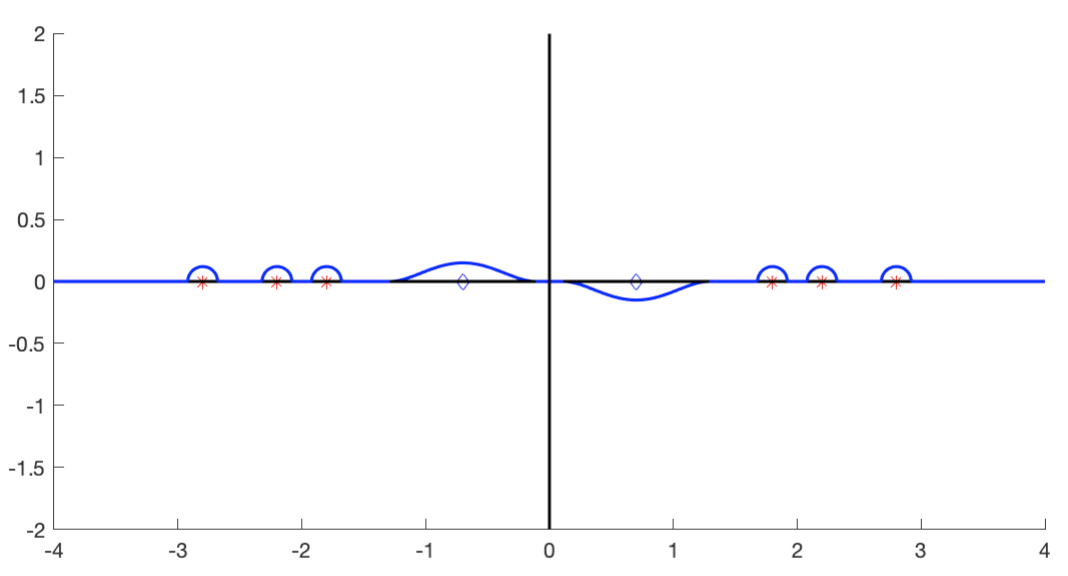}
    \caption{The blue contour is $\Gamma^+_{\nu,\epsilon}$ showing the smooth curves
      replacing intervals
      $[-\epsilon-k_1,\epsilon-k_1]\cup[-\epsilon+k_1,\epsilon+k_1].$ The roots
      of the Wronskian are shown as asterisks, and $\pm k_1$ as diamonds.}  
   \label{fig22}
\end{figure}

  Now we need to estimate the contributions from the remainder of
  $\Gamma_{\nu,\epsilon}^+$ lying over $[k_1+\epsilon,k_2]\cup
  [-k_2,-(k_1+\epsilon)].$ Other than the semi-circles centered on the
  roots of the Wronskian,
  \begin{equation}
    \bigcup_{n=1}^N\{\pm \xi_n+\nu e^{i\theta}:\:\theta\in
           [\pi,0]\}=\bigcup_{n=1}^N
           C^{\pm}_{n,\nu},
  \end{equation}
  we use the evenness of the integrand to restrict our attention to
  $[k_1+\epsilon,k_2].$   
  
  To estimate the contributions from the semi-circles we need to
  estimate both $|\fA_0(\xi,\sqrt{\xi^2-k_1^2})|,$ from above, and
  $|\fW(\xi,i\sqrt{\xi^2-k_1^2})|,$ from below, on these
  arcs. As the roots of the Wronskian are all simple, there
  is a constant $M$ so that
  \begin{equation}
    |\fW(\xi,i\sqrt{\xi^2-k_1^2})|\geq M\nu\text{ for }\xi \in C^{\pm}_{n,\nu}.
  \end{equation}
  On the other hand, for $\xi\in C^{\pm}_{n,\nu},$ 
  \begin{equation}
  |\fA_0(\xi,\upsilon)|\leq \int_{-d}^{d}q(z_2)e^{\alpha(z_2-d)}
  \left[\exp\beta(d-z_2)+|\upsilon|\frac{\exp\beta(d-z_2)}{|\sqrt{k_2^2-\xi^2}|}\right]dz_2,
  \end{equation}
  where
  \begin{equation}
    \alpha=\Re\left[\sqrt{(\pm\xi_n+\nu e^{i\theta})^2-k_1^2}\right]\text{ and }
    \beta=\left|\Im\left[\sqrt{k_2^2-(\pm\xi_n+\nu e^{i\theta})^2}\right]\right|.
  \end{equation}
  There is a constant, $m$ so that, for $\theta\in [0,2\pi],$ we have
  \begin{equation}
    \alpha>\sqrt{\xi_n^2-k_1^2}(1-m\nu)\text{ and }
    \beta<m\nu.
  \end{equation}
  These inequalities rely on the fact that $k_1<\xi_j<k_2.$ If we choose $0<\nu$
  sufficiently small so that
  \begin{equation}
    m\nu-\sqrt{\xi_n^2-k_1^2}(1-m\nu)<0\text{ for }n=1,\dots,N,
  \end{equation}
  then it follows that there exists a constant, $M,$ so that
  \begin{equation}
    |\fA_0(\xi,\sqrt{\xi^2-k_1^2})|\leq M\text{ for }\xi\in C^{\pm}_{n,\nu}\text{
      for }n=1,\dots, N.
  \end{equation}
  For sufficiently small $\nu$ we therefore have the estimates
  \begin{multline}
    \left|\sum_{n=1}^N\int_{C_{n,\nu}^{\pm}}\xi^j\frac{e^{-\sqrt{\xi^2-k_1^2}(x_2+y_2)}\fA_0(\xi,\sqrt{\xi^2-k_1^2})}
              {\fW(\xi,i\sqrt{\xi^2-k_1^2})\sqrt{\xi^2-k_1^2}}d\xi\right|\leq\\
              M\sum_{n=1}^N\frac{e^{-\sqrt{\xi_n^2-k_1^2}(1-m\nu)(x_2+y_2)}}{\nu}.
  \end{multline}
  To treat $\Gamma_{\nu}^+\cap [k_1+\epsilon,k_2],$ we see that $|\fA_0(\xi,\upsilon)|$ is
  bounded by a constant on this set and therefore
  \begin{multline}
    \left|\,\int\limits_{\Gamma_{\nu}^+\cap
      [k_1+\epsilon,k_2]}\xi^j\frac{e^{-\sqrt{\xi^2-k_1^2}(x_2+y_2)}\fA_0(\xi,\sqrt{\xi^2-k_1^2})(1-\varphi(\xi))}
                       {\fW(\xi,i\sqrt{\xi^2-k_1^2})\sqrt{\xi^2-k_1^2}}d\xi\right|\leq\\ M\frac{e^{-\mu
                           k_1(x_2+y_2)/2}}{\nu}.
  \end{multline}

  To summarize, we have shown that for $x_2,y_2>d,$ the contributions to the
  integrals defining $\fw^{[j]}(x_2,y_2),$ from frequencies $\xi\in [-k_2,k_2]$
  are given by the asymptotic expansions in~\eqref{eqn58.55}, with errors
  $O((x_2+y_2)^{-N}),$ for any $N>0.$ We conclude the discussion of this case by
  estimating the contribution from $|\xi|>k_2,$ which we denote by
  $\fw^{[j]}_2(x_2,y_2).$ As before, by evenness, it suffices to consider
  $[k_2,\infty).$ Using the formul{\ae} above we see that
    \begin{equation}
      \fw_2^{[j]}(x_2,y_2)=
      \int_{k_2}^{\infty}\xi^j\frac{e^{-\sqrt{\xi^2-k_1^2}(x_2+y_2)}\fA_0(\xi,\sqrt{\xi^2-k_1^2})}
                     {\fW(\xi,i\sqrt{\xi^2-k_1^2})\sqrt{\xi^2-k_1^2}}d\xi,
    \end{equation}
    where
    \begin{multline}
      \fA_0(\xi,\sqrt{\xi^2-k_1^2})=\int_{-d}^dq(z_2)e^{\sqrt{\xi^2-k_1^2}(z_2-d)}\times
      \\\left[\cosh\sqrt{\xi^2-k_2^2}(d+z_2)+\sqrt{\xi^2-k_1^2}
        \frac{\sinh\sqrt{\xi^2-k_2^2}(d+z_2)}{\sqrt{\xi^2-k_2^2}}\right]dz_2.
    \end{multline}
    From this formula we can easily show that there is a constant $M$ so that
    \begin{equation}
      |\fA_0(\xi,\sqrt{\xi^2-k_1^2})|\leq
      \frac{Me^{2d\sqrt{\xi^2-k_1^2}}}{1+|\xi|}\text{ for }\xi\in
           [k_2,\infty),
    \end{equation}
    and therefore, for  $x_2,y_2>d,\, |\xi|>k_2$ we have the estimate
    \begin{equation}\label{eqn73.9}
      |\tw(\xi,x_2;0,y_2)|\leq M\frac{e^{-\sqrt{\xi^2-k_1^2}(x_2+y_2-2d)}}{(1+|\xi|)^{3}}.
    \end{equation}
    Hence
    \begin{equation}
      \begin{split}
        | \fw_2^{[j]}(x_2,y_2)|&\leq 
      M\int_{k_2}^{\infty}\frac{\xi^{j}e^{-\sqrt{\xi^2-k_1^2}(x_2+y_2-2d)}}{(1+|\xi|)^3}d\xi\\
      &\leq\begin{cases} &Me^{-\sqrt{k_2^2-k_1^2}(x_2+y_2-2d)}\text{ for }j=0,1\\
      &Me^{-\sqrt{k_2^2-k_1^2}(x_2+y_2-2d)}\left|\log\frac{2(x_2+y_2)}{x_2+y_2-2d}\right|\text{
        for }j=2.
      \end{cases}
      \end{split}
    \end{equation}
  The calculation above indicates that the diagonal singularity should be
  $|x-y|^2\log|x-y|,$ and this implies that we should expect a $\log$-singularity
  in $\fw^{[2]}$ where $x_2+y_2=2d.$

    Altogether we have shown that, for $x_2,y_2>d$ we have the asymptotic
    formul{\ae}
    \begin{equation}\label{eqn581.55}
 \fw^{[j]}(x_2,y_2)\sim \\
      \frac{e^{ik_1(x_2+y_2)}}{(x_2+y_2)^{\frac{j+1}{2}}}\left[M_{j0}+\sum_{l=1}^{\infty}\frac{M_{jl}}{(x_2+y_2)^l}\right]
       \text{  for }j=0,1,2. 
    \end{equation}
     It is not difficult to see that the same arguments apply if
    $\pm x_2,\pm y_2>d$
    to show that the same expansions hold with $(x_2,y_2)$ replaced with
    $(|x_2|,|y_2|),$ for example:
    \begin{equation}\label{eqn80.56}
      \fw^{[j]}(x_2,y_2)\sim
      \frac{e^{ik_1(|x_2|+|y_2|)}}{(|x_2|+|y_2|)^{\frac{j+1}{2}}}\left[M^{\pm,\pm}_{j0}+\sum_{l=1}^{\infty}
        \frac{M^{\pm,\pm}_{jl}}{(|x_2|+|y_2|)^l}\right],
       \text{  for }j=0,1,2. 
    \end{equation}
    The function $\fw^{[0]}(x_2,y_2)$ is continuous as
$(x_2,y_2)\to (\pm d,\pm d),$ and
   \begin{equation}\label{eqn230.667}
     |\fw^{[2]}(x_2,y_2)|\leq M|\log(|x_2+y_2|-2d)|,\text{ as }|x_2+y_2|\to 2d.
   \end{equation}
   There is no singularity as $(x_2,y_2)\to (\pm d ,\mp d).$ 
    
    \Rd In sets where $|x_2|, |y_2|>d,$ the kernel is an infinitely
    differentiable function of $|x_2|+|y_2|.$ As the integrand is exponentially
    decaying, provided $|x_2|+|y_2|>2d,$ the formul{\ae} in~\eqref{eqn98.201}
    can be differentiated under the integral sign. Applying
    $\pa_{x_2}^l\pa_{y_2}^m$ introduces a factor of
    $[i\sqrt{k_1^2-\xi^2}]^{l+m}$ into the numerator of the integrand. The
    argument above applies {\em mutatis mutandis} to show that these functions
    also have asymptotic expansions, and by the theorem of Coddington and
    Levinson they must be obtained by differentiating the expansion
    in~\eqref{eqn581.55}. Provided $|x_2|+|y_2|>2d+\eta,$ for an
    $\eta>0,$ it is not difficult to see that the contributions from $|\xi|>k_2$
    are uniformly exponentially decaying. \Bk

  From our observations about the derivatives of these kernels we conclude
     that  for $j=0,1,2,$ with $x_2>d, y_2>d$ we have the expansions.
    \begin{equation}\label{eqn79.55}
      \begin{split}
        \pa_{x_2}\fw^{[j]}(x_2,y_2)-ik_1\fw^{[j]}(x_2,y_2)&\sim \frac{e^{ik_1(x_2+y_2)}}{(x_2+y_2)^{\frac{j+3}{2}}}\left[M'_{j0}+\sum_{l=1}^{\infty}\frac{M'_{jl}}{(x_2+y_2)^l}\right],
      \\
       \pa_{y_2}\fw^{[j]}(x_2,y_2)-ik_1\fw^{[j]}(x_2,y_2)&\sim \frac{e^{ik_1(x_2+y_2)}}{(x_2+y_2)^{\frac{j+3}{2}}}\left[M''_{j0}+\sum_{l=1}^{\infty}
         \frac{M''_{jl}}{(x_2+y_2)^l}\right].
        \end{split}
    \end{equation}
  Similarly, we can show that the functions
   \begin{multline}
             \pa_{x_2}\fw^{[j]}(x_2,y_2)\mp  ik_1\fw^{[j]}(x_2,y_2),
        \text{  for }\pm x_2>d,\text{ and }\\
          \pa_{y_2}\fw^{[j]}(x_2,y_2)\mp
          ik_1\fw^{[j]}(x_2,y_2)\text{ for }\pm y_2>d,
        \end{multline}
have asymptotic expansions like those in~\eqref{eqn79.55}, with $(x_2,y_2)$
replaced by $(|x_2|,|y_2|),$ obtained by applying the appropriate operator,
$\pa_{x_2}\mp ik_1,$ or $\pa_{y_2}\mp ik_1,$ to the expansions
in~\eqref{eqn80.56}.

\subsection{Asymptotics for $|x_2|$ or $|y_2|<d$}\label{sec6.2}
    We now consider what happens if either $|x_2|<d,$ or $|y_2|<d.$ We start by
    assuming that $y_2>d,$ but $|x_2|<d.$ For this case the integrals defining
    $\fw^{[j]}$ take the form
    \begin{equation}\label{eqn233.667}
      \fw^{[j]}(x_2,y_2)=
      \frac{i}{2\pi}\int_{\Gamma_{\nu,\epsilon}^+}\xi^j\frac{e^{i\sqrt{k_1^2-\xi^2}(y_2+2d)}
        \fB(\xi,\sqrt{k_1^2-\xi^2};x_2)}
         {\fW(\xi,\sqrt{k_1^2-\xi^2})\sqrt{k_1^2-\xi^2}}d\xi,
    \end{equation}
    where   
    \begin{multline}\label{eqn86.86}
      \fB(\xi,\upsilon;x_2)=\int\limits_{-d}^{x_2}\Bigg[\cos A(d-x_2)\cos
        A(d+z_2)-i\frac{\upsilon}{A}\sin A(2d+z_2-x_2)-\\
        \frac{\upsilon^2}{A^2}\sin A(d-x_2)\sin
        A(d+z_2)\Bigg] e^{-iz_2\upsilon}q(z_2)dz_2+\\
      \int\limits_{x_2}^{d}\Bigg[\cos A(d+x_2)\cos
        A(d-z_2)-i\frac{\upsilon}{A}\sin A(2d-z_2+x_2)-\\
        \frac{\upsilon^2}{A^2}\sin A(d+x_2)\sin
        A(d-z_2)\Bigg]e^{-iz_2\upsilon}q(z_2)dz_2.      
    \end{multline}
 Again $\fB(\xi,\upsilon;x_2)$ is an entire
 function of $(\xi,\upsilon).$ In the integral over $\Gamma_{\nu,\epsilon}^+$ we
 take $\upsilon=\sqrt{k_1^2-\xi^2},$ which equals $i\sqrt{\xi^2-k_1^2}$ for $k_1<|\xi|.$

 As before the principal contribution to $\fw^{[j]}(x_2,y_2)$, as
 $y_2\to\infty,$ comes from the stationary phase at $\xi=0.$ The
 function $|\fB(\xi,\sqrt{k_1^2-\xi^2};x_2)|$ is easily seen to be
 uniformly bounded where $|\xi|<k_1, |x_2|<d.$
 We separate the contributions from the stationary point at zero, and
 the contributions from endpoints $\pm k_1.$ The leading contributions from
 the stationary point are
 \begin{equation}\label{eqn132.25}
   \fw^{[j]}_{0}(x_2,y_2)=C'_jk_1^{\frac{j-1}{2}}\frac{e^{ik_1y_2}\fB(0,k_1;x_2)}
      {\fW(0,k_1)\,y_2^{\frac{j+1}{2}}}+O\left(|y_2|^{-{\frac{j+3}{2}}}\right)\text{
       for }j=0,1,2.
 \end{equation}
 Indeed, there are complete asymptotic expansions of the form
  \begin{equation}\label{eqn132.251}
   \fw^{[j]}_{0}(x_2,y_2)=k_1^{\frac{j-1}{2}}\frac{e^{ik_1y_2}}
      {\fW(0,k_1)\,y_2^{\frac{j+1}{2}}}\left[\sum_{l=0}^{\infty}\frac{b_{jl}(x_2)}{y_2^l}\right],
       \text{ for }j=0,1,2.
  \end{equation}
 Using the deformed contour $\Gamma^+_{\nu,\epsilon},$ we again see that the
 remainder of the interval $[-k_1-2\epsilon,k_1+2\epsilon]$ contributes
 $O(y_2^{-N})$ for any $N>0.$

 We next consider the integral over $\Gamma_{\nu,\epsilon}^+$ lying above
 $k_1+\epsilon<|\xi|<k_2,$ which is of the form
 \begin{multline}
    \fw^{[j]}_{1}(x_2,y_2)=\\ \frac{1}{2\pi}\int\limits_{\{k_1+\epsilon<|\xi|<k_2\}\cap
      \Gamma_{\nu,\epsilon}^+}\xi^j\frac{e^{-\sqrt{\xi^2-k_1^2}(y_2+2d)}
      \fB(\xi,i\sqrt{\xi^2-k_1^2};x_2)\varphi(\xi)}
       {\fW(\xi,i\sqrt{\xi^2-k_1^2})\sqrt{\xi^2-k_1^2}}d\xi,
 \end{multline}
 for a suitable cut-off function $\varphi.$ As before we use the lower bound on
 the Wronskian, $|\fW(\xi,i\sqrt{\xi^2-k_1^2})|>m\nu.$ We also need to bound
 $\fB(\xi,i\sqrt{\xi^2-k_1^2};x_2)$ from above. For
 $|x_2|<d,$ we can easily show that there is a constant $M$ so that
   \begin{equation}
     |\fB(\xi,i\sqrt{\xi^2-k_1^2};x_2)|\leq Me^{\rho(\xi) d},\text{
       where }\rho(\xi)=\Re\left(\sqrt{\xi^2-k_1^2}\right).
   \end{equation}
   It then follows that
   \begin{equation}
     | \fw^{[j]}_{1}(x_2,y_2)|\leq \frac{M}{\nu}
     \int\limits_{k_1+\epsilon<|\xi|<k_2}e^{-\rho(\xi)(y_2+d)}d\xi
   \end{equation}
   If we set
   \begin{equation}
   \alpha=\min\{\rho(\xi):\:
   \xi\in\Gamma_{\nu}^+\text{ with }k_1+\epsilon<|\xi|<k_2 \}, 
   \end{equation}
 then $\alpha>0,$ provided that $\nu$ is sufficiently small,
   from which it follows easily that
   \begin{equation}
     |\fw^{[j]}_{1}(x_2,y_2)|\leq \frac{Me^{-\alpha y_2}}{\nu}.
   \end{equation}
   
   This leaves the integral over the set $\{|\xi|>k_2\}.$ In this set
   we can prove the estimate
   \begin{equation}
     \left|\fB(\xi,i\sqrt{\xi^2-k_1^2};x_2)\right|\leq
     Me^{(2d+x_2)\sqrt{\xi^2-k_1^2}}\left[\frac{1}{1+|\xi|}+(d-x_2)\right],
   \end{equation}
   and therefore, if either $|x_2|<d$ or $|y_2|<d,$ and $|\xi|>k_2,$ we have the estimate
   \begin{equation}\label{eqn92.9}
     |\tw(\xi,x_2;0,y_2)|<M\frac{e^{-\sqrt{\xi^2-k_2^2}|x_2-y_2|}}{(1+|\xi|)^2}\left[\frac{1}{(1+|\xi|)}+|y_2-x_2|\right],
   \end{equation}
   and therefore for $|x_2|<d<y_2,$ we have
   \begin{equation}
     \begin{split}
       |\fw^{[j]}_2(x_2,y_2)|&=
      \left|\frac{1}{2\pi}\int_{k_2<|\xi|}\xi^j\frac{e^{\sqrt{\xi^2-k_1^2}(y_2+2d)}
        \fB(\xi,i\sqrt{\xi^2-k_1^2};x_2)}
                  {\fW(\xi,i\sqrt{\xi^2-k_1^2})\sqrt{\xi^2-k_1^2}}d\xi\right|\\
     &\leq
                  M\int_{k_2}^{\infty}e^{-\sqrt{\xi^2-k_1^2}(y_2-x_2)}
                      \left[\frac{1}{(1+|\xi|)^{3-j}}+\frac{y_2-x_2}{(1+|\xi|)^{2-j}}\right]d\xi.
     \end{split}
   \end{equation}
   An elementary calculation now shows that
   \begin{equation}
     |\fw^{[j]}_2(x_2,y_2)|\leq \begin{cases}
       &Me^{-\sqrt{k_2^2-k_1^2}(y_2-x_2)}\text{ for }j=0,1\\
       &Me^{-\sqrt{k_2^2-k_1^2}(y_2-x_2)}\left|\log|y_2-x_2|\right|\text{ for }j=2.
       \end{cases}
   \end{equation}
   
   We can again show that this analysis applies equally well
   if $|x_2|<d,y_2<-d,$ and therefore we have the asymptotic expansions
   \begin{equation}\label{eqn246.667}
     \fw^{[j]}(x_2,y_2)\sim
  k_1^{\frac{j-1}{2}}\frac{e^{ik_1|y_2|}}
      {|y_2|^{\frac{j+1}{2}}}\left[\sum_{l=0}^{\infty}\frac{b^{\pm}_{jl}(x_2)}{|y_2|^l}\right],
       \text{ for }\pm y_2>d,\,j=0,1,2.
  \end{equation}
    The function $\fw^{[0]}(x_2,y_2)$ is continuously differentiable as
    $y_2\to x_2,$ and
   \begin{equation}
     |\fw^{[2]}(x_2,y_2)|\leq M|\log|y_2-x_2||,\text{ as }y_2\to x_2.
   \end{equation}
   Approaching from $|x_2|<d, |y_2|>d,$  this term has a logarithmic singularity at
   the points $[\pm d,\pm d].$
   
  \Rd As before we see that $\fw^{[j]}(x_2,y_2)$ are infinitely differentiable
  functions of $y_2,$ where $|x_2|\leq d<|y_2|,$ as the integral
  in~\eqref{eqn233.667} can differentiated as often as we please.  Applying
  $\pa_{y_2}^l$ introduces a factor of $[i\sqrt{k_1^2-\xi^2}]^l$ in the numerator
  of the integrand. The argument showing the existence of an asymptotic
  expansion applies equally well if $l>0,$ and once again the theorem of
  Coddington and Levinson shows that the expansions are obtained by
  differentiating~\eqref{eqn246.667}. The contributions from $|\xi|>k_2$ are
  uniformly exponentially decaying provided that $d+\eta<y_2$ for an
  $\eta>0.$ \Bk
  In particular, we can applying the operator $\pa_{y_2}\mp ik_1$ to the
  expansions above to see that
   \begin{equation}
     \pa_{y_2}\fw^{[j]}(x_2,y_2)\mp ik_1 \fw^{[j]}(x_2,y_2)=O\left(|y_2|^{-\frac{j+3}{2}}\right),
   \end{equation}
   and in fact these functions have complete asymptotic expansions.

   The final case we need to treat is $|y_2|<d, |x_2|>d.$ We begin with $x_2>d.$
   For this case we write
   \begin{equation}\label{eqn249.667}
     \fw^{[j]}(x_2,y_2)=
     \frac{1}{2\pi}\int_{\Gamma_{\nu,\epsilon}^+}\frac{\xi^j
       e^{i\sqrt{k_1^2-\xi^2}x_2}\fC(\xi,\sqrt{k_1^2-\xi^2};y_2)}
          {\fW(\xi,\sqrt{k_1^2-\xi^2})
       \sqrt{k_1^2-\xi^2}}d\xi,
   \end{equation}
   where
   \begin{multline}
     \fC(\xi,\upsilon;y_2)=\\
     e^{i\upsilon(d+y_2)}\int_{-d}^{y_2}e^{-i\upsilon z_2}\left[\cos
       A(d+z_2)-i\upsilon\frac{\sin A(d+z_2)}{A}\right]q(z_2)dz_2+\\
     e^{i\upsilon(d-y_2)}\int_{y_2}^{d}e^{i\upsilon z_2}\left[\cos
       A(d+z_2)-i\upsilon\frac{\sin A(d+z_2)}{A}\right]q(z_2)dz_2,
   \end{multline}
  with $A=\sqrt{k_2^2-\xi^2};$ is an entire function of $(\xi,\upsilon).$

   As in the earlier cases, there is a stationary phase contribution
   from $\xi=0,$ which we denote by $\fw^{[j]}_{0}(x_2,y_2).$ The function
   $\fC(\xi,\sqrt{k_1^2-\xi^2};y_2)$ is an analytic function of $\xi$ in
   $|\xi|<k_1-\epsilon,$ for any $\epsilon>0.$ The stationary phase
   contributions are
   \begin{equation}\label{eqn148.25}
     \fw^{[j]}_{0}(x_2,y_2)\sim
     \frac{e^{ik_1x_2}}{\fW(0,k_1)\,
       x_2^{\frac{j+1}{2}}}\left[\sum_{l=0}^{\infty}
       \frac{c_{jl}(y_2)}{x_2^l}\right]      \text{ for }j=0,1,2.
             \end{equation}
 The analysis for the contribution from $\delta<|\xi|<k_1+\epsilon$  proceeds as before,
 and can be shown to be $O(x_2^{-N}),$ for any $N>0.$

      The remaining portion for $k_1+\epsilon<|\xi|<k_2$ is given by
      \begin{multline}
      \fw^{[j]}_{1}(x_2,y_2)=\\
     \frac{1}{i\pi}\int\limits_{\{k_1+\epsilon<|\xi|<k_2\}\cap
       \Gamma_{\nu,\epsilon}^+}
     \frac{e^{-\sqrt{\xi^2-k_1^2}x_2}\,\xi^j
       \fC(\xi,i\sqrt{\xi^2-k_1^2};y_2)(1-\varphi(\xi))}
          {\fW(\xi,i\sqrt{\xi^2-k_1^2})\sqrt{\xi^2-k_1^2}}d\xi,  
      \end{multline}
      where
      \begin{multline}
        \fC(\xi,i\sqrt{\xi^2-k_1^2};y_2)=
        e^{-\sqrt{\xi^2-k_1^2}(y_2+d)}\times \\
        \int_{-d}^{y_2}e^{\sqrt{\xi^2-k_1^2}
          z_2}\left[\cos A(d+z_2)+\sqrt{\xi^2-k_1^2}\frac{\sin
            A(d+z_2)}{A}\right]q(z_2)dz_2+\\ e^{-\sqrt{\xi^2-k_1^2}(d-y_2)}\times \\
        \int_{y_2}^{d}e^{-\sqrt{\xi^2-k_1^2}z_2}\left[\cos
          A(d+z_2)+\sqrt{\xi^2-k_1^2}\frac{\sin
            A(d+z_2)}{A}\right]q(z_2)dz_2\\
        \text{ with }A=\sqrt{k_2^2-\xi^2}.
      \end{multline}
      It is easy to see that for $|y_2|<d,$ 
      $|\fC(\xi,i\sqrt{\xi^2-k_1^2};y_2)|$ is bounded, and  the Wronskian satisfies
      $|\fW(\xi,i\sqrt{\xi^2-k_1^2})|>M\nu,$ on the part of
      $\Gamma_{\nu}^+$ lying over $k_1+\epsilon<|\xi|<k_2.$ Moreover
      $\Re\sqrt{\xi^2-k_1^2}>\alpha>0,$  therefore
      \begin{equation}
        |\fw^{[j]}_{1}(x_2,y_2)|\leq
        \frac{Me^{-\alpha x_2}}{\nu}.
      \end{equation}

      Where $|\xi|>k_2$ we have
     \begin{multline}\label{eqn155.22}
        \fC(\xi,i\sqrt{\xi^2-k_1^2};y_2)=
        e^{-\sqrt{\xi^2-k_1^2}(y_2+d)}\times \\
        \int_{-d}^{y_2}e^{\sqrt{\xi^2-k_1^2}
          z_2}\left[\cosh A(d+z_2)+\sqrt{\xi^2-k_1^2}\frac{\sinh
            A(d+z_2)}{A}\right]q(z_2)dz_2+\\ e^{-\sqrt{\xi^2-k_1^2}(d-y_2)}\times \\
        \int_{y_2}^{d}e^{-\sqrt{\xi^2-k_1^2}z_2}\left[\cosh
          A(d+z_2)+\sqrt{\xi^2-k_1^2}\frac{\sinh
            A(d+z_2)}{A}\right]q(z_2)dz_2\\ \text{ with }A=\sqrt{\xi^2-k_2^2}.
     \end{multline}
     The first term in~\eqref{eqn155.22} is bounded by
     \begin{equation}
       M\frac{e^{\sqrt{\xi^2-k_2^2}y_2}}{1+|\xi|};
     \end{equation}
     the other term requires somewhat more care. It is bounded by
     \begin{multline}
       Me^{\sqrt{\xi^2-k_1^2}y_2}\int_{y_2}^{d}e^{(\sqrt{\xi^2-k_2^2}-\sqrt{\xi^2-k_1^2})z_2}dz_2=\\
       Me^{\sqrt{\xi^2-k_1^2}y_2}\left[\frac{e^{(\sqrt{\xi^2-k_2^2}-\sqrt{\xi^2-k_1^2})d}-
           e^{(\sqrt{\xi^2-k_2^2}-\sqrt{\xi^2-k_1^2})y_2}}
         {\sqrt{\xi^2-k_2^2}-\sqrt{\xi^2-k_1^2}}\right].
     \end{multline}
     An elementary argument then shows that the expression in the bracket is
     bounded by $|d-y_2|$ and therefore
     \begin{equation}
       | \fC(\xi,i\sqrt{\xi^2-k_1^2};y_2)|\leq
       Me^{\sqrt{\xi^2-k_1^2}y_2}\left[\frac{1}{1+|\xi|}+|d-y_2|\right].
     \end{equation}
     These estimates show that
     \begin{equation}
       \begin{split}
       |\fw^{[j]}_2(x_2,y_2)|&=\frac{1}{\pi}
       \left|\int_{k_2<|\xi|}\frac{e^{-\sqrt{\xi^2-k_1^2}x_2}\,\xi^j
       \fC(\xi,i\sqrt{\xi^2-k_1^2};y_2)}
                 {\fW(\xi,i\sqrt{\xi^2-k_1^2})\sqrt{\xi^2-k_1^2}}d\xi\right|\\
                 &\leq
                 M\int_{k_2}^{\infty}e^{-\sqrt{\xi^2-k_1^2}(x_2-y_2)}\xi^{j-2}\left[\frac{1}{\xi}+(x_2-y_2)\right]d\xi.
      \end{split}           
     \end{equation}
     In the second line we use the fact that $x_2-y_2>d-y_2.$ It is easy to show
     that the second term is bounded by $Me^{-\sqrt{k_2^2-k_1^2}(x_2-y_2)}$ for $j=0,2.$
     We estimate the contribution from the first term, as before, to obtain that
     for $|y_2|<d, x_2>d,$ we have
     \begin{equation}
      |\fw^{[j]}_2(x_2,y_2)|\leq 
      \begin{cases}
        &Me^{-\sqrt{k_2^2-k_1^2}(x_2-y_2)}\text{ for }j=0,1\\
        &Me^{-\sqrt{k_2^2-k_1^2}(x_2-y_2)}\cdot\left|\log|x_2-y_2|\right|\text{ for }j=2.
      \end{cases}
     \end{equation}

     It is not difficult to show that the analogous estimates hold with
     $x_2<-d,$ so that
   \begin{equation}
      |\fw^{[j]}_2(x_2,y_2)|\leq 
      \begin{cases}
        &Me^{-\sqrt{k_2^2-k_1^2}|x_2-y_2|}\text{ for }j=0,1\\
        &Me^{-\sqrt{k_2^2-k_1^2}|x_2-y_2|}\cdot\left|\log|x_2-y_2|\right|\text{ for }j=2.
      \end{cases}
   \end{equation}
   Altogether we have, for $|y_2|<d,|x_2|>d,$
    \begin{equation}\label{eqn112.55}
      \fw^{[j]}(x_2,y_2)\sim
       \frac{e^{ik_1|x_2|}}{ |x_2|^{\frac{j+1}{2}}}\left[\sum_{l=0}^{\infty}
       \frac{c^{\pm}_{jl}(y_2)}{|x_2|^l}\right],    \text{ for }\pm x_2>d,\, j=0,1,2.
    \end{equation}
        The function $\fw^{[0]}(x_2,y_2)$ is continuously differentiable as
    $|x_2-y_2|\to 0,$ and 
    \begin{equation}
      |\fw^{[2]}(x_2,y_2)|\leq M|\log|x_2-y_2||.
    \end{equation}
Approaching from $|x_2|>d, |y_2|<d,$ this term has a logarithmic singularity at the points $[\pm d,\pm d].$
    
   \Rd As before we see the kernels $\fw^{[j]}(x_2,y_2)$ are infinity
   differentiable as functions of $x_2,$ where $|x_2|>d,$ as the formula
   in~\eqref{eqn249.667} can be differentiated as often as we please. The
   argument showing that the derivatives have asymptotic expansions proceeds as
   in the previous case. \Bk We can therefore show that
   \begin{equation}
     \pa_{x_2}\fw^{[j]}(x_2,y_2)\mp ik_1
     \fw^{[j]}(x_2,y_2)=O\left(|x_2|^{-\frac{j+3}{2}}\right)\text{ as }\pm x_2 \to\infty,
   \end{equation}
   have asymptotic expansions obtained by applying $\pa_{x_2}\mp ik_1$ to the
   expansions in~\eqref{eqn112.55}.

   \subsection{The Diagonal Singularity}
   The singularities of the function $w(x;y)$ are confined to the set $x=y$
   where $|x_2|,|y_2|\leq d.$ In this section we use the Fourier representation
   to study the nature of this singularity.  As expected from the discussion
   surrounding~\eqref{eqn87.24}, the principal singularity behaves like
   $(x_2-y_2)^2\log|x_2-y_2|.$ Note that the discussion leading
   to~\eqref{eqn87.24} assumed that the potential $q(x_2)$ is smooth, whereas
   here we continue working with $q$ as defined in~\eqref{eqn45.32}.

   We start by assuming that $-d<y_2<x_2<d,$ so that
   \begin{equation}\label{eqn166.24}
     \fw^{[j]}(x_2,y_2)=\int_{\Gamma_{\nu}^+}\frac{\xi^j\fD(\xi,\sqrt{\xi^2-k_1^2};x_2,y_2)}
        {\fW(\xi,\sqrt{k_1^2-\xi^2})\sqrt{k_1^2-\xi^2}}d\xi,
   \end{equation}
   where
   \begin{multline}\label{eqn266.667}
  \fD(\xi,\upsilon;x_2,y_2)=
  \int_{-d}^{x_2}\tu_+(x_2,0+;\xi)\tu_+(-z_2,0+;\xi)e^{-\upsilon|z_2-y_2|}q(z_2)dz_2+\\
  \int_{x_2}^{d}\tu_+(-x_2,0+;\xi)\tu_+(z_2,0+;\xi)e^{-\upsilon|z_2-y_2|}q(z_2)dz_2\\
  =\int_{-d}^{y_2}e^{-2\upsilon d}\left[\cosh A(d-x_2)+\upsilon\frac{\sinh
      A(d-x_2)}{A}\right]\times \\\left[\cosh A(d+z_2)+\upsilon\frac{\sinh
      A(d+z_2)}{A}\right]
  e^{-\upsilon(y_2-z_2)}q(z_2)dz_2+\\
  \int_{y_2}^{x_2}e^{-2\upsilon d}\left[\cosh A(d-x_2)+\upsilon\frac{\sinh
      A(d-x_2)}{A}\right]\times \\\left[\cosh A(d+z_2)+\upsilon\frac{\sinh
      A(d+z_2)}{A}\right]
  e^{-\upsilon(z_2-y_2)}q(z_2)dz_2+\\
   \int_{x_2}^{d}e^{-2\upsilon d}\left[\cosh A(d+x_2)+\upsilon\frac{\sinh
      A(d+x_2)}{A}\right]\times \\\left[\cosh A(d-z_2)+\upsilon\frac{\sinh
      A(d-z_2)}{A}\right]
  e^{-\upsilon(z_2-y_2)}q(z_2)dz_2,
   \end{multline}
   with $A=\sqrt{\xi^2-k_2^2}.$ As before, $\fD(\xi,\upsilon;x_2,y_2)$ is an entire
   function of $(\xi,\upsilon);$ we use $\upsilon=-i\sqrt{k_1^2-\xi^2},$ for
   $|\xi|<k_1,$ and $\upsilon=\sqrt{\xi^2-k_1^2},$ for $|\xi|>k_1.$ In the second,
   more explicit expression, we use the conditions that $-d<y_2<x_2<d,$ and
   $|\xi|>k_2.$

   The parts of the integrals in~\eqref{eqn166.24} where
   $|\xi|<k_2+2,$ define  $\cC^2$-functions of $(x_2,y_2),$ so we focus our
   attention on $|\xi|\geq k_2+1.$ In this set we have
   \begin{equation}\label{eqn168.24}
     \fw^{[j]}_2(x_2,y_2)=-i\int_{|\xi|>k_2+1}\frac{\xi^j\fD(\xi,\sqrt{\xi^2-k_1^2};x_2,y_2)\varphi(\xi)}
        {\fW(\xi,i\sqrt{\xi^2-k_1^2})\sqrt{\xi^2-k_1^2}}d\xi.
   \end{equation}
   Here $\varphi\in\cC^{\infty}(\bbR)$ is an even non-negative function with
   $\supp\varphi\subset (-\infty,-(k_2+1))\cup (k_2+1,\infty)$ with $\varphi(t)=1$ if $|t|>k_2+2.$

   Using the definition of $\fD$ we see that,  with $B=\sqrt{\xi^2-k_1^2},$
   \begin{multline}
     |\fD(\xi,B;x_2,y_2)|\leq
     M\Bigg[e^{-(Ax_2+By_2)}\int_{-d}^{y_2}e^{(A+B)z_2}dz_2+\\
         e^{-(Ax_2-By_2)} \int_{y_2}^{x_2}  e^{(A-B)z_2}dz_2+
           e^{Ax_2+By_2} \int_{x_2}^{d}  e^{-(A+B)z_2}dz_2\Bigg].
   \end{multline}
   Performing these integrals and using elementary estimates we see that
   \begin{equation}
      |\fD(\xi,B;x_2,y_2)|\leq Me^{-\sqrt{\xi^2-k_2^2}|x_2-y_2|}\left[\frac{1}{1+|\xi|}+|x_2-y_2|\right].
   \end{equation}
   and therefore~\eqref{eqn92.9} holds if both $|x_2|<d,$ and $|y_2|<d.$
      Inserting this into~\eqref{eqn168.24},  we see that
   \begin{equation}\label{eqn171.24}
     \begin{split}
     |\fw^{[j]}_2(x_2,y_2)|&\leq
     M\int_{k_2+1}^{\infty}\xi^{j-2}e^{-\sqrt{\xi^2-k_2^2}|x_2-y_2|}\left[\frac{1}{1+|\xi|}+
       |x_2-y_2|\right]d\xi\\
     &\leq\begin{cases} &M[1+|x_2-y_2|^2|\log|x_2-y_2||]\text{ for }j=0,\\
     &M[1+|x_2-y_2||\log|x_2-y_2||]\text{ for }j=1,\\
     &M[1+|\log|x_2-y_2||]\text{ for }j=2.
     \end{cases}
     \end{split}
   \end{equation}

   It is straightforward to see that $|\fD(\xi,B;x_2,y_2)|$ satisfies the same
   estimate if $-d<x_2<y_2<d,$ and therefore the estimates in~\eqref{eqn171.24}
   hold in this case as well. It is easy to see that differentiating
   $\fw^{[j]}(x_2,y_2)$ with respect to $x_2$ or $y_2$ has the effect of
   increasing $j$ by 1. Thus we see that $\pa_{x_2}\fw^{[0]}(x_2,y_2),
   \pa_{y_2}\fw^{[0]}(x_2,y_2),$ are uniformly bounded, and
   $\pa_{x_2}\fw^{[1]}(x_2,y_2), \pa_{y_2}\fw^{[1]}(x_2,y_2),$ along with the second
   derivatives of $\fw^{[0]}(x_2,y_2)$ have a $\log|x_2-y_2|$-singularity along the
   diagonal.

   \Rd Exercising somewhat more care in evaluating the terms
   in~\eqref{eqn166.24}, which appear in~\eqref{eqn266.667}, one can show that
   these kernels have expansions near the diagonal within $B_d$ of the form
   \begin{multline}\label{eqn271.678}
     \fw^{[j]}_2(x_2,y_2)=|x_2-y_2|^{2-j}\log|x_2-y_2|\,\psi_0^{[j]}(|x_2-y_2|)+\psi_1^{[j]}(|x_2-y_2|)+\\
     |x_2+y_2-2d|^{2-j}\log|x_2+y_2-2d|\,\chi_0^{[j]}(x_2+y_2-2d)+\\
     |x_2+y_2+2d|^{2-j}\log|x_2+y_2+2d|\,\chi_1^{[j]}(x_2+y_2+2d)+\chi_2^{[j]}(x_2+y_2).
   \end{multline}
   Here $\psi_0^{[j]},\psi_1^{[j]},\chi_0^{[j]},\chi_1^{[j]},\chi_1^{[j]}$ are
   smooth functions defined near $x=0.$ The existence of the additional
   singularities in the corners $(\pm d,\pm d)$ is suggested
   by~\eqref{eqn230.667}, though these singularities do not extend beyond $B_d.$
   From this expansion it is clear that these kernels with $j=0,1,2$ define
   smoothing operators. For $j=0,1$ this is obvious; for $j=2$ classical
   estimates in potential theory show that $\fw^{[2]}$ maps data in
   $\cC_{\alpha}(\bbR),$ for $0<\alpha<\frac 12,$ into H\"older continuous
   functions.  \Bk


\end{document}